\documentclass[12pt]{article}
    \usepackage{amsmath}
    \usepackage{psfrag,epsf}
    \usepackage{enumerate}
    \usepackage[square,numbers]{natbib}
    \usepackage{url} 
    
    \usepackage{amsmath,amssymb,amsthm,bm,subcaption,bigstrut,multirow,enumitem,bbm}
    \usepackage{booktabs,mathtools,thmtools,thm-restate,ifpdf,appendix}
    
    \usepackage[noindentafter]{titlesec}
    
    \titlespacing\section{0pt}{12pt plus 2pt minus 2pt}{12pt plus 2pt minus 2pt}
    \titlespacing\subsection{0pt}{12pt plus 2pt minus 2pt}{12pt plus 2pt minus 2pt}
    \titlespacing\subsubsection{0pt}{12pt plus 2pt minus 2pt}{12pt plus 2pt minus 2pt}
    
    \ifpdf
    \usepackage[pdftex]{graphicx}
    \else
    \usepackage{graphicx}
    \fi
    
    \usepackage[final]{changes}
    
    \newcommand{\ignore}[1]{}
    \newcommand{\distiid}{\overset{iid}{\sim}}
    \newcommand{\asconv}{\xrightarrow{a.s.}}

    \newcommand\independent{\protect\mathpalette{\protect\independenT}{\perp}}
    \def\independenT#1#2{\mathrel{\rlap{$#1#2$}\mkern2mu{#1#2}}}
        
    \newtheorem{thm}{Theorem}
    \newtheorem{lem}{Lemma}
    \newtheorem{cor}{Corollary}
    \newtheorem{prop}{Proposition}
    
    \author{\textbf{Edward McFowland III}\\
       Technology Operations and Management\\
       Harvard Business School\\
       Boston, MA 02163, USA \\\\
        \textbf{Sriram Somanchi} \\
       IT, Analytics, and Operations\\
       University of Notre Dame\\\\
        \textbf{Daniel B. Neill}\\
       Machine Learning for Good Laboratory\\
       New York University
       }

 \date{}   
\title{Efficient Discovery of Heterogeneous Quantile Treatment Effects in Randomized Experiments via Anomalous Pattern Detection}

\begin{document}
\maketitle
\begin{abstract}
  In the recent literature
  on estimating heterogeneous treatment effects, each proposed method makes its own set of
  restrictive assumptions about the intervention's effects and which subpopulations to explicitly estimate.
  Moreover, the majority of the literature provides no mechanism to identify which
  subpopulations are the most affected--beyond manual inspection--and provides
  little guarantee on the correctness of the identified subpopulations. Therefore,
  we propose Treatment Effect Subset Scan (TESS), a new method for discovering
  which subpopulation in a randomized experiment is most significantly affected by
  a treatment. We frame this challenge as a
  pattern detection problem where we efficiently maximize a nonparametric scan statistic (a measure of the conditional quantile treatment effect)
  over subpopulations. Furthermore,
  we identify the subpopulation which experiences the largest distributional
  change as a result of the intervention, while making minimal assumptions about
  the intervention's effects or the underlying data generating process. In addition
  to the algorithm, we demonstrate that under the sharp null hypothesis of no treatment effect, the asymptotic Type I and II error can
  be controlled, and provide sufficient conditions for detection
  consistency--i.e., exact identification of the affected subpopulation. Finally, we
  validate the efficacy of the method by discovering heterogeneous treatment
  effects in simulations and in real-world data from a well-known program
  evaluation study.
\end{abstract}

\vfill

\newpage
\section{Introduction}
\label{sec:Intro}
The randomized experiment is employed across many empirical
disciplines as an
important tool for discovery, by estimating the causal impact of a
particular stimulus, treatment or intervention. Moreover, the increasing popularity of large-scale experiments~\citep{kohavi-online_experiments-2013} has resulted in a
widespread interest in discovering fine-grained truths about experimental
units, most prominently in the form of heterogeneous treatment effects (HTE). Discovering heterogeneity can be challenging because there are
exponentially many subpopulations--with respect to the number of
observable covariates--to consider, potentially resulting in multiple hypothesis testing issues and raising 
questions of unprincipled post-hoc investigation: searching for a fortuitously 
statistically significant result \citep{assmann-subgroup-2000, 
weisberg-subgroup-2015}. Nevertheless, uncovering affected 
subpopulations can lead to important scientific progress. In a ``step toward a new 
frontier of personalized medicine''~\citep{saul-bidil-2005}, the FDA approved the 
first race-specific drug, whose impact on African-American subjects was first 
discovered post-hoc from more general experiments~ \citep{cohn-bidil-1986, 
cohn-bidil-1991}. Conversely, the Perry preschool experiment found significant effects of preschool education on educational and life outcomes
\citep{barnett-perry-1985,schweinhart-perry-1993,angrist-mhe-2008}, while a
re-analysis focused on heterogeneity and multiple hypothesis testing concluded that 
only girls experience these benefits~\citep{anderson-perry-2008}. The original Perry preschool results
were fundamental to the creation of the Head Start preschool program 
\citep{angrist-mhe-2008} a national social program that provides, among other 
services, early childhood education to low-income children. If large-scale medical 
and policy decisions are made as a result of such experiments, then it is clear that 
identifying whether there is heterogeneity in treatment effects should be an integral
component of the analysis.

In this work we propose a novel computationally efficient framework--Treatment Effect Subset Scanning
(TESS)--for \emph{discovering} which subpopulations in a randomized experiment are the most significantly affected by a treatment. The contributions of this work can be summarized as follows:
\begin{itemize}
    \item Our TESS algorithm enables efficient discovery of subpopulations where the individuals
    affected by the treatment have observed outcome distributions that are unexpected given the distributions of their
    corresponding control groups. \ignore{ Unlike the standard approaches, TESS frames
    the challenge of discovering causal effects in subpopulations as one
    of~\emph{anomalous pattern detection}, and provides a computationally
    efficient approach for finding conditionally optimal subpopulations.
    Therefore, TESS is a novel contribution to the burgeoning literature on
    subset scanning, providing conditions under which the linear-time subset scanning property~\citep{neill-ltss-2012}
    can be exploited in the context of high-dimensional tensors, and thus
    extending the applicability of this efficient optimization approach beyond the standard low-dimensional
    context~\citep{neill-ltss-2012, mcfowland-fgss-2013,
    speakman-graphscan-2015}.}
    
    \item We formalize the objective of identifying subpopulations with significant distributional treatment effects by developing a new measure and test statistic for heterogeneous quantile treatment effects.

    \item We provide
    theoretical results on the detection properties of TESS. When the maximum subpopulation score identified by TESS is used
    as a test statistic under the sharp null hypothesis of no treatment effect, we
    demonstrate the conditions under which the Type I (Theorem
    \ref{thm:false_posotive}) and Type II (Theorem \ref{thm:power}) errors can
    jointly be controlled asymptotically. Furthermore, we provide sufficient conditions on
    how ``homogeneous'' (Theorem \ref{thm:subset_homo}) and ``strong'' (Theorem
    \ref{thm:subset_strength}) the treatment effect must be across the affected
    subpopulation, such that the TESS test statistic is maximized at the precisely correct
    subpopulation. Finally, we show that asymptotically these conditions are met (Theorems \ref{thm:assympt_hom_str} and \ref{thm:assympt_TESS}), guaranteeing that, in the large-sample limit, TESS will recover the precisely correct subpopulation.
    
    \item In the process of developing theory for TESS, we prove results
    for the general nonparametric scan statistic (NPSS), which has been used in
    the scan statistics literature~\citep{mcfowland-fgss-2013,
    feng-npss_graph-2014}. We are the first to provide theoretical
    guarantees on the detection behavior of subset scanning
    algorithms. Furthermore, our theory is derived for the higher
    dimensional (tensor) context, with nonparametric score functions, and our
    results directly hold for the lower-dimensional and parametric cases as well.
    
    \item Our empirical results (\S\ref{sec:star-eda}) provide useful insights to practitioners, revealing a
    potentially affected subpopulation in the Tennessee STAR study of class size and educational outcomes, who may have benefited from an
    intervention (the use of a teacher's aide) that was generally considered ineffective.
\end{itemize}

These contributions are enabled by structuring the question of causal inference as one of \textit{anomalous pattern detection} and effect maximization, rather than model fitting and risk minimization. In some contexts, the standard approach of learning an overall
model of the treatment effect response surface is desirable; however, in
many cases, the identification of affected subpopulations is the primary goal and model learning is simply a step toward this
goal. For these cases it seems prudent and efficient to circumvent this
first step and solve the subpopulation identification problem by framing it
as one of pattern or subset discovery. Such a framing has not previously been considered in the literature.

\section{Heterogeneous Quantile Treatment Effects}
\label{sec:hetero_quant_effects}
Most contemporary causal methods are estimators for the (conditional) average treatment effects, or CATE, $\tau_{CATE}(x) = \mathbb{E}\left[Y(1)-Y(0) | X = x \right]$, which in turn limits empirical studies of treatment effects from considering effects beyond mean shifts~\citep{abadie-qte-2002}. However, social scientists argue that effects can greatly vary along the outcome distribution, and distributional impacts beyond the average effect are critical for policy-makers, across a wide range of social programs~\citep{firpo-qte-2007,abadie-qte-2002,chernozhukov-ivqte-2005,schiele-qte-2016}. The primary distributional alternative to ATEs has been Quantile Treatment Effects (QTE) and the subsequent conditional QTEs, or CQTE~\citep{firpo-qte-2007, koenker-QTE-1978, koenker-hqte-2010, chernozhukov-QTE-generalization-2013} at a given quantile $\alpha$:
\begin{equation}
  \label{eq:CQTE}
    \tau_{\text{CQTE}_{\alpha}}(x) = \mathbb{F}_{Y(1)|X=x}^{-1}(\alpha) - \mathbb{F}_{Y(0)|X=x}^{-1}(\alpha),
\end{equation}
where $\mathbb{F}_{Y(1)|X=x}^{-1}$ and $\mathbb{F}_{Y(0)|X=x}^{-1}$ denote the inverse cumulative distribution functions of the outcome $Y$, conditional on covariates $X=x$, under the counterfactual assignments to the treatment group ($W=1$) and control group ($W=0$) respectively. 

In this work, we consider the challenge of~\textit{heterogeneous} quantile treatment effects, i.e.,~\textbf{detecting} the existence of a subpopulation $S$ (characterized by a subset of values for each attribute) for which the CQTE is non-zero, at some quantile, even if there is not a significant effect in the overall population. This motivates the need for a measurement of the heterogeneous treatment effect $\tau_{\text{CQTE}_\alpha}(S)$ and a corresponding test statistic $F_\alpha(S)$ that can be optimized over both subpopulations $S$ and quantiles $\alpha$, capturing unknown heterogeneity in both covariates and the treatment effect distribution, respectively. While a simple extension to \eqref{eq:CQTE},
\begin{equation}
    \label{eq:CQTE(S)-simple}
    \max_S \tau_{\text{CQTE}}(S) = \max_S \max_\alpha \tau_{\text{CQTE}_\alpha}(S) = \max_S \max_\alpha \mathbb{F}_{Y(1)|X\in S}^{-1}(\alpha) - \mathbb{F}_{Y(0)|X\in S}^{-1}(\alpha),
\end{equation}
may appear to be an attractive alternative, this formulation is inadequate for detecting subpopulations with distributional effects. Note that in~\eqref{eq:CQTE(S)-simple} the effect is represented by a difference in scalar summaries of the potential outcome distributions, instead of capturing ~\textit{full} distributional effect~\citep{chernozhukov-QTE-generalization-2013,van-structural-2014}. Moreover, it~\textit{first} aggregates $\mathbb{F}_{Y(W)|X=x}~\forall x\in S$ to construct $\mathbb{F}_{Y(W)|X \in S}$ and then compares these aggregate conditional distributions, instead of~\textit{first} comparing each $\mathbb{F}_{Y(1)|X=x}$ to the corresponding $\mathbb{F}_{Y(0)|X = x} ~\forall x\in S$ and~\textit{then} aggregating.\footnote{When the effects of interest are simply differences in scalar summaries of each potential outcome distribution, the results are equivalent for any order of comparison and aggregation. However, for more general distributional effects, this equivalence does not hold.}  The effect of interest can easily be obfuscated when aggregating before comparing: consider that $Y(W)|X=x$ need not be on the same scale for different $x\in S$.  Additionally,~\eqref{eq:CQTE(S)-simple} tends to be maximized at extreme values $\alpha\approx 0$ and $\alpha\approx 1$ and is thus highly sensitive to outliers, losing power to detect QTEs occurring at non-extreme values of $\alpha$. Finally,~\eqref{eq:CQTE(S)-simple} fails to appropriately calibrate the treatment effect across potential subpopulations $S$ of varying sizes, and therefore in~\eqref{eq:CQTE(S)-simple} the optimal $S$ equates to $\max_{x,\alpha} \tau_{\text{CQTE}_{\alpha}}(x)$, i.e., a singular covariate profile (see proof in Appendix~\ref{sec:scoring_functions}). This last issue also arises when maximizing other popular conditional treatment effect estimands in the literature, such as CATE, over subpopulations. \ignore{$\mathbb{F}_{Y|X \in S} = \mathbb{E}\left[\mathbb{F}_{Y}|X = x, X\in S\right]$}

To avoid these limitations, we first recognize that our primary goals are to \emph{test} the null hypothesis
\begin{equation*}
\mathbb{F}_{Y(0)|X=x}^{-1}(\alpha) = \mathbb{F}_{Y(1)|X=x}^{-1}(\alpha),
\end{equation*}
for $\alpha \in (0,1)$ and $\forall x$, and to \emph{detect} subpopulations $S$ for which the two counterfactual outcome distributions differ significantly. An equivalent test is for
\begin{equation*}
\mathbb{F}_{Y(1)|X=x}(\mathbb{F}_{Y(0)|X=x}^{-1}(\alpha)) = \mathbb{F}_{Y(1)|X=x}(\mathbb{F}_{Y(1)|X=x}^{-1}(\alpha)) = \alpha.
\end{equation*}
Moreover, re-defining $\tau_{\text{CQTE}_{\alpha}}(x) = \mathbb{F}_{Y(1)|X=x}(\mathbb{F}_{Y(0)|X=x}^{-1}(\alpha))$ captures the \textit{full} distributional effect (as demonstrated by~\cite{chernozhukov-QTE-generalization-2013,van-structural-2014}) prior to aggregation, and is constrained to $(0,1)$, allowing coherent aggregation and calibration over $x\in S$.  More precisely, we can define $\tau_{\text{CQTE}_\alpha}(S) = \sum_{x \in S} \tau_{\text{CQTE}_\alpha}(x) P(X=x \mid X \in S)$, and then define a test statistic $F_\alpha(S)$ to measure the significance of the divergence between $\tau_{\text{CQTE}_\alpha}(S)$ and $\alpha$.

Therefore, we can make specific, simple, and testable assumptions about the relationships between each $\mathbb{F}_{Y(0)|X=x}$ and $\mathbb{F}_{Y(1)|X=x}$, under the null and alternative hypotheses, and construct a generalized likelihood ratio test that maximizes detection power for distinguishing these hypotheses:
\begin{flalign}
  \label{eq:test_ceqte}
  &\!\begin{aligned}
  H_0:~ &\tau_{\text{CQTE}_{\alpha}}(x) = \alpha \quad \forall x, \alpha &\\
  H_1\left(S\right):~ &\begin{cases}
  \exists~\beta,\alpha,~\text{with}~ \beta>\alpha,~ s.t.~ &\tau_{\text{CQTE}_{\alpha}}(x) = \beta \quad \forall x \in S,\\ 
  \forall \alpha~ &\tau_{\text{CQTE}_{\alpha}}(x) = \alpha \quad \forall x \not\in S. 
\end{cases}
  \end{aligned}&
\end{flalign}
We define $H_1(S)$ as in~\eqref{eq:test_ceqte} (with constant $\beta_x = \beta$) because 
of our interest in detecting subsets $S$ where $\mathbb{F}_{Y(1)|X=x}$ 
differs~\emph{systematically} from $\mathbb{F}_{Y(0)|X=x}$ for $x \in S$, thus grouping together covariate 
profiles that exhibit similar treatment effects, rather than massively overfitting
to individual covariate profiles.\footnote{This is analogous to tree-based 
methods which assume the same conditional average treatment effect (CATE) for all cells assigned to
a given leaf, but rather than looking for a mean shift, we measure how much of the probability 
density of $Y(1)$ has been ``shifted'' into the $\alpha$-tail of $Y(0)$.} Moreover, if
$\mathbb{F}_{Y(1)|X=x} \ne \mathbb{F}_{Y(0)|X=x}$, we know that there exists some $\alpha$ and some subset $S$ such that $\tau_{\text{CQTE}_\alpha}(S) = \beta \ne \alpha$, while if $\mathbb{F}_{Y(1)|X=x} = \mathbb{F}_{Y(0)|X=x}$, then $\beta = \alpha$ everywhere and there is no shift. Therefore, defining the alternative in this way allows us to prove desirable theoretical results both for detection power and for subset correctness, as described in Theorems \ref{thm:false_posotive}-\ref{thm:power} and \ref{thm:subset_homo}-\ref{thm:assympt_TESS} respectively.

As we show in Appendix~\ref{sec:scoring_functions}, the log-likelihood ratio statistic for~\eqref{eq:test_ceqte}, for a given sample, corresponds to the Berk-Jones nonparametric scan statistic~\citep{mcfowland-fgss-2013,berk-bj-1979}:
\begin{equation}
    \label{eq:BJ-teststat}
  \max_S F(S) = \max_{S,\alpha}F_{\alpha}^{BJ}(S) = \max_{S,\alpha}N(S)KL\left(  \hat{\tau}_{\text{CQTE}_{\alpha}}(S) , \alpha \right),
\end{equation}
where each $\hat{\tau}_{\text{CQTE}_{\alpha}}(x)$ is computed using its potential outcome empirical distribution, $N(S)$ is the number of treatment group units, and $KL(\beta, \alpha)=\beta \log \frac{\beta}{\alpha} + (1-\beta) \log \frac{1-\beta}{1-\alpha}$ is the Kullback-Leibler divergence between Bernoulli distributions with the corresponding parameters. The statistic in~\eqref{eq:BJ-teststat} includes a maximization over subpopulations $S$ and thresholds $\alpha$ to identify the most significantly affected (highest scoring) subpopulation, and its significance can then be determined by a randomization test, appropriately controlling for multiple testing.

\section{Treatment Effect Subset Scanning}
\label{sec:tess}
Treatment Effect Subset Scan (TESS) is a novel framework for identifying
subpopulations in a randomized experiment which experience treatment
effects, built atop the heterogeneous quantile treatment effect test statistic established in~\eqref{eq:BJ-teststat}.
Unlike previous methods, TESS structures the challenge of treatment effect
identification as an anomalous pattern detection problem--where the objective is
to identify patterns of systematic deviations away from expectation--which is
then solved by scanning over subpopulations. TESS therefore searches for subsets
of values of each attribute for which the distributions of outcomes in the
treatment groups are systematically anomalous, i.e., significantly different
from their expectation as derived from the control group. More precisely, we
define a real-valued outcome of interest $Y$ and a set of discrete covariates $X =
(X^{1}, \ldots, X^{d}$), where each $X^{j}$ can take on a vector of values
$V^{j}=\{v^{j}_{m}\}_{m = 1...|V^j|}$. We note that continuous covariates can be discretized into categories, using the observed covariate distribution or domain knowledge.\footnote{An extension could include considering the intervals of the continuous covariate created by each of its unique split points (realized values) in the data; this is similar to how tree-based methods determine discrete splits on continuous variables.} With continuous and discrete covariates, the distribution and quantile functions are well-defined and unique for all levels
$\alpha\in (0, 1)$ when the outcome $Y$ is real-valued. We define the arity of covariate $X^{j}$
as $|V^{j}|$ (i.e., the cardinality of $V^j$) and note that for any covariate profile $x$ (i.e., 
a realization of $X$) it follows that $x \in V^1 \times \ldots \times V^d$. We
then define a dataset as a sample
$\mathcal{N}$ composed of $n$ records (units) $\{R_{1}, \ldots, R_{n}\}$,
drawn independently and identically distributed from population $\mathcal{P}$. Each 3-tuple $R_{i} = (Y^{\text{obs}}_i, X_i, W_i)$ is described by an observed potential outcome $Y^{\text{obs}}_i = Y_i(W_i)$, covariates $X_i$, and an
indicator variable $W_i$, which indicates if the unit was randomly assigned to
the treatment condition; see Table~\ref{table:example_data} for a demonstrative
example. We define the subpopulations $S$ under consideration to be $S = v^1
\times \ldots \times v^d$, where $v^j \subseteq V^j$. Therefore, we consider subsets $S$ representing \emph{subspaces} of the attribute space, i.e., the Cartesian product of a subset of values for each attribute. This is important because the treatment of interest may affect multiple values, e.g., African-Americans \textit{or} Hispanics who live in New York \textit{or} Pennsylvania. Finally, we wish to find the most
anomalous subset
\begin{equation}
  \label{eqn:objective_function}
  S^\ast = v^{1*}\times \ldots \times v^{d*} = \arg \max_{S} F(S),
\end{equation}
where $F(S)$ is commonly referred to in the anomalous pattern detection literature
as a score function, to measure the anomalousness of a subset $S$. In the context of TESS, this function is a test statistic of the treatment effect--i.e., the
divergence between the treatment and control group--in subpopulation $S$, like the one defined in \eqref{eq:BJ-teststat}.

\begin{table}
  \captionsetup{font=scriptsize,skip=0pt}
  \centering
  \resizebox{0.33\columnwidth}{!}{
  \begin{tabular}{@{}ccccc@{}}
    \toprule
      Record & $Y$  & $X^{\text{gender}}$ & $X^{\text{race}}$ & $W$ \\ \midrule
      1      & 2.35 & Female              & Black             & 1   \\
      2      & 2.06 & Female              & White             & 1   \\
      3      & 2.92 & Male                & Black             & 1   \\
      4      & 2.27 & Male                & White             & 1   \\
      5      & 1.73 & Female              & Black             & 0   \\
      6      & 1.84 & Female              & White             & 0   \\
      7      & 1.7  & Male                & Black             & 0   \\
      8      & 1.59 & Male                & White             & 0   \\ \bottomrule
  \end{tabular}
}
  \caption{This table is a demonstrative dataset of $n=8$ records, with a $d=2$
  sized vector of covariates, $X = (X^{\text{gender}},X^{\text{race}})$. The
  first, $X^{\text{gender}}$, can take values in
  $V^{\text{gender}}=\{\text{Female, Male}\}$, and the second
  $X^{\text{race}}$ can take values in $V^{\text{race}}=\{\text{Black,
  White}\}$. A covariate profile $x$, and realization of $X$, is an element in
  the set of all covariate profiles $V^{\text{race}} \times
  V^{\text{gender}} = \{(\text{Female, Black}), (\text{Female,
  White}), (\text{Male, Black}), (\text{Male, White})\}$.}
  \label{table:example_data}
\end{table}

We accomplish this by first partitioning the experimental dataset into control
and treatment groups, and passing the groups to the TESS algorithm.
For each unique covariate profile $x$ in the treatment group, TESS uses the
control group to compute a conditional outcome distribution $\hat{\mathbb{F}}_{Y^C \mid X=x}$,
providing an estimate of the conditional outcome distribution under the
null hypothesis $H_{0}$ that the treatment has no effect on units with this
profile. Then for each record $R_{i}$ in the treatment group, TESS computes an empirical $p$-value $\hat{p}_{i}$, which
serves as a measure of how uncommon it is to see an outcome as extreme as
$Y^{\text{obs}}_i$ given $X=x_{i}$ under $H_0$. The ultimate goal of TESS is to discover
subpopulations $S$ with a large amount of evidence against $H_0$, i.e., the outcomes
of units in $S$ are consistently extreme given $H_0$. Thus,
TESS searches for subpopulations which contain an unexpectedly large number of
low (significant) empirical $p$-values, as such a subpopulation is more
likely to have been affected by the treatment.

\subsection{Estimating Reference Distributions}
\label{sec:ref-dist}
After partitioning the data into treatment and control groups, the TESS
framework obtains an estimate of the reference distribution for each unique covariate profile in the treatment
group. To obtain the estimates of $\mathbb{F}_{Y(0) \mid X}~\forall X$, TESS relies on two assumptions: randomization
and a sharp null hypothesis of no treatment effect. First, randomization implies that the potential outcomes $ Y_i(0), Y_i(1) \independent W_i~\forall R_i$: selection into treatment and control groups is completely random.
Secondly, the sharp null hypothesis that no
subpopulation is affected by the treatment implies that $\mathbb{F}_{Y(0) \mid X}
= \mathbb{F}_{Y(1) \mid X}$. With these two assumptions in hand, the TESS framework includes two options for estimating the necessary reference distributions. The first is more flexible but may encounter estimation challenges in extremely sparse, high-dimensional settings; the second is useful for higher-dimensional settings, but requires additional structural assumptions on the data generating process. Finally, given a chosen procedure for estimating reference distributions, TESS uses the distributions to convert each observed outcome $Y^{\text{obs}}_i$ (for data records $R_i$ in the treatment group) to an empirical $p$-value range, capturing how ``anomalous''
that outcome is given its reference distribution.

\subsubsection{Empirical Distribution Estimation}
\label{sec:sec:ref-dist-emp}
The first option we present for deriving reference distributions involves estimating the empirical conditional probability function as follows:
\begin{equation}
    \label{eq:control-edf}
    \hat{\mathbb{F}}_{Y^C|X}(y|x) = \frac{\sum_{i=1}^n{\zeta_i(x)\mathbbm{1}_{\{Y^{\text{obs}}_i \le y\}}}}{\sum_{i=1}^n \zeta_i(x)},
\end{equation}
representing a weighted average across data units, with a weight function defined as
\begin{equation}
    \label{eq:control-edf-weights}
    \zeta_i(x) = \mathbbm{1}_{\{W_i = 0, X_i = x\}}.
\end{equation}
For any record $R_j$ in the treatment group (i.e., with $W_j = 1$) we can use $\hat{\mathbb{F}}_{Y^C|X}(\cdot|x_j)$ as an estimate of its distribution function. When we use \eqref{eq:control-edf-weights} as the weight definition, then \eqref{eq:control-edf} amounts to the empirical density function derived from the control units that share covariate profile $X = x_j$.
Moreover, it follows directly from TESS's assumption of randomization and the Glivenko-Cantelli Theorem~\citep{gaenssler-glivenko_cantelli-2004} that $\hat{\mathbb{F}}_{Y^C \mid X} \xrightarrow{a.s.} \mathbb{F}_{Y(0) \mid X}$. Therefore, TESS can use $\hat{\mathbb{F}}_{Y^C \mid X}$ as
an unbiased and strongly-consistent estimator of the unknown $\mathbb{F}_{Y(1) \mid X}$ under
$H_{0}$. Intuitively, under this sharp null, the
outcomes of the treatment and control groups are drawn from the same
distribution, allowing $\hat{\mathbb{F}}_{Y^C \mid X = x}$ to serve as an outcome
reference distribution for treatment units with covariate
profile $X = x$. 

\subsubsection{Model-based Estimation}
\label{sec:model-based estimation}
Although we define and estimate \eqref{eq:control-edf} individually for
each unique covariate profile $X=x$ using the empirical distribution function, we note that
TESS only requires some means of computing the conditional probability of observing each treatment unit outcome. The empirical distribution allows TESS to
accommodate arbitrary differences in conditional outcome distributions
across covariate profiles, enabling general applicability without
a priori contextual knowledge. However, it is also possible to combine data across profiles to
estimate
the conditional probability distributions. This aggregation of information can help alleviate challenges that arise when there is data sparsity, i.e., when there are covariate profiles present in the treatment group that have few or no corresponding control data records. Intuitively, ``neighboring'' covariate profiles in the control group can be pooled and leveraged to improve local estimation. However, this improved estimation comes by imposing additional structure or assumptions on the underlying data generating process.

Statistical learning offers many options for distribution (or density) estimation, any of which can be utilized in TESS. We identify the Random Forest estimator that underpins the Quantile Regression Forests algorithm~\citep{meinshausen-quantileforests-2006} as an attractive alternative to the purely empirical estimator described above. Random Forest can be cast as an adaptive locally weighted estimator, where the forest places more weight on observations with more similar covariates. Therefore, TESS can learn a Random Forest on the control data, still using \eqref{eq:control-edf} as its reference distribution, but redefining its weights as:
\begin{equation}
    \label{eq:control-forest-weights}
    \zeta_i(x_j) = \frac{1}{B}\sum_{b=1}^B \mathbbm{1}_{\{W_i = 0, X_i \in L_{b}(x_j)\}},
\end{equation}
where $B$ corresponds to the number of trees in the forest, and $L_{b}(x_j)$ captures the leaf node--i.e., a subset of the covariate space--of tree $b$ that $x_j$ falls into. Therefore, \eqref{eq:control-forest-weights} can be seen as a relaxation of \eqref{eq:control-edf-weights}, where a control unit can have non-zero weight even if its profile does not match $x_j$ precisely. Also the weights are adaptive and more smoothly increase with how similar a control unit is to the treatment unit. Intuitively, this adaptive similarity ``kernel'' is particularly helpful in sparse and/or high-dimensional settings, where the curse of dimensionality makes estimation challenging, because it allows for local estimation within covariate subspaces of similar units. Importantly, this similarity is measured along the subset of dimensions that are discovered as relevant (via the random forest learning procedure), which manifests as how often the two points would appear in the same leaf node of the learned trees. Moreover, it has been shown that such a random forest based weighting scheme for estimation can alleviate the curse of dimensionality~\citep{athey-generalforest-2019}. It has also been shown that with weights as in \eqref{eq:control-forest-weights}, \eqref{eq:control-edf} is weakly-consistent, given a set of regularity conditions and the assumption that the true distribution function is Lipschitz continuous~\citep{meinshausen-quantileforests-2006}. Therefore, it follows directly from this property of consistency and TESS's assumption of randomization that $\hat{\mathbb{F}}_{Y^C \mid X} \xrightarrow{p} \mathbb{F}_{Y(0) \mid X}$. Therefore, TESS is able to use a random forest estimator of $\hat{\mathbb{F}}_{Y^C \mid X}$ as
a weakly-consistent estimator of the unknown $\mathbb{F}_{Y(1) \mid X}$ under the null hypothesis
$H_{0}$.

\subsection{Computing Empirical P-value Ranges}
\label{sec:p-values}
Given a mechanism for estimating the conditional probabilities of outcomes, TESS calculates an empirical $p$-value range~\citep{mcfowland-fgss-2013} for each
treatment unit to obtain a measure of how ``anomalous'' or unusual a particular
unit's outcome is given its reference distribution. For each unit $R_i$ in the
treatment group ($W_i=1$), using \eqref{eq:control-edf} and an appropriate weighting scheme,\ignore{the empirical $p$-value $p(y_i,x_i)$(or $p_{i}$ for notational convenience) is derived from its corresponding $p$-value range. Observe that} the standard
empirical $p$-value would be
\begin{equation}
  \label{eq:p-value}
  \hat{p}(y; x) = \hat{\mathbb{F}}_{Y^C|X}(y|x) =
  \frac{\sum_{i=1}^n{\zeta_i(x)\mathbbm{1}_{\{Y^{\text{obs}}_i \le y\}}}}{\sum_{i=1}^n \zeta_i(x)}.
\end{equation}
The empirical $p$-value \emph{range} is an extension of this
traditional empirical $p$-value, defined as 
\begin{equation}
  \label{eqn:p-value-range}
  \begin{split}
      \hat{p}(y; x) & = \left[ \hat{p}_{\text{min}}\left(y; x\right) , \hat{p}_{\text{max}}\left(y; x\right) \right] \\
            & = \left[
                    \frac{\sum_{i=1}^n{\zeta_i(x)\mathbbm{1}_{\{Y^{\text{obs}}_i < y\}}}}{1 + \sum_{i=1}^n \zeta_i(x)},
                    \frac{1 + \sum_{i=1}^n{\zeta_i(x)\mathbbm{1}_{\{Y^{\text{obs}}_i \le y\}}}}{1 + \sum_{i=1}^n \zeta_i(x)}
                \right],
  \end{split}
\end{equation}
where the sums are taken over all control observations. The numerator of 
 $\hat{p}_{\text{max}}$, but not 
$\hat{p}_{\text{min}}$, includes ``tied'' observations (i.e., $Y^{\text{obs}}_i = y$). The treatment observation $y$ is also considered part of its own reference distribution, following from the assumption of exchangeability of control and treatment outcomes under $H_0$, and thus adding one to the denominators of $\hat{p}_{\text{min}}$ and $\hat{p}_{\text{max}}$ as well as the numerator of $\hat{p}_{\text{max}}$. Following~\cite{mcfowland-fgss-2013}, we use empirical
$p$-value ranges because they improve upon traditional empirical $p$-values. The $p$-value ranges
are better equipped for sparsity in high-dimensional data, as the range naturally adapts to the amount of reference data used for estimation: a treatment unit's $p$-value range shrinks as more control units are used to estimate its reference distribution. Additionally, 
if we represented $R_i$ with an empirical $p$-value $\hat{p}_i$ that is drawn uniformly at random from its empirical $p$-value range
$\hat{p}(y_i; x_i)$, then under $H_0$, $\hat{p}_i \sim \text{Uniform}(0,1)$.\footnote{From the exchangeability under $H_0$ of $Y^{\text{obs}}_i \sim \mathbb{F}_{Y(1) \mid X}$ and $Y^{\text{obs}}_j \sim \mathbb{F}_{Y(0) \mid X}$, and the probability integral transform.} Standard empirical $p$-values are only \textit{asymptotically} distributed as Uniform$[0,1]$ and exhibit finite sample bias, while the ranges are unbiased in finite samples, ensuring that $\mathbb{E}\left[ \hat{\mathbb{F}}_{Y (1)| X=x}(\hat{\mathbb{F}}^{-1}_{Y (0)|X=x} (\alpha)) \right] = \alpha$ under $H_0$.

The left-tailed $p$-value ranges defined in 
\eqref{eqn:p-value-range} identify outcomes in the extremes of the lower-tail of the reference distribution. The $p$-value range in relation to only the right-tail
of the reference distribution or both tails can be derived from the $p$-value
range specified for the left-tail. The right-tail
range is
\begin{equation*}
  \hat{p}(y; x) = \left[ 1-\hat{p}_{\text{max}}(y; x) , 1-\hat{p}_{\text{min}}(y; x)  \right];
\end{equation*}
while the two-tailed range is
\begin{equation*}
  \hat{p}(y; x) =
  \begin{cases}
    \left[ 2 \hat{p}_{\text{min}}(y; x), 2 \hat{p}_{\text{max}}(y; x) \right] & \text{if }
      \hat{p}_{\text{max}}(y; x) < 0.5 \\
    \left[ 2\left(1-\hat{p}_{\text{max}}(y; x)\right),
      2\left(1-\hat{p}_{\text{min}}(y; x)\right) \right] & \text{if }
      \hat{p}_{\text{min}}(y; x) \ge 0.5 \\
    \left[ 2\min\{ \hat{p}_{\text{min}}(y; x) , 1-\hat{p}_{\text{max}}(y; x)\}, 1 \right] &
      \text{otherwise}.
  \end{cases}
\end{equation*}
Finally, the significance of a
$p$-value range, for a significance level $\alpha$, is defined as
\begin{equation*}
 n_\alpha(\hat{p}(y; x)) =
  \begin{cases}
   1 & \text{if } \hat{p}_{\text{max}}(y; x) < \alpha \\
   0 & \text{if } \hat{p}_{\text{min}}(y; x) > \alpha \\
   \frac {\alpha-\hat{p}_{\text{min}}(y; x)} {\hat{p}_{\text{max}}(y; x)-\hat{p}_{\text{min}}(y; x)} & \text{otherwise}.
  \end{cases}
\end{equation*}
Intuitively, $n_\alpha(\hat{p}(y; x))$ measures the proportion of the range that is
significant at level $\alpha$, or equivalently, the probability that a $p$-value
drawn uniformly from $\left[\hat{p}_{\text{min}}(y; x), \hat{p}_{\text{max}}(y; x)
\right]$ is less than $\alpha$.

\subsection{Subpopulations}
\label{sec:subpopulations}
Given $p$-values as a measure of the anomalousness of individual
treatment units, we now consider how TESS combines these measures to form
subpopulations. For intuition, we propose representing the data as a tensor,
where each covariate is represented by a mode of the tensor, $X = (X^{1},
\ldots, X^{d})$, resulting in a $d$-order tensor. $|V^j|$, the arity of the
$j^{th}$ covariate, is the size of the $j^{th}$ mode. Therefore, each
covariate profile $x$ maps to a unique cell in the tensor, which contains the
$p$-values of the treatment units that share $x$ as their covariate profile.
As stated above, a subpopulation is $S = v^1 \times \ldots \times v^d$,
where $v^j \subseteq V^j$; therefore, an individual cell (i.e.,  covariate profile $x$) is itself a
subpopulation: $S = \{x^1\} \times \ldots \times \{x^d\}$, where $x^j \in V^j$. For a
demonstrative example see Table \ref{table:example_tensor}. For a given
subpopulation $S$, we define the quantities
\begin{equation}
 \label{eqn:C(S),N_alpha,N}
    N_\alpha(S) =  \sum_{x \in U_{X}(S)}\sum_{y \in Y^{Tr}(x)}{ n_\alpha\left(\hat{p}\left(y; x\right)\right)};\quad N(S) =   \sum_{x \in U_{X}(S)}\sum_{y \in Y^{Tr}(x)}{1}
\end{equation}
\begin{table}
  \captionsetup{font=scriptsize,skip=0pt}
  \centering
  \resizebox{0.30\columnwidth}{!}{
  \begin{tabular}{llll}
                                              & \multicolumn{1}{l|}{}      & \multicolumn{2}{c|}{Gender}                                       \\
                                              & \multicolumn{1}{l|}{}      & \multicolumn{1}{l|}{Male}       & \multicolumn{1}{l|}{Female}     \\ \hline
    \multicolumn{1}{c}{\multirow{2}{*}{Race}} & \multicolumn{1}{l|}{Black} & \multicolumn{1}{l|}{$\{1.7\}$}  & \multicolumn{1}{l|}{$\{1.73\}$} \\ \cmidrule{2-4}
    \multicolumn{1}{c}{}                      & \multicolumn{1}{l|}{White} & \multicolumn{1}{l|}{$\{1.59\}$} & \multicolumn{1}{l|}{$\{1.84\}$} \\ \hline
    \multicolumn{4}{l}{}                                                                                                                       \\
    \multicolumn{4}{c}{Control Group}
  \end{tabular}
  }
  \quad\quad
  \resizebox{0.30\columnwidth}{!}{
  \begin{tabular}{llll}
 
                                              & \multicolumn{1}{l|}{}      & \multicolumn{2}{c|}{Gender}                                       \\
                                              & \multicolumn{1}{l|}{}      & \multicolumn{1}{l|}{Male}       & \multicolumn{1}{l|}{Female}     \\ \hline
    \multicolumn{1}{c}{\multirow{2}{*}{Race}} & \multicolumn{1}{l|}{Black} & \multicolumn{1}{l|}{$\{2.92\}$}  & \multicolumn{1}{l|}{$\{2.35\}$} \\ \cmidrule{2-4}
    \multicolumn{1}{c}{}                      & \multicolumn{1}{l|}{White} & \multicolumn{1}{l|}{$\{2.21\}$} & \multicolumn{1}{l|}{$\{2.06\}$} \\ \hline
    \multicolumn{4}{l}{}                                                                                                                       \\
    \multicolumn{4}{c}{Treatment Group}
  \end{tabular}
  }
  \caption{A demonstrative tensor--representing the example dataset
in Table \ref{table:example_data}--containing a $d=2$-order tensor for both the
control and treatment group. The top-left cell of each tensor represents the
subpopulation of black males in the data, $S = \{\text{Black}\} \times
\{\text{Male}\}$. There are also the subpopulation of all males, $S =
\{\text{Black, White}\} \times \{\text{Male}\}$, all black subjects, $S =
\{\text{Black}\} \times \{\text{Male, Female}\}$, or the entire population,
$S = \{\text{Black, White}\} \times \{\text{Male, Female}\}$; there are a
total of nine subpopulations in this simple example. We note that the example
dataset has only one unit with each unique covariate profile, therefore the set
of values in each tensor cell is of size one.}
  \label{table:example_tensor}
\end{table}
where $U_{X}(S)$ is
the set of non-empty covariate profiles in $S$, $Y^{Tr}(x)= \{Y_i^{obs} | X_i = x, W_i = 1\}$ is the collection of treatment units' outcomes with covariate
profile $x$, $N(S)$ represents the total number of empirical $p$-values contained in $S$, and
$N_{\alpha}(S)$ is the number of $p$-values in $S$ that are less than $\alpha$.\footnote{For $p$-value ranges, as in \cite{mcfowland-fgss-2013}, $N_{\alpha}(S)$ is more precisely the total probability mass less than $\alpha$ over the $p$-value ranges in $S$.} Given that the distribution of each $p$-value is
Uniform(0,1) under the null hypothesis that the treatment has no effect, for a
subpopulation $S$ consisting of $N(S)$ empirical $p$-values,
$\mathbb{E}\left[N_{\alpha}(S)\right] = \alpha N(S)$. Under the alternative hypothesis,
we expect the outcomes of the
affected units to be more concentrated in the tails of their reference
distributions; thus, the $p$-values for these affected units
will be lower. Therefore, subpopulations composed of covariate profiles that are
systematically affected by the treatment should express higher values of
$N_\alpha(S)$ for some $\alpha$. Consequently, a subpopulation $S$ where
$N_\alpha(S) > \alpha N(S)$ (i.e., with a higher than expected
number of low, significant $p$-values) is potentially affected by the
treatment.

\subsection{Nonparametric Scan Statistic}
\label{sec:npss}
TESS utilizes the nonparametric scan statistic
\citep{mcfowland-fgss-2013,feng-npss_graph-2014} to evaluate the statistical
anomalousness of a subpopulation $S$ by comparing the observed and expected
number of significantly low $p$-values it contains.
The general form of the nonparametric scan
statistic is
\begin{equation*}
  \ignore{\label{eqn:F(S)-npss}}
  F(S)=\max_{\alpha}F_{\alpha}(S)=\max_{\alpha}\Delta(\alpha,N_{\alpha}(S),N(S)),
\end{equation*}
where $N_{\alpha}(S)$ and $N(S)$ are defined as in \eqref{eqn:C(S),N_alpha,N}, and $\Delta$ is a measure of divergence measuring the anomalousness of the $p$-values in $S$. See Appendix~\ref{sec:scoring_functions}for a collection of goodness-of-fit scoring functions written in the general form of the nonparametric scan statistic. As described in~\eqref{eq:BJ-teststat}, in this work we utilize the Berk-Jones scan statistic: $\max_{\alpha} \Delta_{BJ}\left(\alpha,N_{\alpha}(S),N(S)\right) = \max_{\alpha}N(S)KL\left(  \frac{N_{\alpha}(S)}{N(S)} , \alpha \right)$, a log-likelihood ratio test statistic of the distributional treatment
effect in subpopulation $S$. Maximizing
$F(S)$ over a range of $\alpha$, rather than a single
arbitrarily-chosen $\alpha$ value, enables TESS
to detect a small number of highly anomalous $p$-values, a larger
subpopulation with subtly anomalous $p$-values, or anything in between.  We consider ``significance levels'' $\alpha \in [\alpha_{\text{min}},\alpha_{\text{max}}]$, for constants $0 < \alpha_{\text{min}} < \alpha_{\text{max}} < 1$.  The range of $\alpha$ to consider can be specified based on the quantile values of interest.  The choice of $\alpha_{\text{max}}$ describes how extreme a value must be, as compared to the reference distribution, in order to be considered significant.  We often choose $\alpha_{\text{min}} \approx 0$, but larger values can be used to avoid returning subsets with a small number of extremely significant $p$-values.

\subsubsection{Efficient Scanning}
\label{sec:eff_scan}
The next step in the TESS framework is to detect the subpopulation most
affected by the treatment, i.e., to identify the most anomalous subset of values
for each of the $d$ modes of the tensor, or equivalently for each covariate $X^1
\ldots X^d$. More specifically, the goal is to identify the set of subsets
$\{v^{1},\ldots, v^{d}\}$ where each element corresponds to values in a
tensor-mode (covariate), such that $F(v^{1} \times \ldots \times v^{d})$ is
jointly maximized. The computational complexity of solving this optimization
naively is $O^{(2^{\sum_j |V^j|})}$, where $|V^j|$ is the size of mode $j$ (the
arity of $X^j$), and is computationally infeasible for even moderately sized
datasets.

We therefore employ the linear-time subset scanning property (LTSS) \citep{neill-ltss-2012}, which allows for efficient and exact maximization of any function satisfying LTSS over all subsets of the data. We formally define the LTSS property below, but intuitively it guarantees that the optimization over all subsets $S$ can be done by ranking data elements (according to a specific ``priority function'') and then only considering the top-$t$ subsets as candidates.\\[1ex]
\underline{\emph{LTSS Property Definition:}} Given a set of data elements $R=\{R_1,
\ldots, R_n \}$, a score function $F(S)$ mapping $S\subseteq R$ to a real
number, and a priority function $G(R_i)$ mapping a single data element $R_i \in
R$ to a real number. If $F(S)$ satisfies the \emph{LTSS property} with priority function $G(R_i)$,
then the only subsets with the
potential to be optimal are those consisting of the top-$t$ highest priority
records, $S \in \{\{R_{(1)}, \ldots, R_{(t)}\} \}_{t\in\{1,2,\ldots,n\}}$. In other words, there exists some $t \in \{1,2,\ldots,n\}$ such that $\arg\max_S F(S) = \{R_{(1)}, \ldots, R_{(t)}\}$.  We also
formally restate the original LTSS theorem:
\begin{thm}[\cite{neill-ltss-2012}]
  \label{thm:LTSS}
    Let $F(S) = F(X, Y)$ be a function of two additive
  sufficient statistics of subset $S$, $X(S) = \sum_{R_i \in S} x_i$ and
  $Y(S) = \sum_{R_i \in S} y_i$, where $x_i$ and $y_i$ depend only on
  element $R_i$. Assume that $F(S)$ is monotonically increasing with $X(S)$,
  that all $y_i$ values are positive, and that $F(X, Y)$ is convex. Then
  $F(S)$ satisfies the LTSS property with priority function $G(R_i) =
  \frac{x_i}{y_i}$.
\end{thm}
In this work, we use Theorem~\ref{thm:LTSS} to optimize $F_\alpha(S) = \Delta(\alpha,N_\alpha(S),N(S))$ with a fixed value of $\alpha$; therefore, $X(S) = N_\alpha(S)$ and $Y(S) = N(S)$ are ``additive sufficient statistics'', i.e., both $N_\alpha(S)$ and $N(S)$ are additive statistics of $S$, from \eqref{eqn:C(S),N_alpha,N}, and $F_\alpha(S)$ can be written as $F_\alpha(N_{\alpha}(S), N(S))$. Moreover, for $F_\alpha(S)$ (with $\alpha$ fixed) to satisfy LTSS, we also assume: \emph{(A1)} $\Delta$ is monotonically \emph{\textbf{increasing}} w.r.t.
$N_{\alpha}$, \emph{(A2)} $\Delta$ is monotonically \emph{\textbf{decreasing}} w.r.t. $N$, and \emph{(A3)} $\Delta$ is \emph{\textbf{convex}} w.r.t. $N_{\alpha}$ and $N$.
These properties are intuitive because the ratio of observed to expected number
of significant $p$-values $\frac{N_\alpha}{\alpha N}$ increases with the
numerator (A1) and decreases with the denominator (A2). Also, a fixed ratio of
observed to expected is more significant when the observed and expected
counts are large (A3). In Appendix \ref{sec:scoring_functions}, we show that nonparametric scan statistics using a large class of goodness of fit functions, including the Berk-Jones scan statistic utilized in this work, exhibit these properties.

We now extend Theorem \ref{thm:LTSS} to the (potentially high-dimensional) tensor context using Corollary \ref{cor:LTSS-NPSS} below. Essentially, the corollary demonstrates that the nonparametric scan statistic
satisfies LTSS in the context of TESS, and therefore a single mode of a tensor can be efficiently optimized over subsets, conditioned on the subsets of values for the other modes. Let $U_{\alpha}(S)$ be
the set of unique $p$-values between $\alpha_{\text{min}}$ and $\alpha_{\text{max}}$ contained in subpopulation $S$. Then the quantity $\max_S F(S) = \max_{\alpha \in U_{\alpha}(S)} \max_S  F_\alpha(S)$
can be efficiently and exactly computed over all subsets $S = v^j \times
v^{-j}$, where $v^j \subseteq V^j$, for a given subset of values for each of the
other modes $v^{-j}$.\footnote{Note that for convenience of notation we define $S = v^j \times v^{-j}$; however, the elements of the set $v^j$ still appear at the $j^{\text{th}}$ position of the covariate profiles in $S$.} To do so, consider the set of distinct $\alpha$ values, $U =
U_{\alpha}(V^j \times v^{-j})$.  For each $\alpha \in U$ we employ
the logic described in Corollary \eqref{cor:LTSS-NPSS} to optimize $F_{\alpha}(S)$: we compute the priority $G_\alpha(v^{j}_m)$ for each value
($v^{j}_m \in V^j$), sort the values based on priority function $G_{\alpha}(v^{j}_m)$, and evaluate subsets of the form $S=\{v^{j}_{(1)}, \ldots, v^{j}_{(t)}\}
\times v^{-j}$ consisting of the top-$t$ highest priority values, for $t=1,
\ldots,|V^j|$.

\begin{cor}
  \label{cor:LTSS-NPSS} Consider the nonparametric scan
  statistics $F(S) = \max_\alpha F_\alpha(S)$, where the significance level
  $\alpha \in [\alpha_{\text{min}},\alpha_{\text{max}}]$, for constants $0 < \alpha_{\text{min}} < \alpha_{\text{max}} < 1$.
  For a given value of $\alpha$ and $v^{-j} = v^1\times \ldots\times v^{j-1}\times v^{j+1}\times
  \ldots\times v^{d}$ under consideration, $F_{\alpha}(S)$ can
  be efficiently maximized over all subpopulations $S = v^j \times v^{-j}$, for
  $v^j \subseteq V^j$. 
 \end{cor}
\begin{proof}
 ~We have $F_{\alpha}(S) = \Delta(\alpha,N_{\alpha}(v^j), N(v^j))$, with the additive sufficient statistics $N_{\alpha}(v^j) = \sum_{x \in U_{X}(v^j \times v^{-j})} \sum_{y \in Y^{Tr}(x)} n_{\alpha}(\hat{p}(y; x))$ and $N(v^j)=\sum_{x \in U_{X}(v^j \times v^{-j})} \sum_{y \in Y^{Tr}(x)} 1$, noting that the number of $p$-values in every
  $v^j$ is positive, as we only consider the values of a covariate that are expressed by some treatment
  unit. Since the nonparametric scan statistic is defined to be monotonically
  increasing with $N_{\alpha}$ (A1), monotonically decreasing with $N$ (A2), and
  convex (A3), we know that $F_{\alpha}(S)$ satisfies the LTSS property with
  priority function, over the values of mode (covariate) $j$,  $G_{\alpha}(v^{j}_m) = \frac{\sum_{x \in U_{X}(v^{j}_m \times v^{-j})} \sum_{y \in Y^{Tr}(x)} n_{\alpha}(\hat{p}(y; x))} {\sum_{x \in U_{X}(v^{j}_m \times v^{-j})} \sum_{y \in Y^{Tr}(x)} 1}$ for $v^j_m \in
  V^j$.
  Therefore the LTSS property holds for each value of $\alpha$, enabling each
  $F_{\alpha}(S)$ to be efficiently maximized over subsets of values for the
  $j^{th}$ mode of the tensor, given values for the other $d-1$ modes.
\end{proof}

TESS iterates over
modes of the tensor, using the efficient optimization steps described above to
optimize each mode: $v^j = {\arg\max}_{v^j \subseteq V^j} F(v^j \times v^{-j}),  j = 1 \ldots d$. The cycle of optimizing each mode continues until convergence, at
which point TESS has
reached a conditional maximum of the score function, i.e., $v^j$ is
conditionally optimal given $v^{-j}$ for all $j = 1 \ldots d$. This ordinal ascent
approach is not guaranteed to converge to the joint optimum, but with multiple random restarts the combination of subset scanning and ordinal ascent has been shown to locate near
globally optimal subsets with high probability~\citep{neill-mvltss-2013, mcfowland-fgss-2013}. We further provide asymptotic guarantees for TESS to  recover the precisely correct subpopulation that is also shown to be the globally optimal subset (Theorems \ref{thm:assympt_hom_str} and \ref{thm:assympt_TESS}). Moreover, if $\sum_{j=1}^d |V^j|$ is large,
this iterative procedure makes the ability to detect anomalous subpopulations
computationally feasible, without excluding potentially optimal subpopulations from the search space
(as a greedy top-down approach may). A single iteration 
(optimization of mode $j$ of the tensor) has a complexity of
$O\left(|U| \left(n_t + |V^j| \log |V^j|\right)\right)$, where the $n_t$ term---the number of treatment units---results from collecting the $p$-values for all units in
$V^j \times v^{-j}$ over our sparse tensor; $U = U_{\alpha}\left(V^j
\times v^{-j}\right)$, with $|U| \le n_t$~\citep{mcfowland-fgss-2013}; and $O\left(|V^j| \log
|V^j|\right)$ is required to sort, based on the priority, the values of tensor mode $j$. Therefore a
step in the procedure (a sequence of $d$ iterations over all modes of the
tensor) has complexity $O\left(\bar{U} d \left(n_t + \bar{V} \log \bar{V}\right)\right)$,
where $\bar{U}$ and $\bar{V}$ are the average numbers of $\alpha$ thresholds considered
and covariate arity, respectively. Thus
the TESS search procedure has a total complexity of $O\left(I \bar{Z} \bar{U} d \left(n_t +
\bar{V} \log \bar{V}\right)\right)$, where $I$ is the number of random restarts and $\bar{Z}$
is the average number of iterations required for convergence.  We note that $\bar{Z}$ is typically very small; $\bar{Z} \le 5$ across all simulations discussed in \S\ref{sec:results}.

\subsection{TESS Algorithm}
\label{sec:algo}

Inputs: randomized experiment dataset, $\alpha_{\text{min}}$, $\alpha_{\text{max}}$, number of iterations $I$.\vspace{-1mm}
\begin{enumerate}

  \item For each unique covariate profile $x$ in the treatment group:
    \begin{enumerate}
      \item Estimate $\hat{\mathbb{F}}_{Y^C \mid X = x}$ from the outcomes of the units in the control group.
      \item Compute the $p$-value (range) $\hat{p}_i = \hat{p}(y_i; x_i)$ for each treatment unit $i$ with profile
            $x$ from $\hat{\mathbb{F}}_{Y^C \mid X = x}$.
    \end{enumerate}

  \item Iterate the following steps $I$ times. Record the maximum value
  $\hat F^\ast$ of $F(S)$, and the corresponding subset of values $v^{j\ast}$ for each of the $d$ modes, over all such iterations:

  \begin{enumerate}

    \item For each of the $d$ modes, initialize $v^j$ to a random subset of values $V^j$.

    \item Repeat until convergence:
    \begin{enumerate}
      \item For each of the $d$ modes:
      \begin{enumerate}
        \item Maximize $F(S) = \max_{\alpha \in [\alpha_{\text{min}}, \alpha_{\text{max}}]}
        F_\alpha(v^j \times v^{-j})$ over subsets of values for $j^{th}$ mode
        $v^j \subseteq V^j$, for the current subset of values of the other $d-1$
        modes $v^{-j}$, and set $v^j \leftarrow \arg\max_{v^j \subseteq V^j}
        F(v^j \times v^{-j})$.
      \end{enumerate}
    \end{enumerate}

  \end{enumerate}

  \item Output $\hat S^\ast = v^{1 \ast}\times \ldots\times v^{d \ast}$ and the corresponding score $\hat F^\ast=F(\hat S^\ast)$.
\end{enumerate}

\subsection{Estimator Properties}
\label{sec:estimator-theory}
In the above sections we outline a procedure to efficiently estimate $\max_{S \in Rect} F(S)$,
where $Rect$ represents the space of all rectangular subsets of $D$. In this section we
treat $\max_{S \in Rect} F(S)$ as a statistic of the data, and aim to show that it
has desirable statistical properties. It is known that for data $X_1,\ldots,X_n \distiid \mathbb{F}$ and the corresponding empirical distribution function $\mathbb{F}_n$, $\|\mathbb{F}_n - \mathbb{F}\|_\infty \xrightarrow{a.s.}\ 0$.  Many goodness-of-fit statistics $GoF(\mathbb{F}_n, \mathbb{F})$ are equivalent to an empirical process over centered and scaled empirical measures; and empirical process theory provides tools to control Type I and II error \citep{gaenssler-glivenko_cantelli-2004,dvoretzky-dkw-1956,shorack-emp_process-1986}. However, 
in a general sense our goal is to control the behavior of $\max_{S \subseteq \{X_1,\ldots,X_n\}}GoF(\mathbb{F}_S, \mathbb{F})$, where $\mathbb{F}_S$ is the empirical distribution given by the subset $S$. It is not obvious whether the desirable properties present for $\mathbb{F}_n$ will persist when considering the empirical distribution of $\mathbb{F}_S$, a non-random subset of the data chosen by our optimization procedure. Given that this context of optimization over subsets is not considered in the current goodness-of-fit literature, we provide various theoretical results in support of our
subset scanning algorithm. In the remainder of the section we present the key statements necessary to show our desired properties below, while additional results and all proofs can be found in Appendix~\ref{sec:proofs}. We begin with
\ignore{:
\begin{restatable*}{lem}{bjtona}
  \label{lem:BJ_to_NA}
  $F^{BJ}(S) \asymp F^{NA}(S)$ as $N(S)\longrightarrow \infty$.
\end{restatable*}
\noindent This result indicates that, as the number of subjects in a
given subpopulation grows, its score under $F^{BJ}$ is well approximated by
$F^{NA}$, where both functions are described in \S\ref{sec:div-stat}. Given the fact that a large class of other
goodness-of-fit statistics in the literature are monotonic transformations of
$F^{NA}$ (see Appendix~\ref{sec:scoring_functions}), this result allows us to focus the remainder of our results on $F^{NA}$. Our}
the fact that our score function can be considered a test statistic for a hypothesis test analogous to that described in~\eqref{eq:test_ceqte}:
\begin{flalign}
  \label{eq:test}
  &\!\begin{aligned}
  H_0:~ &Y_i(1)|X_i~ \sim \mathbb{F}_{Y_i(0) | X_i}~ \forall X_i \in U_X(D)&\\
  H_1\left(S\right):~& \begin{cases}
      Y_i(1)|X_i~\not\sim \mathbb{F}_{Y_i(0) | X_i} ~ \forall X_i \in
  U_X(S)\text{, }S \in Rect\\
  Y_i(1)|X_i~\sim \mathbb{F}_{Y_i(0) | X_i}~ \forall X_i \not\in
  U_X(S)\text{, }S \in Rect&
  \end{cases}
  \end{aligned}&
\end{flalign}
where $D$ is our dataset (or tensor) of treatment units and $Rect$ is the set of all rectangular subsets of $D$.\footnote{The null hypotheses defined in~\eqref{eq:test_ceqte} and~\eqref{eq:test} are analogous because if $\tau_{\text{CQTE}_{\alpha}}(x) = \alpha~ ~ \forall x\in S, \forall \alpha$, then the distributions for $Y(1) | X$ and $Y(0) |X$ are the same. Similarly, the alternative hypotheses are analogous because if $\exists \beta > \alpha$ such that $\tau_{\text{CQTE}_{\alpha}}(x) = \beta~ \forall x \in S$, then $Y(1) | X$ is different from $Y(0) | X$, and if $\tau_{\text{CQTE}_{\alpha}}(x) = \alpha~~ \forall x \notin S,~ \forall \alpha$, then again the distributions for $Y(1) | X$ and $Y(0) |X$ are the same.} The null hypothesis
is that all of the observed outcomes of treatment units are drawn from the same
conditional outcome distribution (given the observed covariates) as their
control group.
\ignore{The hypothesis tests which serve as the foundation for the score functions described in \S\ref{sec:div-stat} are special cases of this more general hypothesis test in~\eqref{eq:test}.}
Recall that $U_{X}(D)$ is
the set of unique covariate profiles (non-empty tensor cells) in our data, with cardinality $|U_{X}(D)| = M$; while $S^{\ast}=\arg \max_{S \in
Rect} F(S)$ and $S^{\ast}_{u} = \arg \max_S F(S)$ represent the most anomalous
rectangularly constrained subset and the most anomalous unconstrained subset respectively. For mathematical convenience, we
assume $N(x) = n$ for all $x \in U_{X}(D)$, i.e., $n$ units
belong to each unique
covariate profile (non-empty cell) in the data and treatment condition.\footnote{We can redefine $n = \min_{x} N(x)~\forall x \in U_{X}(D)$ and our results can be extended.}  We consider the case where $n \longrightarrow \infty$, maximizing $F(S)= \max_{\alpha\in[\alpha_{\min},\alpha_{\max}]} F_\alpha(S)$ for $0 < \alpha_{\min} < \alpha_{\max} < 1$.
We can therefore
demonstrate: 
\begin{restatable*}{lem}{nullconverg}
  \label{lem:null_converg}
  Under $H_{0}$ defined in \eqref{eq:test}, let $N(x) = n~ \forall x \in U_{X}(D)$, then as $n\rightarrow \infty$,
\ignore{
  \begin{align*}
      \sqrt{F(S^{\ast}_{u})} &\stackrel{d}{\longrightarrow}
 \max_Z \left( \sqrt{\frac{M\phi(Z)^2}{2(1-\Phi(Z))}} +  \mathbb{W}(\epsilon) \sqrt{\frac{V(Z)}{2}} \right)\\ 
      &\le C \sqrt{M} + \frac{ \mathbb{W}(\epsilon)}{\sqrt{2}},
  \end{align*}}
   \begin{align*}
      \sqrt{F(S^{\ast}_{u})} \stackrel{d}{\longrightarrow}~&G \left( \mathbb{W}\left(\alpha_{\min},
      \alpha_{\max}\right), M\right)\\ 
      \le~&C \sqrt{M} + \frac{ \mathbb{W}(\alpha_{\min},\alpha_{\max})}{\sqrt{2}},
  \end{align*}
where the function $G$ and constant $C<1$ are known; $\mathbb{W}(\alpha_{\min},\alpha_{\max}) = \sup_{\alpha \in [\alpha_{\min},\alpha_{\max}]} \frac{|B(\alpha)|}{\sqrt{\alpha(1-\alpha)}}$, and $B(\alpha)$ represents a Brownian bridge on $[0,1]$.
\end{restatable*}
\noindent Thus, when the null hypothesis is true, the most anomalous unconstrained subset's score distribution can be upper bounded. Our ability to understand the limiting behavior of the $F\left(S^{\ast}_{u}\right)$ exploits  its structure, which we get from LTSS theory: the optimal unconstrained subset will be $S^{\ast}_{u} \in \{\{x_{(1)}, \ldots, x_{(t)}\}\}_{t\in\left\{1,2,\ldots,M\right\}}$, where $x_{(t)}$ has the $t^{th}$ largest value of the random variable $\frac{N_{\alpha}(x)}{N(x)}~\forall x \in U_{X}(D)$. Next we note that the score maximized over the space of unconstrained subsets upper bounds the score maximized over the subspace of rectangular subsets, i.e.,  $F\left(S^{\ast}\right) \le F\left(S^{\ast}_{u}\right)$. We use this fact to obtain the following result:
\begin{restatable*}{thm}{falseposotive}
  \label{thm:false_posotive}
   Under $H_{0}$ defined in~\eqref{eq:test}, let $N(x) = n~ \forall x \in U_{X}(D)$ and fix Type-I error rate $\delta > 0$, then there exists a critical value $h(\delta)$ such that
  \begin{equation*}
      \lim_{n \rightarrow \infty}P_{H_0}\left(\max_{S \in Rect} F(S) > h\left(\delta\right)\right) \le \delta.
  \end{equation*}
\end{restatable*}
\noindent Theorem~\ref{thm:false_posotive} indicates that $\max_{S \in Rect} F(S)$ provides a statistic to quantify the evidence to reject $H_{0}$, enabling an (asymptotically) valid $\delta$-level hypothesis test such that $P_{H_0}(\text{Reject }H_0) \le \delta$, for any fixed Type I error rate $\delta > 0$. From the proof of Theorem~\ref{thm:false_posotive}, in Appendix~\ref{sec:proofs}, we derive that $h(\delta)= \left(0.45\sqrt{M} + \frac{w(\delta)}{\sqrt{2}}\right)^2$, where $w(\delta)$ returns $w$ such that $P\left(\mathbb{W}(\alpha_{\min},\alpha_{\max}) > w\right) = \delta$. For intuition, $w(\delta)$ is typically small, e.g., $w(\delta) \approx 5.81$ for $\delta = 10^{-6}$ and $(\alpha_{\min}, \alpha_{\max}) = (.01, .99)$. We note that because we are maximizing both over subsets $S$ and thresholds $\alpha$, these results are distinct from the straightforward application of known results from empirical process theory or Dvoretzky-Kiefer-Wolfowitz bounds, which would give us $\max_\alpha | \frac{N_\alpha(S)}{N(S)} - \alpha| \longrightarrow 0$ for a given $S$.

Next, we turn our attention to the alternative hypothesis, where $S^T \in Rect$ represents the truly affected subset, $k=\frac{|U_{X}(S^T)|}{|U_{X}(D)|}$ is the proportion of covariate profiles included in $S^T$, and $H_{1}\left(S^T\right)$ implies that there exist constants $\alpha$ and $\beta(\alpha) > \alpha$ such that $\beta(\alpha) = \mathbb{F}_{Y(1) \mid X=x}\left(\mathbb{F}_{Y(0) \mid X=x}^{-1}(\alpha)\right)$ for all $x \in U_{X}(S^T)$.
We then have the following results:

\begin{restatable*}{lem}{altconverg}
\label{lem:alt_converg}
Under $H_{1}\left(S^T\right)$ defined in~\eqref{eq:test}, let $N(x) = n~ \forall x \in U_{X}(D)$, and consider $F_{\alpha^{\ast}}(S^T)$ for $\alpha^{\ast} = \arg\max_{\alpha \in [\alpha_{\min},\alpha_{\max}]} \frac{(\beta(\alpha)-\alpha)^2}{2 \alpha(1-\alpha)}$ and $\beta^{\ast} = \beta(\alpha^{\ast})$. Then as $n\longrightarrow \infty$,
\ignore{
  Under $H_{1}\left(S^T\right)$, $F\left(S^T\right) \asconv
  \max_{\alpha} (\beta\left(\alpha\right) -
  \alpha)^2\frac{kMn}{2\alpha(1-\alpha)}$.
}
\begin{equation*}
    \sqrt{F_{\alpha^{\ast}}(S^T)}-O\left(\sqrt{kMn}\right) \stackrel{d}{\longrightarrow} \mbox{Gaussian}\left(0,\sigma^2_{\alpha^{\ast}\beta^{\ast}}\right),
\end{equation*}
where $\sigma^2_{\alpha^{\ast}\beta^{\ast}} > 0$ does not depend on $k$, $M$, or $n$.
\end{restatable*}

\noindent Thus, when the null hypothesis is false, the expected value of the true subset's score at the $\alpha^{\ast}$ quantile, $F_{\alpha^{\ast}}(S^T)$, is increasing with $n$. This result, and the fact that $F_{\alpha}(S) \le F(S)~\forall S,\alpha$, are used to obtain the following result:
\begin{restatable*}{thm}{power}
  \label{thm:power}
  Under $H_{1}\left(S^T\right)$ defined in \eqref{eq:test}, let $N(x) = n~ \forall x \in U_{X}(D)$ and critical value $h(\delta)$ be set for the same fixed Type-I error rate $\delta > 0$ as in Theorem \ref{thm:false_posotive}, then
  \begin{equation*}
      \lim_{n \rightarrow \infty}P_{H_1}\left(\max_{S \in Rect} F(S) > h\left(\delta\right)\right) = 1.
  \end{equation*}
\end{restatable*}

\noindent As a consequence of Theorem~\ref{thm:power}, the $\delta$-level hypothesis test based on $\max_{S \in Rect} F(S)$ has full asymptotic power $P_{H_1}(\text{Reject }H_0) \longrightarrow 1$. Note that in this context we consider a fixed alternative $\beta(\alpha)$, as opposed to a local alternative where $\beta_{n}(\alpha) \longrightarrow \alpha$ as $n \longrightarrow \infty$.\ignore{Additionally, because the critical value $h\left(\delta\right)$ is the same in Theorems \ref{thm:false_posotive} and \ref{thm:power}, we are therefore showing that $P(Type~I ~error) + P(Type~II~error) \rightarrow 0$.} 

In practice, when we consider an experiment with finite $M$ and $n$, permutation testing can be used to control the Type I error rate of our scanning procedure, and conditions have been shown where permutation calibrations achieve the Type II error rates of an oracle scan test~\citep{arias-castro-calibration-2017}. These theoretical and practical results intuitively capture our statistic's ability to conclude that the
null hypothesis is false--i.e., that there exists some subset that follows $H_1$, and
therefore invalidates $H_{0}$. However, this does not necessarily provide a
guarantee that the statistic will exactly capture the true subset. Therefore, next
we derive finite sample sufficient conditions under which our framework achieves subset correctness: $S^{\ast} = S^T$. We then show asymptotic convergence of $P(S^{\ast} = S^T ) \xrightarrow[]{} 1$ as $n\xrightarrow[]{}\infty$.

We introduce additional notation for this discussion: \ignore{$\beta_x(\alpha) = \mathbb{F}_{Y(1) \mid X=x}\left(\mathbb{F}_{Y(0) \mid X=x}^{-1}(\alpha)\right)$,} $r_{\text{mle}}(x) = \frac{N_{\alpha}(x)}{N(x)} - \alpha$, $r^{\text{aff}}_{\text{mle}-h} = \max_{x \in U_{X}(S^T)} r_{\text{mle}}(x)$, $r^{\text{aff}}_{\text{mle}-l} = \min_{x\in U_{X}(S^T)} r_{\text{mle}}(x)$, $r^{\text{unaff}}_{\text{mle}-h} = \max_{x \not\in U_{X}(S^T)} r_{\text{mle}}(x)$, and 
  $\eta =  \left( \frac{\sum_{x \in U_{X}(S^T)}{N\left(x\right)} }{ \sum_{x \in U_{X}(D)}{N\left(x\right)} } \right)$. We also introduce the concepts of $\nu-homogeneous$, which means that $\frac{r^{\text{aff}}_{\text{mle}-h}}{r^{\text{aff}}_{\text{mle}-l}} < \nu$, and $\delta-strong$, which means that $ \frac{r^{\text{aff}}_{\text{mle}-l}}{r^{\text{unaff}}_{\text{mle}-h}} > \delta$. Intuitively, the concept of \emph{homogeneity} measures how similarly the
treatment affects each $\mathbb{F}_{Y| X=x}$ for $x \in U_{X}(S^T)$, while
\emph{strength} measures how large of an effect the treatment exhibits across
all $\mathbb{F}_{Y | X=x}$ for $x \in U_{X}(S^T)$.  More specifically, these concepts respectively imply that for any pair of the affected covariate profiles $\left(x_i, x_j \in U_{X}(S^T)\right)$, the anomalous signal (i.e., treatment effect)
observed in $x_i$ is less than $\nu$ times that which is observed in $x_j$, and the treatment effect observed in every affected covariate profile is more than $\delta$ times that of the unaffected profiles. Using these concepts we have the following results:

\begin{restatable*}{thm}{homo}
  \label{thm:subset_homo}
  Under $H_1(S^T)$ defined in \eqref{eq:test}, where $|U_{X}(S^T)|=t > 0$, $\exists~\nu > 1$ such
  that if the observed effect across the $t$ covariate profiles in $S^T$ is $\nu-homogeneous$, and at least $1$-strong, then the highest scoring subset $S^{\ast} \supseteq S^{T}$.
\end{restatable*}

\begin{restatable*}{thm}{strength}
  \label{thm:subset_strength}
   Under $H_1(S^T)$ defined in \eqref{eq:test}, where $|U_{X}(S^T)|=t > 0$, $\exists~\delta > 1$ such
  that if the observed effect across the $t$ covariate profiles in $S^T$ is $\frac{\delta}{\eta}-strong$, then the highest scoring
  subset $S^{\ast} \subseteq S^{T}$.
\end{restatable*}

\begin{restatable*}{thm}{assymptseteq}
  \label{thm:assympt_hom_str}
   Under $H_1(S^T)$ defined in, where $|U_{X}(S^T)|=t>0$, let $N(x) = n~ \forall x \in U_{X}(D)$. If $|U_{X}(D)| = M$ is fixed  then as $n \xrightarrow[]{} \infty$, $P(S^* = S^T) \xrightarrow[]{} 1$.
\end{restatable*}

While Theorem~\ref{thm:assympt_hom_str} shows that $S^\ast = S^T$ with high probability as $n\rightarrow\infty$---i..e, the most anomalous (rectangular) subset is the true subset---it does not guarantee that the TESS algorithm presented in Section~\ref{sec:algo} will converge to the true subset, because iterative ascent algorithms converge to a local maximum. For example, if Step 2(a) of the
algorithm chooses an initial subset $S_0$ that is disjoint from $S^T$, then it is possible that no maximization over subsets of values for any single mode in Step 2(b).i.A. will improve the score, and TESS will fail to identify $S^T$ \textit{on that iteration}.  However, we can show the following:

\begin{restatable*}{thm}{assymptTESS}
  \label{thm:assympt_TESS}
   Under $H_1(S^T)$ defined in \eqref{eq:test}, where $|U_{X}(S^T)|=t>0$ and $S^T \in Rect$, let $N(x) = n~ \forall x \in U_{X}(D)$. Assume $|U_{X}(D)| = M$ is fixed.  Let $\hat S^\ast$ denote the subset returned by a given iteration of the TESS algorithm, which was initialized to some subset $S_0 \in Rect$, such that $S^T \cap S_0 \ne \emptyset$.  Then as $n \xrightarrow[]{} \infty$, $P(\hat S^* = S^T) \xrightarrow[]{} 1$. 
\end{restatable*}

Thus as $n\rightarrow\infty$ for fixed $M$, TESS will identify the correct rectangular subset $S^T$ w.h.p., as long as it is initialized to some rectangular subset $S_0$ that overlaps $S^T$, for at least one of the its iterations. For example, $S_0$ could be the entire dataset $D$, thus guaranteeing $S_0 \cap S^T \ne \emptyset$.

Together, these results demonstrate that the test statistic $F^\ast = \max_{S \in Rect} F(S)$ and corresponding subset $S^\ast = \arg\max_{S \in Rect} F(S)$ possess desirable statistical properties. Theorems \ref{thm:false_posotive} and \ref{thm:power} imply that the asymptotic Type I and II errors of our procedure can be controlled, with implications for maximization over subsets of empirical processes more generally. Theorems \ref{thm:subset_homo} and \ref{thm:subset_strength} indicate that for a score function there exist constants $\nu$ and $\delta$\ignore{, both of which equal 2 for $F^{NA}$,}  that define how similar and strong the treatment effect must be in the affected subpopulation, to ensure that the highest-scoring subset corresponds exactly to the true affected subset $S^{\ast} = S^T$. Finally, Theorem \ref{thm:assympt_hom_str} shows asymptotic convergence for $P(S^{\ast} = S^T) \xrightarrow[]{} 1$ as $n\xrightarrow[]{}\infty$, and Theorem \ref{thm:assympt_TESS} shows that the TESS algorithm will identify $S^T$ w.h.p. as $n\rightarrow\infty$. To our knowledge, this is the first work on heterogeneous treatment effects that provides conditions on the exactness of subpopulation discovery.

\section{Related Work}
\label{sec:related}
\ignore{Inference on Heterogeneous Quantile Treatment Effects via Rank-Score Balancing, does HTE Quantile Effects}
There has been a growing literature using statistical learning methods to provide
data-driven approaches for estimating heterogeneous treatment effects in
randomized experiments. Recent work has adapted regularized regression for treatment effect heterogeneity~\citep{imai-hte_lasso-2013,tian-hte_lasso-2014,weisberg-subgroup-2015}. These regularized regression approaches, however, require the researcher to select which covariate and treatment
interactions to include in the model specification, compromising their ability to
\emph{discover} unexpected treatment patterns in subpopulations.

Regression tree based methods~\citep{su-subgroup-2009,athey-hte-2016} select subpopulations and estimate treatment
effects by recursively partitioning the
data into homogeneous subpopulations that share a subset of covariate profile
values and have similar outcomes. The effectiveness of tree methods can be severely compromised in many settings as a result of their greedy partitioning.\footnote{Though there are non-greedy tree based learning methods, these methods are used for optimal treatment assignments for individual units in observational data \citep{zhou-offline-policy-2018}.} Tree models
can be unstable; they can provide extremely discontinuous approximations of an
underlying smooth function, limiting overall accuracy; and they can struggle to
estimate functions which exhibit specific properties, including when a small
proportion of the covariates constitute the influential interactions
\citep{friedman-mars-1991}.

Other treatment
effect estimation approaches use ensemble methods, including the use of Bayesian Additive Regression Trees~\citep{hill-hte_bart-2011, green-hte_bart-2012},
Random Forests~\citep{foster-hte_randomforest-2011, wager-causalforest-2018}, and ensembles of strong learners~\citep{grimmer-hte_ensemble-2017}. Ensembles provide more stable and smooth function estimates
\citep{wager-causalforest-2018}; however, they lose the interpretability of
natural groupings (e.g., specific combinations of covariates or clearly defined
leaves) which is important for \textit{identifying} affected subpopulations.

Finally,~\cite{chernozhukov-generic-2018} propose approaches to test if there is detectable heterogeneity in treatment effects, finding the quantiles exhibiting heterogeneous treatment effects induced by the machine learning proxy predictors, and then identifying the covariates that appear associated with the heterogeneity. Therefore this approach is a post-hoc analysis of existing machine learning predictors (e.g., regression tree or ensemble methods) which are optimized for overall risk minimization and not the discovery of subpopulations with significant treatment effects. Furthermore, there are orthogonal research streams focused on (heterogeneous) treatment assignment that include policy learning~\citep{zhou-offline-policy-2018, athey-policy-learning-2020} and welfare maximization~\citep{toru-welfare-max-2018, mbakop-welfare-max-2020} in observational and experimental data. These methods assign a personalized treatment for each individual unit. Instead, TESS takes as input a set of units which have already been randomly assigned a treatment.

Although this literature contains a growing set of novel statistical learning methods for causal inference, at the core of the majority of these approaches are objective functions designed for flexible estimation (and risk minimization) instead of subpopulation discovery (and significant effect maximization). TESS is therefore unique as it is optimized to discover interpretable subpopulations that exhibit significant evidence of treatment effects. When necessary, TESS can use flexible (risk minimizing) statistical learning models for a purpose that is aligned with their objective: as an accurate estimator of the conditional outcome distribution for the control group as in Section~\ref{sec:model-based estimation}.

\section{Empirical Analysis}
\label{sec:results}
In this section we empirically demonstrate the utility of the TESS framework as
a tool to identify subpopulations with significant treatment effects. We use
data from the Tennessee Student/Teacher Achievement Ratio (STAR)
randomized experiment~\citep{word-star-1990} in order to provide representative performance in
real-world policy analysis. We review the original STAR data (\S\ref{sec:star-data}), and describe our procedure for simulating affected
subpopulations (\S\ref{sec:simulations}).

Through the simulation results described in \S\ref{sec:sim-results}, we
compare the ability of TESS to detect significant subpopulations to three
recently proposed statistical learning approaches: Causal
Tree \citep{athey-hte-2016}, Interaction
Tree \citep{su-subgroup-2009,athey-hte-2016}, and Causal Forest
\citep{wager-causalforest-2018}. Specifically, we evaluate each method on two
general metrics: detection power and subpopulation accuracy. Detection power
measures $P_{H_1}(Reject~H_0)$, or how well a method can detect the
existence--not necessarily the location--of treatment effect heterogeneity in
the experiment. Subpopulation accuracy, on the other hand, is specifically designed to measure
how well a method can precisely and completely capture the subpopulation(s) with
significant treatment effects.

Finally, we conduct an exploratory analysis of the STAR dataset, and in
\S\ref{sec:star-eda} discuss the subpopulations identified by TESS as
affected by treatments. In some cases, the identified subpopulation
is consistent with the literature on the STAR experiment; in other cases,
TESS uncovers previously unreported, but intuitive and believable,
subpopulations. These empirical results demonstrate TESS's potential to
generate potentially useful and non-obvious hypotheses for further exploration and testing.

\subsection{Tennessee STAR Experiment}
\label{sec:star-data}

The Tennessee Student/Teacher Achievement Ratio (STAR) experiment is a
large-scale, four-year, longitudinal randomized experiment started in 1985 by the
Tennessee legislature to measure the effect of class size on student
educational outcomes, as measured by standardized test scores.
The experiment started monitoring students in kindergarten (during the
1985-1986 school year) and followed students until third grade. Students and
teachers were randomly assigned into conditions during the first school year,
with the intention for students to continue in their class-size condition
for the entirety of the experiment. The three potential experiment conditions were
not based solely on class size, but also the presence of a full-time teaching
aide: small classrooms (13-17 pupils), regular-size classrooms (22-25 pupils),
and regular-size classrooms with aide (still 22-25 pupils). Therefore, the
difference between the former two conditions is classroom size, and the difference
between the latter two conditions is the inclusion of a full-time teacher's aide
in the classroom. The experiment included approximately 350 classrooms from 80 schools,
each of which had at least one classroom of each type.
Each year more than 6,000 students participated in this experiment, with the
final sample including approximately 11,600 unique students.

The Tennessee STAR dataset has been well studied and analyzed, both by the project's
internal research team~\citep{word-star-1990, folger-star-1989}
and by external researchers~\citep{krueger-star-1999,nye-star-2000}. As
indicated by \cite{krueger-star-1999}, the investigations have primarily focused
on comparing means and computing average treatment effects. \cite{krueger-star-1999}
presents a detailed econometric analysis and draws similar
conclusions to the previous research: students in small classrooms perform
better than those in regular classrooms,
while there is no significant effect of a full-time teacher's aide, or moderation from teacher
characteristics. Moreover, the effect accumulates each year a student spends in a
small classroom \citep{krueger-star-1999}. Additionally, these conclusions are robust in the presence of potentially compromising experimental design challenges: imbalanced
attrition, subsequent changes in original treatment assignment, and fluctuating
class sizes~\citep{krueger-star-1999}.

\subsection{Experimental Simulation Setup}
\label{sec:simulations}

The goal of our experimental simulation is to replicate conditions under
which a researcher would want to use an algorithm to discover subpopulations
with significant treatment effects, and to observe how capable various
algorithms are at identifying the correct subpopulation(s). In order to replicate
realistic conditions, we use the STAR experiment as our base dataset, and inject
into it subpopulations (of a given size) with a treatment effect (of a given
magnitude). More specifically, we treat each student-year as a unique record and
for each record capture ten covariates: student gender, student ethnicity,
grade, STAR treatment condition, free-lunch indicator, school development
environment, teacher degree, teacher ladder, teacher experience, and teacher
ethnicity. We note that each of these variables, other than teacher experience,
is discrete; we discretize experience into five-year intervals:
$[0,5), [5,10), \ldots, [30, \infty)$. The number of values a covariate can take
ranges from two to eight. By preserving the overall data
structure of the STAR experiment--number of covariates, covariate value
correlations, subpopulations, sample sizes, etc.--our simulations are more able to
replicate the structure (and challenges) faced by experimenters.

The process we follow to generate a simulated treatment effect begins with
selecting a subpopulation $S_{\text{affected}}$ to affect. Recall that the dataset
contains a set of discrete covariates $X = (X^{1}, \ldots, X^{d}$), where each
$X^{j}$ can take on a vector of values $V^{j}=\{v^{j}_{m}\}_{m = 1...|V^j|}$ and
$|V^{j}|$ is the arity of covariate $X^{j}$. Therefore, we define a
subpopulation as $S = v^1 \times \ldots \times v^d$, where $v^j \subseteq
V^j$. The affected subpopulation is generated at random based on two parameters:
$num\_covs$, or the number of covariates to select, and $value\_prob$, or
probability a covariate value is selected. We select $num\_covs$ covariates at
random, and for each of these covariates we select each of their values with
probability equal to $value\_prob$, ensuring that at least one value for each of
these covariates is selected. The final affected subpopulation is then
$S_{\text{affected}} = v^1 \times \ldots \times v^d$, where $v^j$ is the
selected values if $X^{j}$ is one of the $num\_covs$ covariates, and otherwise
$v_j = V^j$. In other words, for a random subset of covariates,
$S_{\text{affected}}$ only includes a random subset of their values, and for all
other covariates $S_{\text{affected}}$ includes all of their values. This treatment
effect simulation scheme allows for variation in the size of the subpopulation
that is affected: instances of $S_{\text{affected}}$ can constitute a
small subpopulation (a challenging detection task), a large subpopulation (a
relatively easier detection task), or something in between. Therefore a
set of simulations, with varying parameter values, captures the spectrum of
conditions a researcher may face when analyzing an experiment to identify subpopulations
with significant treatment effects.

The next step in the process involves partitioning the dataset into treatment and control groups, and generating outcomes for each
record. Outcomes are drawn randomly from one of two distributions: the null
distribution ($f_0$) or the alternative distribution ($f_1$). Any record in the
treatment group that has a covariate profile $x \in U_{X}\left(S_{\text{affected}}\right)$ has
outcomes generated by $f_1$; all other records have outcomes drawn from
$f_0$. Therefore only $S_{\text{affected}}$ has a treatment effect, whose
effect magnitude is the distributional difference
between $f_0$ and $f_1$, represented by the parameter $\delta$.

Each of the methods we consider in these experiments has a unique approach to
identifying potential subpopulations with differential treatment effects.
Furthermore, as mentioned in \S\ref{sec:related}, most methods in the
literature do not provide a process for identifying extreme treatment effects.
Therefore, we devise intuitive post-processing steps in an attempt to represent how
researchers would use each method to identify potential subpopulations
that have significant treatment effects. Each method returns identified
subpopulations and corresponding scores (measures) of the treatment effect. For
the single tree-based methods \citep{athey-hte-2016,su-subgroup-2009} we follow
the suggestion of \cite{athey-hte-2016} to perform inference (via a two-sample
Welch T-Test) in each leaf of the tree, and we then sort the leaves based on
their statistical significance. The final subpopulation returned by the tree is
the leaf with the most statistically significant treatment effect, and the final
treatment effect measure is this leaf's statistical significance ($p$-value). 
For a method that provides an individual level treatment effect (and
estimate of variance) \citep{wager-causalforest-2018}, we propose to perform
inference for each unique covariate profile, and return those that are
statistically significant. The final treatment effect measure is the smallest
$p$-value of the covariate profiles. We empirically compare these prominent methods from the literature to our TESS algorithm, selecting the Berk-Jones nonparametric scan statistic to score a subset (\S\ref{sec:hetero_quant_effects}) and the most general empirical estimation approach to model the reference distribution (\S\ref{sec:sec:ref-dist-emp}). The TESS algorithm, by design, provides the
subpopulation it determines to have a statistically significant distributional
change (treatment effect) and a measure of this change, so no
post-processing is necessary.

\subsubsection{Detection Power}
For any given combination of
simulation parameter values ($\delta, num\_covs, value\_prob$), detection power
measures $P(Reject~H_0 \mid H_1(S_{\text{affected}}))$, or how
well a method is able to identify the presence of $S_{\text{affected}}$. This is
accomplished by comparing the treatment effect measure (score of the detected subset) found
under $H_1(S_{\text{affected}})$ to the distribution of the treatment effect
measure under $H_0$. More specifically, for a given set of parameter values, we
generate a random dataset which only exhibits a treatment effect in the randomly
selected subpopulation $S_{\text{affected}}$; each method attempts to detect
this subpopulation. As described in \S\ref{sec:simulations}, each method returns
a final treatment effect measure for the subpopulation it detects in this
affected dataset. For the same dataset, we then conduct randomization
testing to determine how significant this treatment effect measure is under
$H_0$. We make many copies of the dataset (1000 in our
experiments) and in each copy, we generate new outcomes (drawn from $f_0$) such
that no subpopulation has a treatment effect. Each method then generates a
detected subpopulation and corresponding treatment effect measure for each of
these null datasets. These treatment effect measures from the null datasets
together provide an empirical estimate of the distribution of the treatment
effect measure under $H_0$ for that method. Subsequently, a $p$-value is computed for the
treatment effect measure captured under $H_1(S_{\text{affected}})$. This process
is repeated many times (300 in our experiments), where each time we
1) generate a random $S_{\text{affected}}$, 2) generate a random dataset under $H_1(S_{\text{affected}})$ and
compute each method's treatment effect measure, and 3) generate 1000
copies of the dataset with no treatment effect to compute each
method's treatment effect measure distribution under $H_0$. This process creates 300 $p$-values
for each method which describe how extreme each of
the $S_{\text{affected}}$ appear under $H_0$. A method rejects $H_0$ for a given
$p$-value if it is less than or equal to some test-level $\gamma$, corresponding to the
$1-\gamma$ quantile of the null distribution ($\gamma = 0.05$ in our experiments). Therefore, the detection power $P\left(Reject~H_0 \mid
H_1(S_{\text{affected}})\right)$ is captured as the proportion of $p$-values
that are sufficiently extreme that they lead to the rejection of $H_0$ at level $\gamma$.

\subsubsection{Detection Accuracy}
While
detection power measures how well a method identifies the presence of a
subpopulation with a treatment effect $S_{\text{affected}}$, as compared to datasets with no
treatment effect, detection accuracy measures how well a method can precisely and completely identify the affected subpopulation
$S_{\text{affected}}$. Accurately identifying in which subpopulation(s) a
treatment effect exists can be crucial, particularly when there is no prior
theory to guide which subpopulations to inspect, or when the goal itself is to
develop intuition for new theory. As described in \S\ref{sec:simulations}, each
of the methods we consider is able to return the subpopulation that it
determines as having the most statistically significant treatment effect
$S_{\text{detected}}$. Each method will pick out a set of covariate profiles, which
could have coherent structure (as with TESS, Causal Tree, and
Interaction Tree), or be an unstructured collection of individually significant covariate profiles
(as with Causal Forest). To
accommodate both types of subpopulations, we therefore define detection accuracy
as
\begin{equation}
  \label{eqn:accuracy}
  \begin{split}
    \text{accuracy} = \frac{|S_{\text{detected}}\:\cap\: S_{\text{affected}}|}
                            {|S_{\text{detected}}\:\cup\: S_{\text{affected}}|}
                    = \frac{\sum_{R_i} \mathbbm{1}_{\{R_i \in S_{\text{detected}} \cap S_{\text{affected}}\}}}
                            {\sum_{R_i} \mathbbm{1}_{\{R_i \in S_{\text{detected}} \cup S_{\text{affected}}\}}}.
  \end{split}
\end{equation}
where $R_i$ are records in the treatment group. This definition
of accuracy, commonly known as the Jaccard coefficient, is intended to balance
precision (i.e., what proportion of the detected subjects truly have a treatment
effect) and recall (i.e., what proportion of the subjects with a treatment effect are correctly
detected). We note that $0 \le \text{accuracy} \le 1$; high accuracy values
correspond to a detected subset $S_{\text{detected}}$ that captures many of the
subjects with treatment effects and few or no subjects without treatment effects.

\subsection{Simulation Results}
\label{sec:sim-results}
Our first set of results involve a treatment effect that is a mean shift in a
normal distribution: the null distribution $f_0 = N(0,1)$ and the alternative
$f_1= N(\delta,1)$, where $\delta$ captures the magnitude of the signal
(treatment effect). Recall from \S\ref{sec:simulations} that there are three
parameters that we can vary to change the size and magnitude of the signal. For
our simulation, we specifically consider $\delta \in \{0.25, 0.5, \ldots, 3.0\}$,
$num\_covs \in \{1,2, \ldots, 10\}$, and $value\_prob \in \{0.1, 0.2, \ldots,
0.9\}$; the former controls magnitude of the treatment effect, while the latter
two control the concentration of the treatment effect (i.e., the expected size
of the affected subpopulation). Instead of considering every combination, we
select the middle value of each parameter interval as a reference point
($\delta= 1.5, num\_covs = 5, value\_prob = 0.5$) and measure performance
changes for one parameter, while keeping the others fixed.

Figure \ref{fig:parametric-detection-power} shows the changes in each method's
detection power performance as we vary each of the three parameters that
contribute to the strength of the treatment effect. From each of the three
graphs we observe that TESS consistently exhibits more power than (or equivalent
to) the other methods. More importantly, TESS exhibits statistically significant
improvements in power for the most challenging ranges of parameter values (i.e., more subtle signals). The top plot varies effect size (or $\delta$),
which is positively associated with signal strength and negatively
associated with detection difficulty; for values $2.0$ and below TESS has
significantly higher detection power than the competing methods. The middle plot varies the
number of covariates selected to have only a subset of values be affected
($num\_covs$). This parameter is negatively associated with signal strength
and positively associated with detection difficulty; for values $5$
and above, TESS has significantly higher detection power. The
bottom plot varies the expected proportion of values, for the selected
covariates, which will be affected ($value\_prob$). This parameter is
positively associated with signal strength and negatively associated
with detection difficulty; for values $0.5$ and below TESS exhibits
significantly higher detection power. We see that, for sufficiently strong signals
(based on both signal magnitude and concentration), all methods are able to
distinguish between experiments with and without a subpopulation exhibiting a
treatment effect, while TESS provides significant advantages in detection power for weaker signals.
\begin{figure}
  \captionsetup{font=scriptsize,skip=0pt}
  \begin{subfigure}{.5\textwidth}
    \centering
    \includegraphics[width=.99\linewidth]{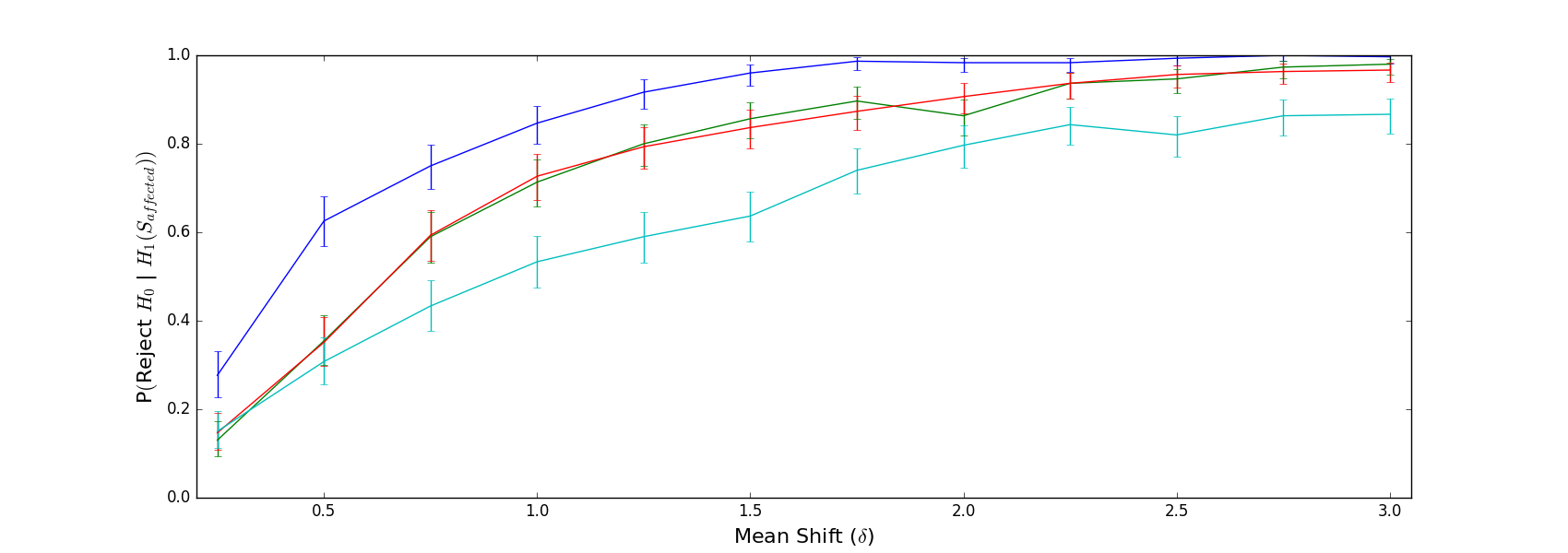}

    \includegraphics[width=.99\linewidth]{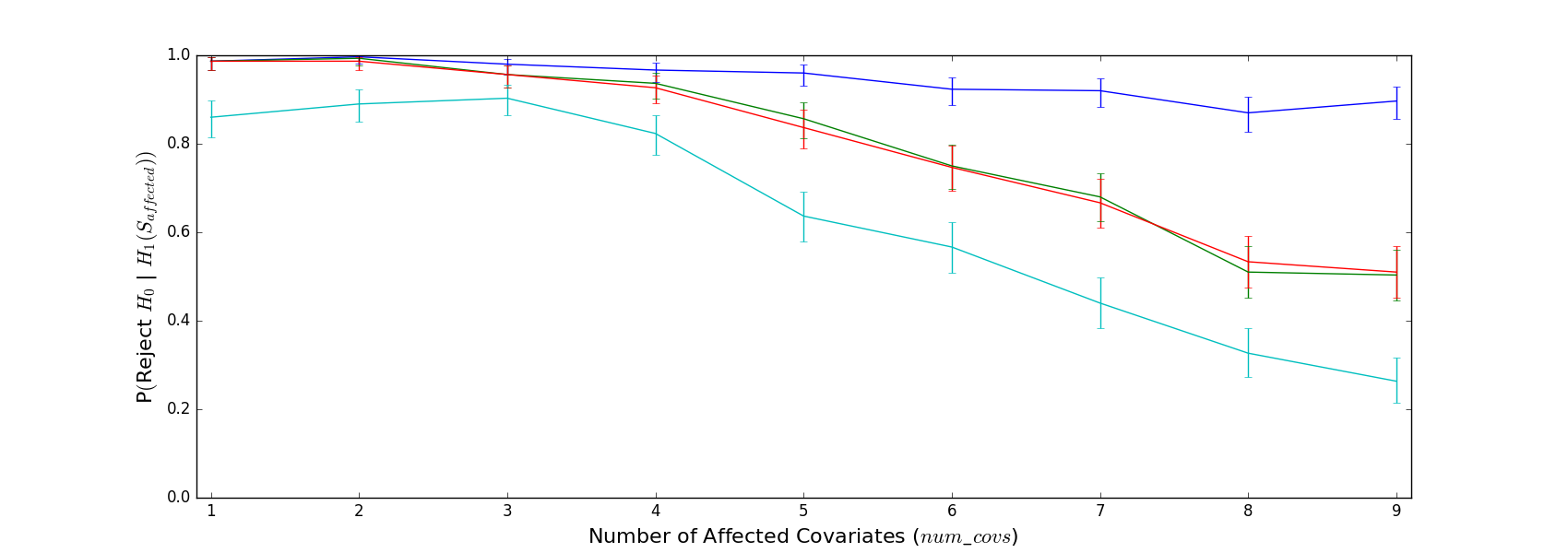}

    \includegraphics[width=.99\linewidth]{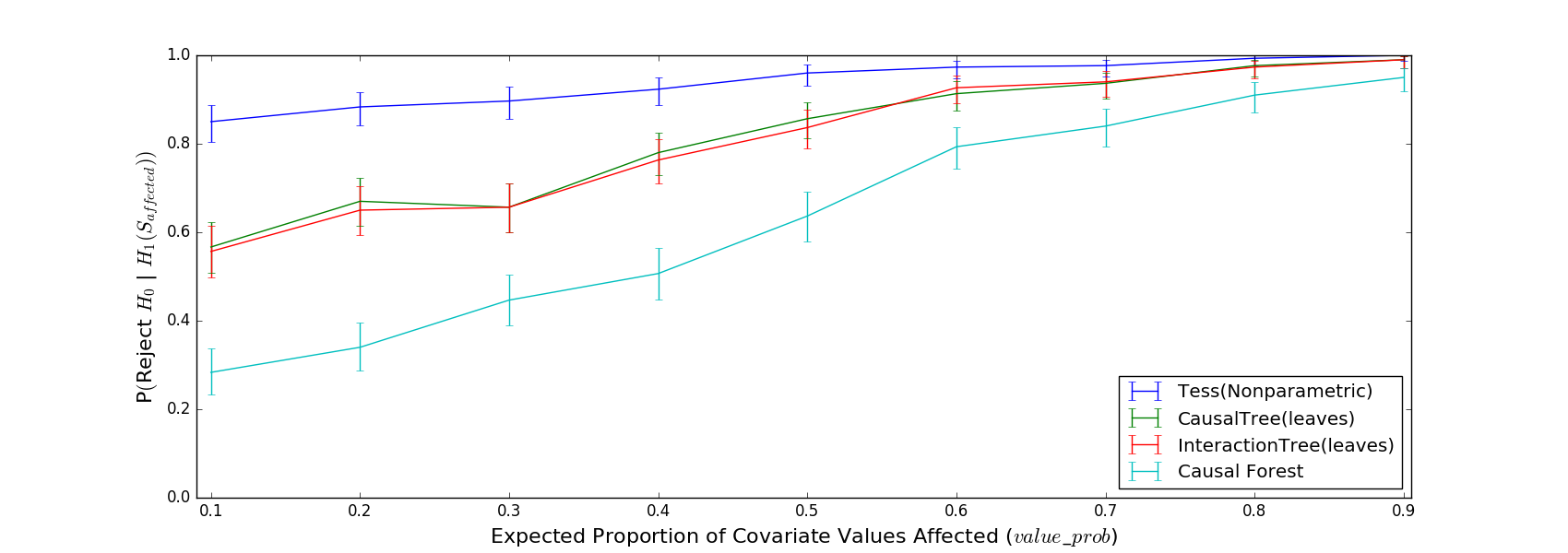}
    \caption{Detection Power}
    \label{fig:parametric-detection-power}
  \end{subfigure}%
  \begin{subfigure}{.5\textwidth}
    \centering
    \includegraphics[width=.99\linewidth]{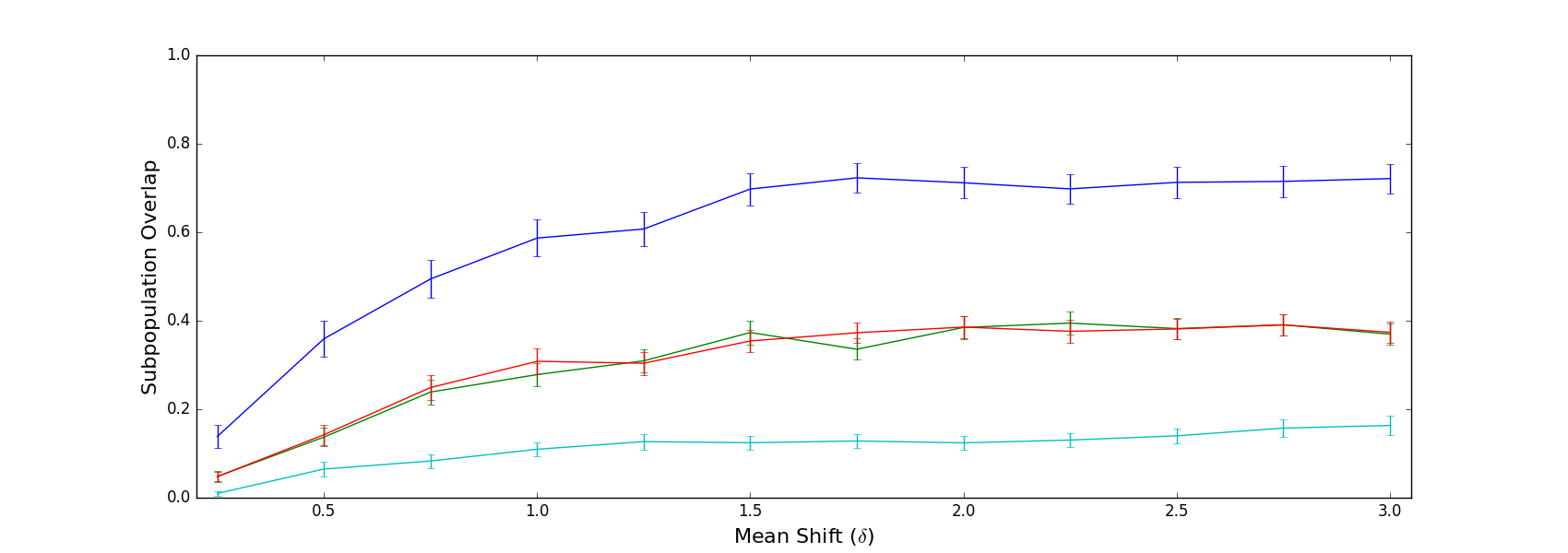}

    \includegraphics[width=.99\linewidth]{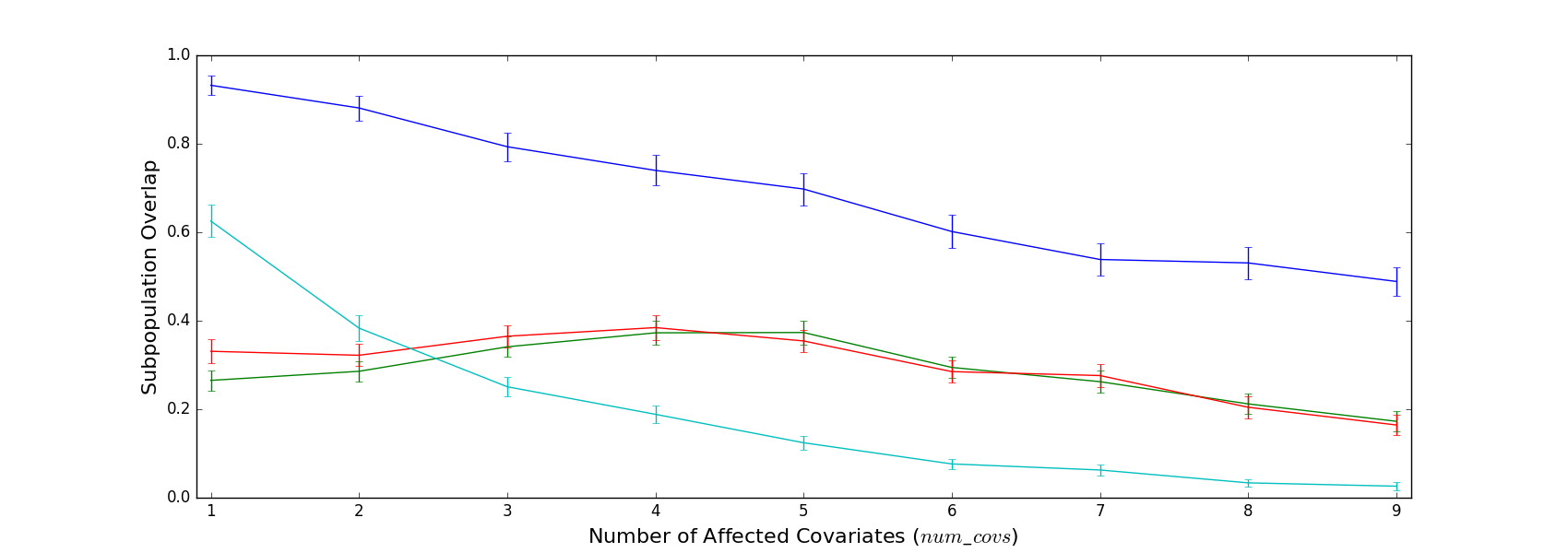}

    \includegraphics[width=.99\linewidth]{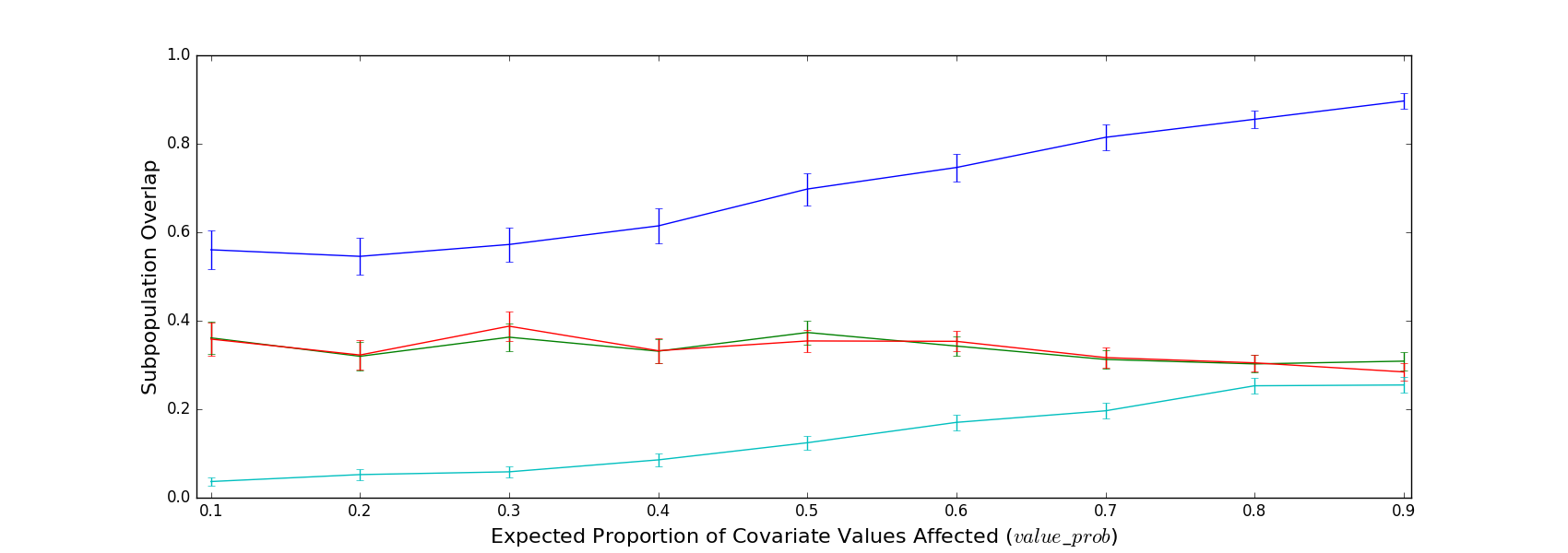}
    \caption{Detection Accuracy}
    \label{fig:parametric-detection-accuracy}
  \end{subfigure}
  \captionsetup{font=footnotesize}
  \caption{Ability of each method to identify subpopulations with mean shift
  treatment effects. The three parameters start as fixed ($\delta=1.5,
  num\_covs=5, value\_prob=0.5$) and then are varied individually to see how
  detection ability varies.}
  \label{fig:parametric-detection-all}
  \begin{subfigure}{.5\textwidth}
    \centering
    \includegraphics[width=.99\linewidth]{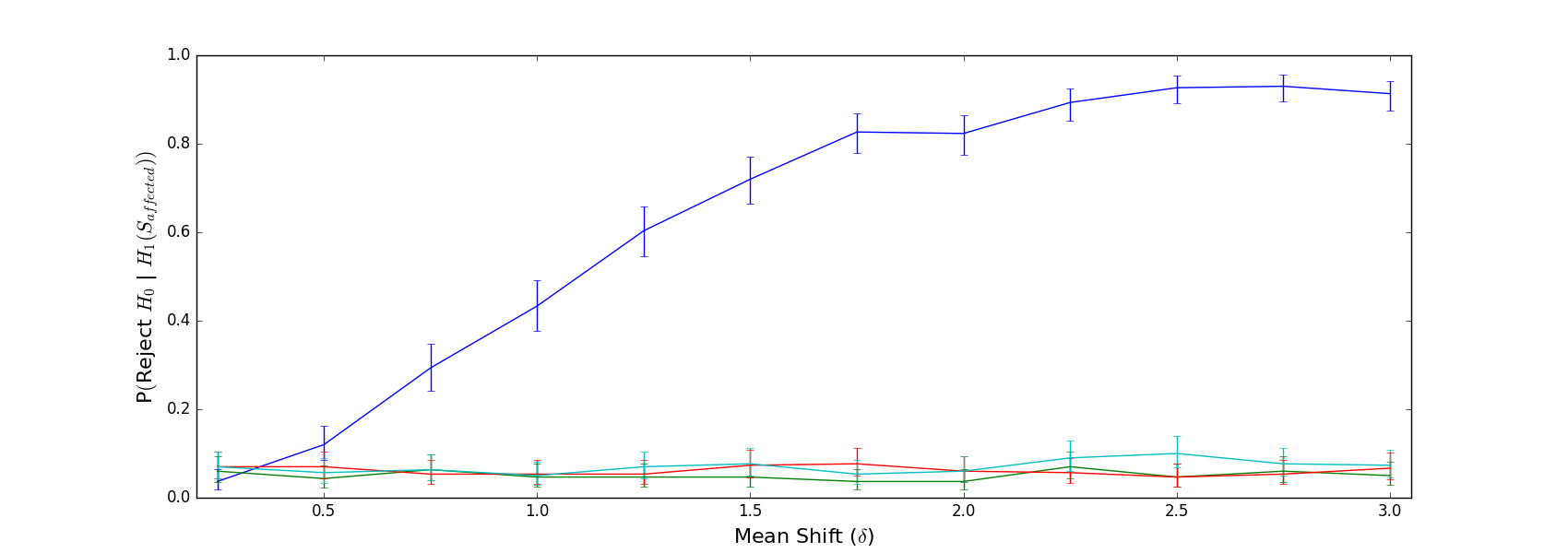}

    \includegraphics[width=.99\linewidth]{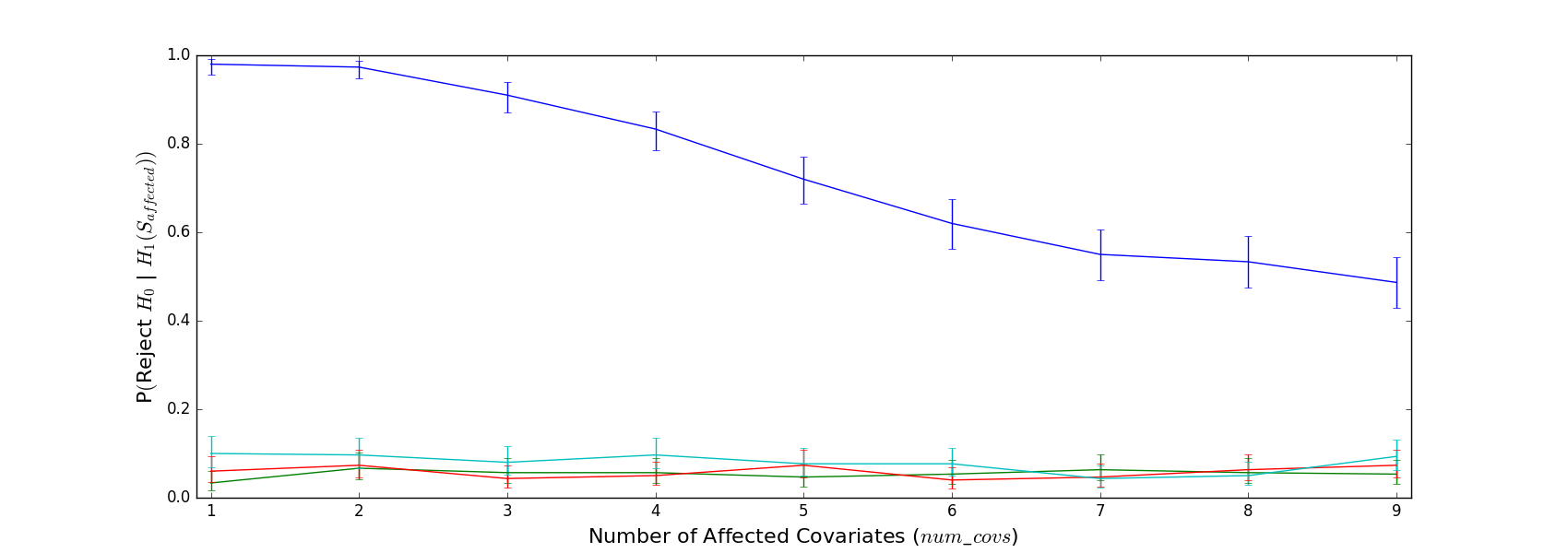}

    \includegraphics[width=.99\linewidth]{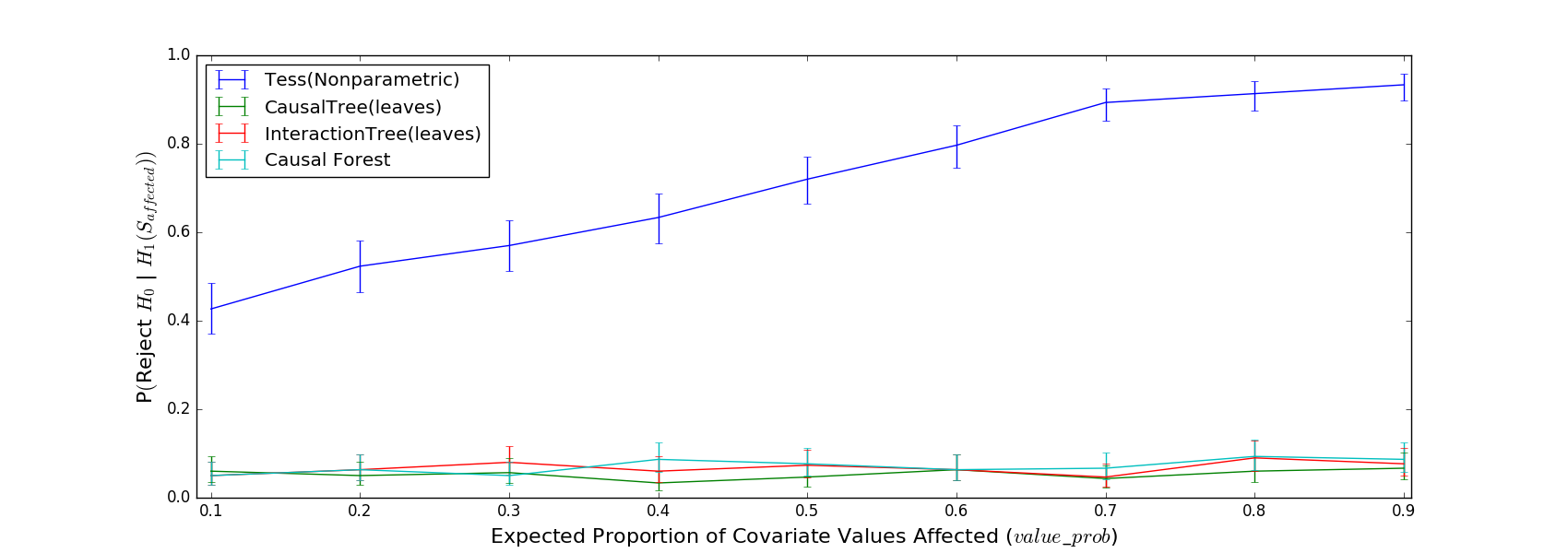}
    \caption{Detection Power}
    \label{fig:nonparametric-detection-power}
  \end{subfigure}%
  \begin{subfigure}{.5\textwidth}
    \centering
    \includegraphics[width=.99\linewidth]{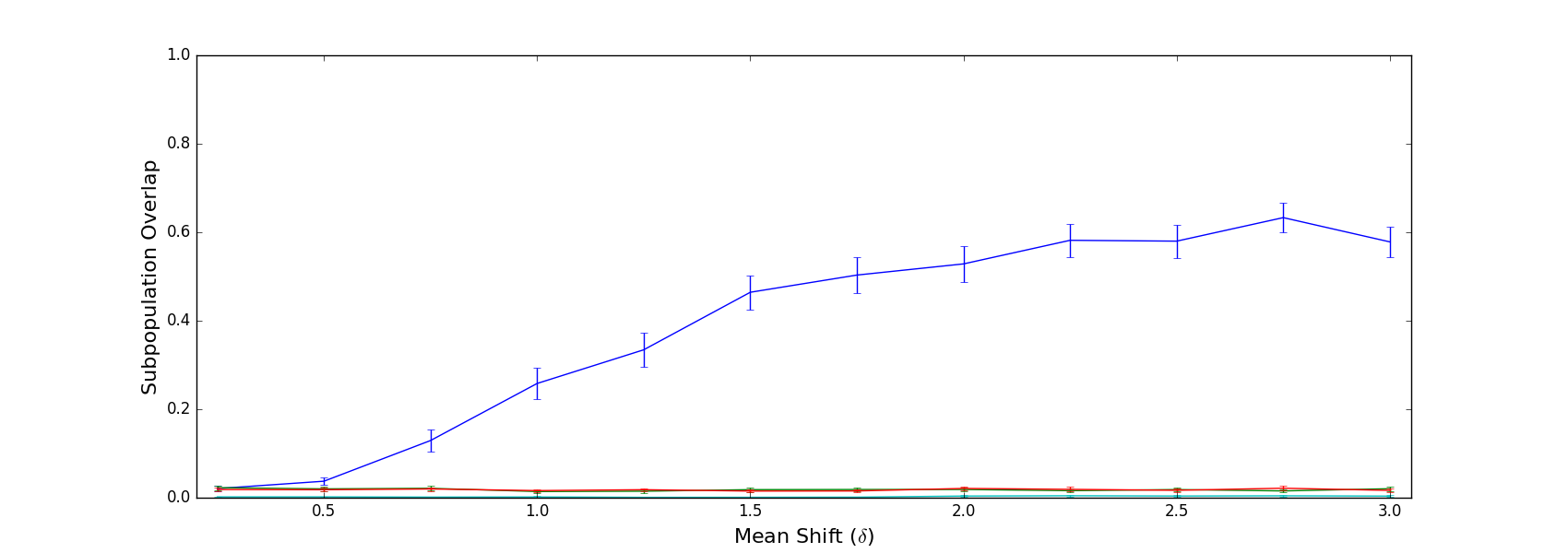}

    \includegraphics[width=.99\linewidth]{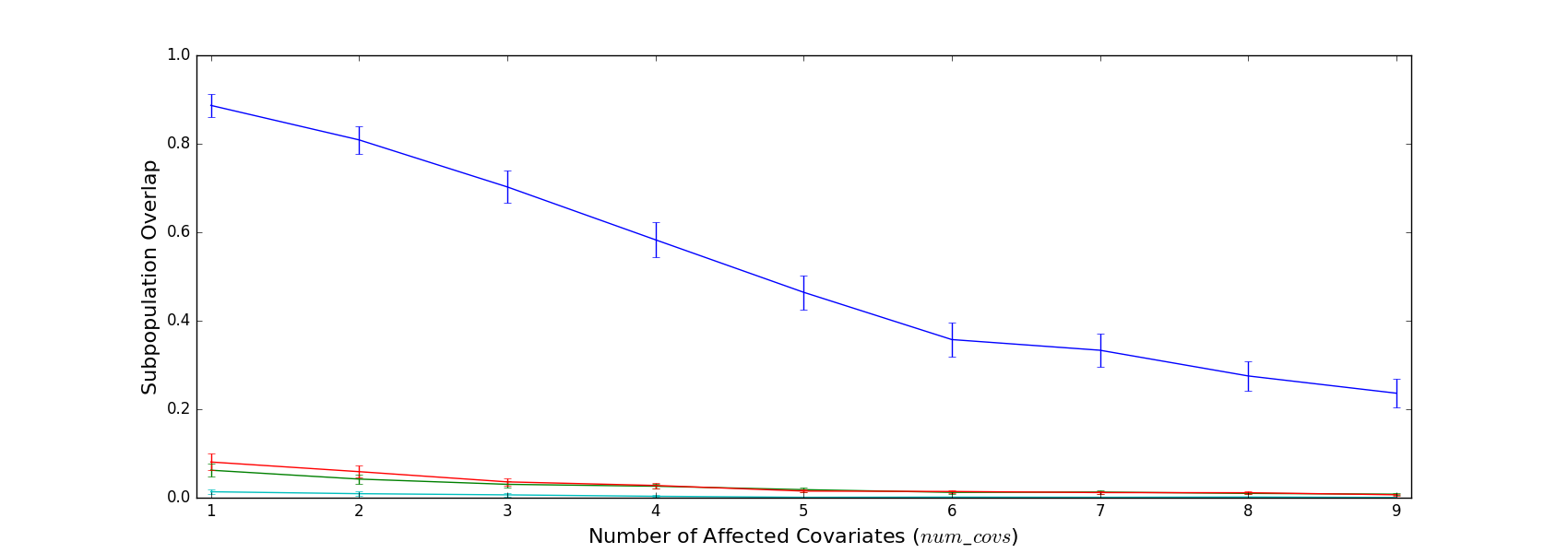}

    \includegraphics[width=.99\linewidth]{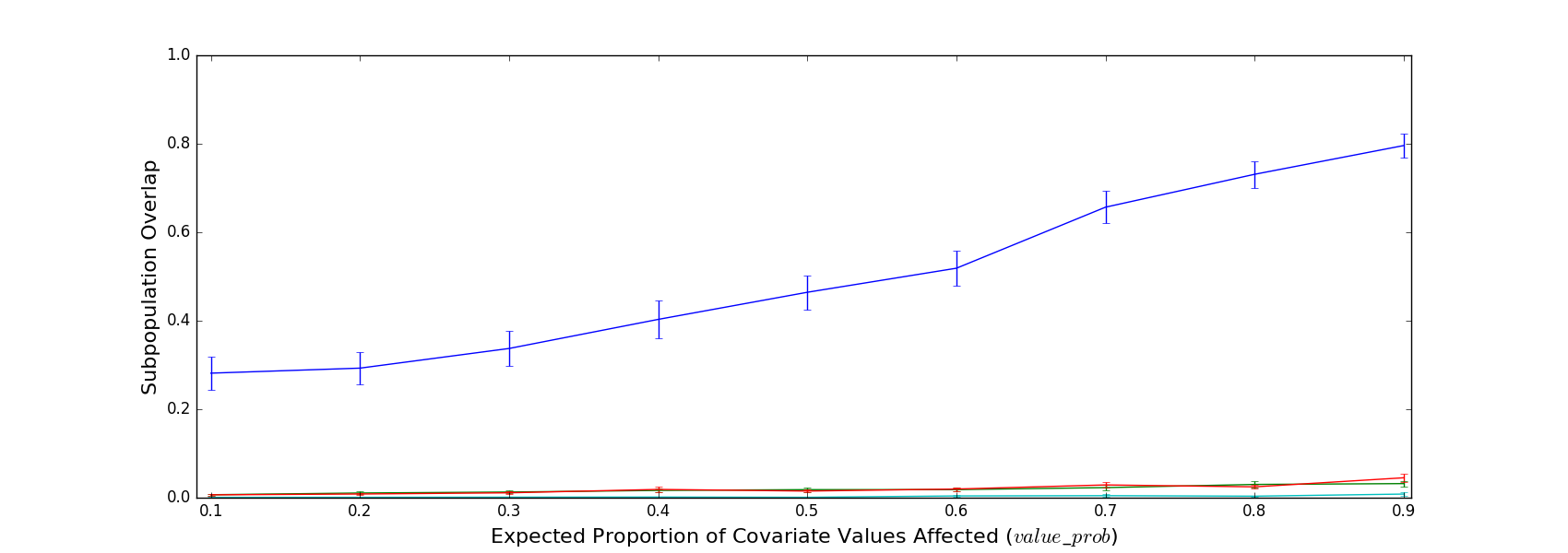}
    \caption{Detection Accuracy}
    \label{fig:nonparametric-detection-accuracy}
  \end{subfigure}
  \caption{Ability of each method to identify subpopulations with an unaffected
  mean, but distributional treatment effect. The three parameters start as
  fixed ($\delta=1.5, num\_covs=5, value\_prob=0.5$) and then are varied
  individually to see how detection ability varies.}
  \label{fig:nonparametric-detection-all}
\end{figure}

Figure \ref{fig:parametric-detection-accuracy} shows the changes in each
method's detection accuracy as we vary each of the three parameters
that contribute to the strength of the treatment effect. From each of the three
graphs we observe that TESS consistently exhibits significantly
higher accuracy than any other method. Recall that we measure subpopulation
accuracy as in \eqref{eqn:accuracy}, which captures both precision and recall of
the subpopulation returned by a method. The single tree methods tend to have
high precision but low recall, resulting in compromised overall accuracy.
Intuitively, these results indicate that the truly affected subpopulation is
being spread over multiple leaves of the tree, despite its goal of
partitioning the data into subpopulations with similar outcomes. This phenomenon may be
caused by the greedy search aspect of tree learning: if the tree splits the affected subpopulation between two branches of the tree, the recall of
any leaf will be compromised, especially when this split occurs close to the
root of the tree. The
Causal Forest ensemble method, on the other hand, exhibits relatively higher recall than precision.
These results indicate that it is difficult for Causal Forest to
distinguish between the covariate profiles that do and do not make up the truly
affected subpopulation, as profiles from both sets appear to have statistically
significant treatment effects. This inability stems from the fact that ensemble
methods are designed to provide individual level predictions, therefore their
conclusions regarding the statistical significance of a covariate profile are made
in isolation from the other covariate profiles that also make up the affected
subpopulation. Unlike single-tree methods, ensemble methods do not provide
coherent and natural groupings of subpopulations. TESS, however, does provide a
coherent subpopulation, which seems to balance precision and recall, maintaining
a significantly higher subpopulation accuracy.

It is also important to note that the data generating process for these simulations
(a treatment effect that occurs as a mean shift between treatment and control distributions) corresponds to the modeling assumptions of the current methods in the literature,
which specifically attempt to detect mean shifts, while TESS is designed to detect more general distributional changes. TESS's improved performance, as compared to the competing methods, in these adverse conditions may be due to its subset-scanning based approach, which combines information across groups of data in an attempt to find exactly and only the affected subset of data. Even if each individual covariate profile that is truly affected exhibits small evidence of a treatment effect, TESS can leverage the group
structure and signal of all the affected covariate profiles, and correctly conclude that collectively the subpopulation exhibits significant evidence of a
treatment effect. Additionally, the fact that TESS executes its optimization iteratively, unlike the greedy search of tree-based methods, enables it to rectify
initial choices of subset that are later determined to be inferior.

Our second set of results considers treatment effects that do not
align with the mean shift assumption that pervades the literature. Therefore,
the null distribution is still $f_0 = N(0,1)$; however, the alternative is a
mixture distribution $f_1= \frac{1}{2}N(-\delta,1)+\frac{1}{2}N(\delta,1)$.
Here $\delta$ still captures the magnitude of the signal (treatment effect),
and the remainder of the simulation process remains unchanged. This mixture
distribution alternative, however, changes the detection task
dramatically: while the average treatment effect is zero, there is still a
clear difference in the outcome distribution between treated and control individuals.

Figure \ref{fig:nonparametric-detection-all} shows how each method's detection
power and accuracy change as we vary each of three parameters that contribute to the strength
of the treatment effect. If we compare these simulations to those above with a
mean shift, TESS exhibits a consistent pattern of high performance, while the
performance of the competing methods is dramatically lower. The
detection power results indicate that, for the competing methods, it is hard to
distinguish even strong distributional changes from random chance, while the accuracy
results indicate that their pinpointing of the affected subpopulation is little better
than random guessing. Given that there is no observable mean shift in these simulations, these results
are consistent with what we expect: TESS is designed to identify more general distributional changes, while
the other methods are unable to identify distributional changes without corresponding mean shifts.

\subsection{A Case Study on Identifying Subpopulations: Tennessee STAR}
\label{sec:star-eda}

There appears to be a consensus in the literature that the presence of a teaching aide in a
regular-size classroom has an insignificant effect on test
scores~\citep{word-star-1990,krueger-star-1999,folger-star-1989,stock_watson-econ-2nd}.
(One significant effect was observed in first grade, but this effect was
largely considered to be a false positive.) Therefore, we want to use TESS to
compare regular classrooms with an aide to regular classrooms without an aide, to determine if
there appears to be a subpopulation that was significantly and positively affected by the
treatment. To do so, we replicated the analysis of the internal
STAR team, using TESS to extend the results, with the goal of demonstrating what
the STAR team could have surmised with present-day tools for uncovering
heterogeneity. We replicate the original STAR analysis from \citep{word-star-1990,stock_watson-econ-2nd}
which includes the sum of the Stanford math and reading scores as the outcome of
interest. For the data provided to TESS for detection, we combine the panel
data across years and include student's grade level as a covariate.
\begin{figure}
  \centering
  \captionsetup{font=scriptsize,skip=0pt}
  \begin{subfigure}{.5\textwidth}
    \centering
    \includegraphics[width=.75\linewidth]{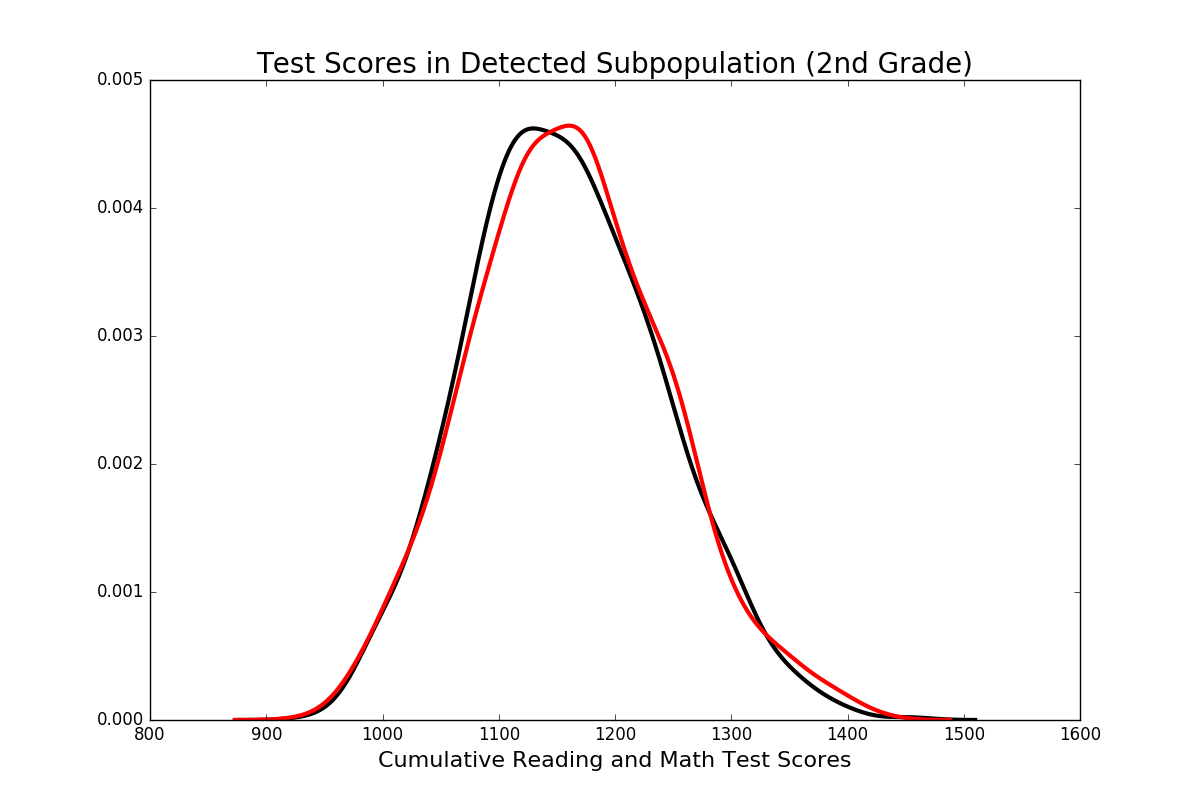}
    \caption{All students in 2nd grade}
    \label{fig:density-2nd-all}
  \end{subfigure}%
  \begin{subfigure}{.5\textwidth}
    \centering
    \includegraphics[width=.75\linewidth]{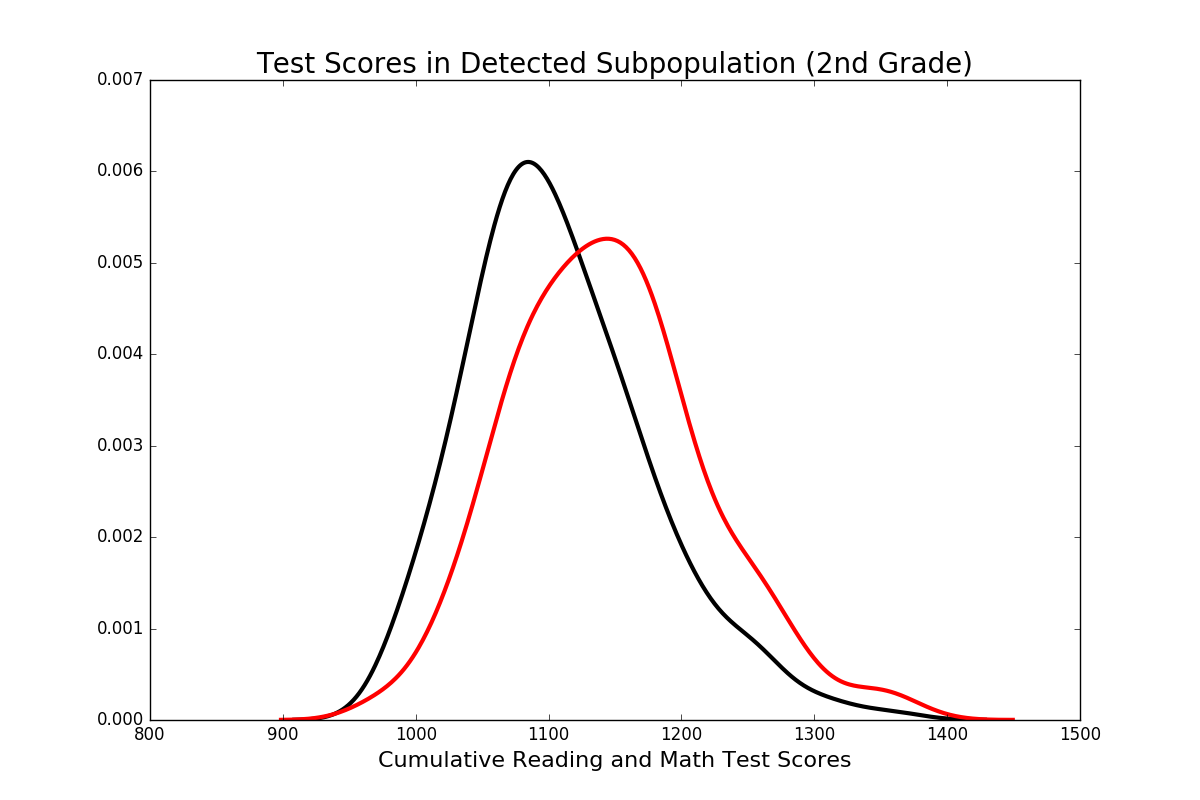}
    \caption{Detected Subpopulation in 2nd grade}
    \label{fig:density-2nd-subpop}
  \end{subfigure}
  \caption{Kernel density plots of 2nd grade test scores for treatment students
  (red) who were in a regular classroom with a teacher's aide and control
  students (black) who did not have a teacher's aide.}
  \label{fig:density-2nd}

  \begin{subfigure}{.5\textwidth}
    \centering
    \includegraphics[width=.75\linewidth]{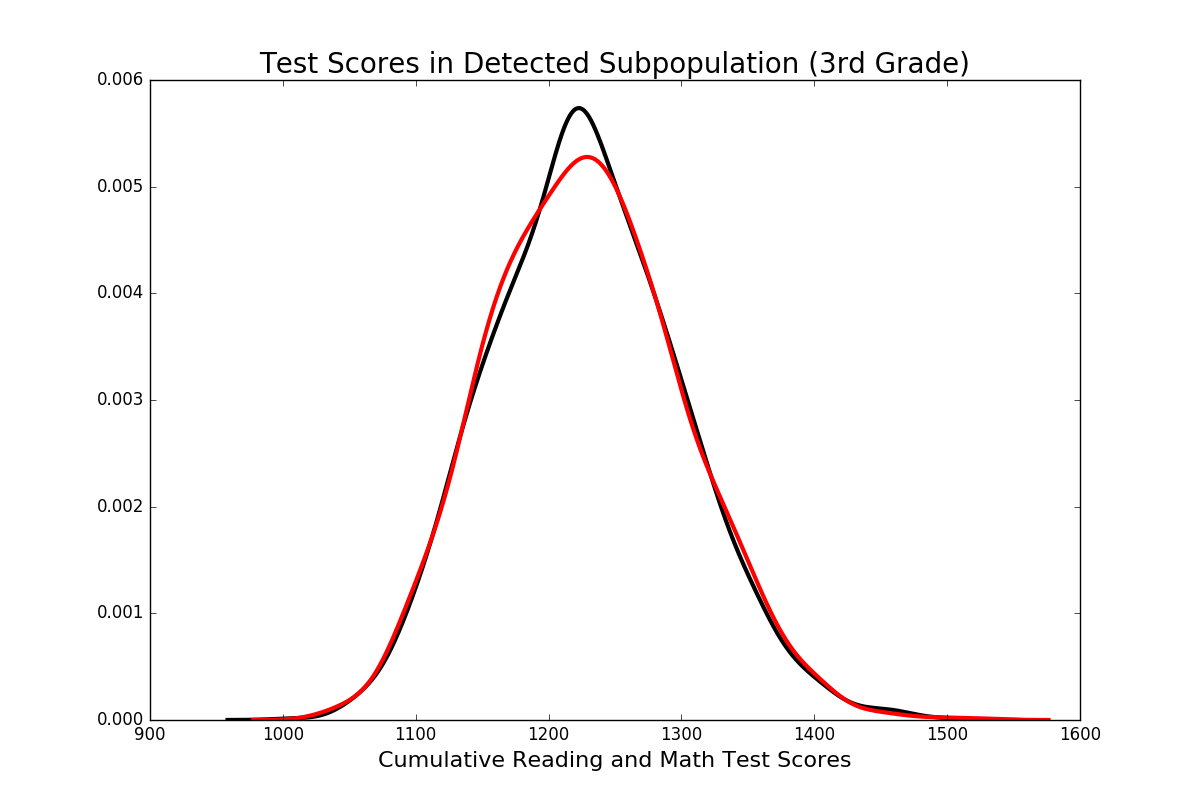}
    \caption{All students in 3rd grade}
    \label{fig:density-3rd-all}
  \end{subfigure}%
  \begin{subfigure}{.5\textwidth}
    \centering
    \includegraphics[width=.75\linewidth]{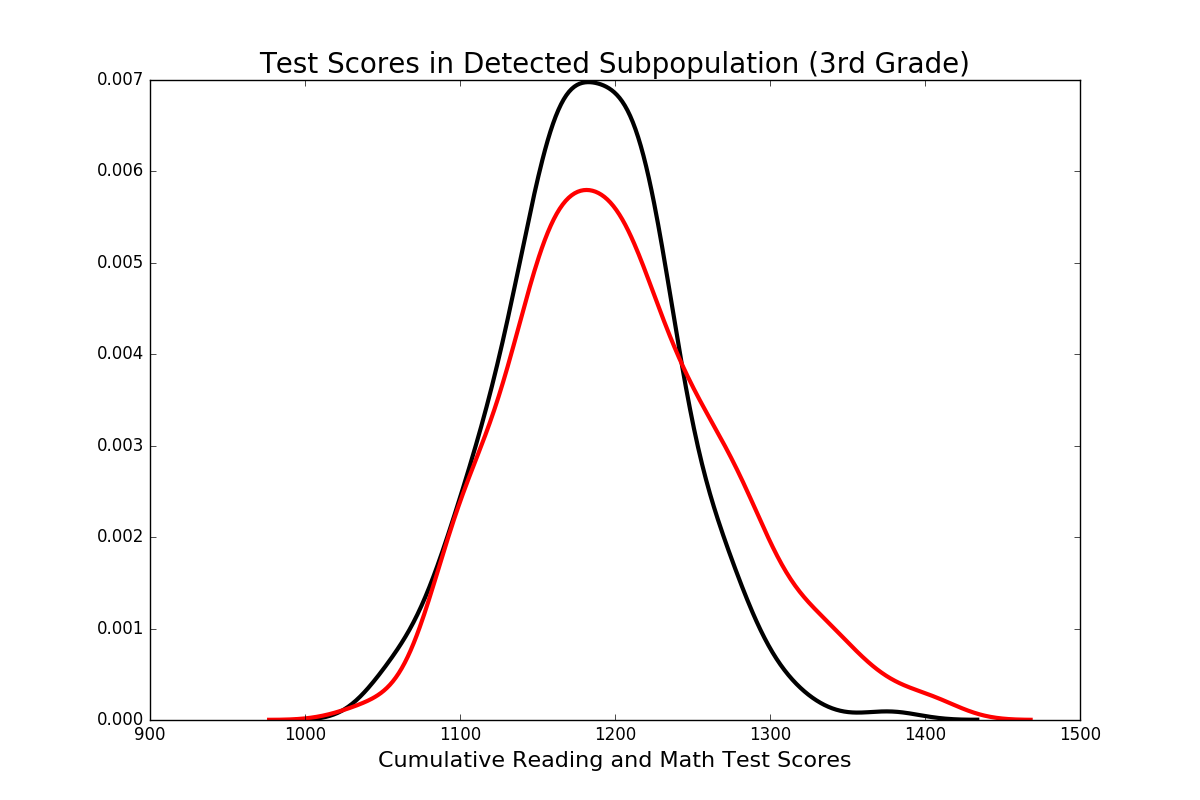}
    \caption{Detected Subpopulation in 3rd grade}
    \label{fig:density-3rd-subpop}
  \end{subfigure}
  \caption{Kernel density plots of 3rd grade test scores for treatment students
  (red) who were in a regular classroom with a teacher's aide and control
  students (black) who did not have a teacher's aide.}
  \label{fig:density-3rd}
\end{figure}
We would also like to obtain an unbiased estimate of the average treatment effect
in the subpopulation identified by TESS. Therefore, we follow a cross-validation
paradigm, where the entire dataset is partitioned into ten folds, and
iteratively each fold is held out as a validation set (to obtain an estimate of the
treatment effect) while the remaining nine folds are provided to TESS (for
detection). We further partition the data into records
corresponding to students observed in a regular classroom with an aide and a
regular classroom without an aide, which serve as treatment and control groups respectively. In three of the ten folds, TESS identified exactly the same subset, which we will call the ``detected subpopulation''.
Essentially, this detected subpopulation is composed of students in second or
third grade, who attended an inner-city or urban school, receiving instruction
from a teacher with 10 or more years of experience\footnote{The detected
subpopulation excluded teacher experience between 25 and 30
years. Including this range yields qualitatively the same results and
conclusions.}. Therefore, it appears that the presence of an aide raised the
test-scores of students exhibiting the selected covariate values described above
for grade, school type, and teacher experience, in addition to any values for
gender, free-lunch status, teacher ethnicity, and teacher degree. The
subpopulations that were returned in each of the ten folds exhibited a large amount
of agreement with the detected subpopulation: the fold subpopulations exhibited
88\% agreement (on average) with the detected subpopulation on the detection
status of a record. The estimated average treatment effect for this detected
subpopulation, averaged across all validation folds, is approximately a 34.19 point
increase in total test score (36.45 and 22.28 for second and third grades respectively).

Given this consistency across folds, we use the full data to better understand
the effect in the detected subpopulation generally. Table \ref{table:reg}
shows the evaluation of the treatment effect for all second-grade students
(column 1), second-grade students in the detected subpopulation (column 2), and second-grade students in the
complement of the detected subpopulation (column 3). Additionally, Figure
\ref{fig:density-2nd} shows the kernel density plots of the cumulative scores for
all second-grade students and students in the detected subpopulation
respectively. Figure \ref{fig:density-2nd-all} depicts a strong similarity in
the distribution of all second graders' scores with and without a full-time aide; there is a slight
difference around the center of the distribution, but its magnitude is not
sufficiently large to be significant, as seen by column 1 of Table \ref{table:reg}. Conversely,
Figure \ref{fig:density-2nd-subpop} depicts a difference in test scores for the detected subpopulation
of second graders: there
appears to be a clear effect of the treatment (dominated by a large mean shift),
supported by column 2 of Table \ref{table:reg}. We conduct a similar
analysis with third graders, and observe similar results in Figure
\ref{fig:density-3rd} and Table \ref{table:reg}. However, the effect of the
treatment in third grade appears to result in less of a mean shift, and is
better characterized by a change in the skew (third moment) and
therefore, the overall form of the distribution (Figure
\ref{fig:density-3rd-subpop}). We note that because TESS is able to identify effects
that change the distribution (and therefore higher order moments) of test
scores, even if the difference in mean score between treatment and control
students in third grade was smaller, TESS could potentially still identify the
existence of a treatment effect.
\begin{table}
  \captionsetup{font=scriptsize,skip=0pt}
  \centering
    \resizebox{\columnwidth}{!}{
    \begin{tabular}{|r|c|c|c|c|c|c|}
\cline{2-7}    \multicolumn{1}{r|}{} & All ($2^{nd}$) & Detected ($2^{nd}$) & Undetected ($2^{nd}$) & All ($3^{rd}$) & Detected ($3^{rd}$) & Undetected ($3^{rd}$) \bigstrut\\
    \hline
    \multicolumn{1}{|l|}{Treatment} & 3.479 & 36.066*** & 1.309 & -0.291 & 18.703*** & 0.1 \bigstrut\\
    
    \multicolumn{1}{|l|}{(std. dev.)}  & (2.547) & (6.055) & (2.772) & (2.277) & (5.18) & (2.478) \bigstrut\\
    \hline
    \multicolumn{1}{|l|}{P-value} & 0.172 & \textless0.001 & 0.637 & 0.898 & \textless0.001 & 0.968 \bigstrut\\
    \hline
    \multicolumn{1}{|l|}{Observations} & 4263 & 620 & 3643 & 4063 & 706 & 3357 \bigstrut\\
    \hline
    \end{tabular}%
    }
    \caption{Table of estimated treatment effects on student test scores in 2nd and 3rd grade. *** indicates $p<0.001$.}
  \label{table:reg}%
\end{table}%

There appears to be another consensus in the literature that small classrooms
have a consistent, positive, and significant effect
\citep{word-star-1990,krueger-star-1999,folger-star-1989};
therefore, we also compare small classrooms to regular classrooms, and
determine whether there appears to be a subpopulation which is the main driver of
this effect. We conduct an analysis as above but with STAR data records
corresponding to students observed in a small classroom (treatment group) and a regular classroom (control group).
For this analysis, TESS identified the entire population, which is congruent with the
previous literature's analysis of the consistent and significant average
treatment effect in each grade. This result from TESS appears to indicate that
the effect of small classroom size was not limited to a specific subpopulation.
For both TESS analyses,
we also conducted permutation testing to compensate for multiple hypothesis
testing. Based on these results, we conclude that there is a less than 0.01\% chance we would
obtain a subpopulation with a score as extreme under the null hypothesis.

The detected subpopulation in the classrooms with aides is not only statistically
significant, but may also
provide domain insight into the efficacy of full-time aides. A possible explanation for
the effect we observe in the detected subpopulation is the fact that 13 schools
were chosen at random to have teachers participate in an in-service training session, which
the literature has also deemed ineffective \citep{word-star-1990}. More
specifically, 57 teachers were selected each summer from these schools
to participate in a three-day in-service to help them teach more effectively in
whatever class type they were assigned to; part of the instruction focused on
how to work with an aide and also had the aides present. We note that the in-service only 
occurred during the summers prior to 2nd and 3rd grade, which are the grades identified by TESS.
Therefore, it is possible that when
provided proper training, the combination of an aide and an experienced teacher
can provide a significantly enhanced education environment even in the
challenging teaching environments that exist in inner-city and urban schools.
An additional explanation is
that the educational benefits may be cumulative--i.e., in each additional year a
student in this subpopulation has access to the combination of an aide and
experienced teacher, the treatment effect compounds--similar to what has been
demonstrated in small classrooms for the overall population
\citep{krueger-star-1999}. However, unlike in small classrooms, for this
subpopulation in regular classrooms with an aide, the effects were not large
enough to be distinguishable from zero (given the much smaller sample size of the affected subpopulation and smaller treatment effect) until
after two years. While a more detailed follow-up analysis of these hypotheses might reveal other causal mechanisms at work, we believe that these
results do present evidence that a treatment previously
believed to be ineffective may actually have been effective for a particularly
vulnerable subpopulation. Therefore, this analysis provides a sense of how TESS
can be used as a tool for data-driven hypothesis generation in real-world
policy analysis.

\section{Conclusions}
\label{sec:conclusion}
This paper has presented several contributions to the literature on statistical
machine learning approaches for heterogeneous quantile treatment effects. Specifically, we detect the existence of a subpopulation for which the conditional quantile treatment effect (CQTE) is non-zero. This allows
detection of treatment effects that manifest as arbitrary effects
on the potential outcome distributions (or specific quantiles), rather than being limited to detection of
mean shifts. Furthermore, we consider the challenge of identifying whether any subpopulation has been
affected by treatment, and precisely characterizing the affected subpopulation, as opposed to
the more typical problem setting of estimating individual-level treatment effects.
We formalize the identification of subpopulations with significant treatment effects as an anomalous pattern
detection problem, and present the Treatment Effect Subset Scan (TESS) algorithm, which
serves as a computationally efficient test statistic for the maximization of CQTE over all subpopulations. We demonstrate that the estimator used by TESS
satisfies the linear-time subset scanning property, allowing it to be efficiently
and exactly optimized over subsets of a covariate's values, while evaluating only
a linear rather than exponential number of subsets. This efficient conditional optimization step is incorporated into
an iterative procedure which jointly maximizes over subsets of values for each
covariate in the data: the result is a subpopulation, described as a subset of
values for each covariate, which demonstrates the most evidence for a
statistically significant treatment effect. In addition to its computational
efficiency, we derive desirable statistical properties for the TESS estimator:
bounded asymptotic probability of Type I and Type II errors under the sharp null hypothesis of no treatment effect, as well as providing sufficient
conditions under the alternative hypothesis that will result in TESS exactly
identifying the affected subpopulation. These properties apply more generally to the class of nonparametric scan statistics
upon which TESS is built; therefore, this theory also provides additional contributions
to the anomalous pattern detection, scan statistics, and goodness-of-fit literatures.

In addition to proposing a novel algorithm with desirable properties, we provide
an extensive comparison between TESS and other recently proposed statistical
machine learning methods for heterogeneous treatment effects (Causal Tree,
Interaction Tree, and Causal Forest) through semi-synthetic simulations. Our
results indicate that TESS consistently outperforms the other methods in its
ability to identify and precisely characterize subpopulations which exhibit treatment effects.
TESS significantly outperforms competing methods in the challenging scenarios
where the treatment effect signal is weak (i.e., the signal magnitude is low or the affected subpopulation is
small) because the subset scanning approach allows it to combine
subtle signals across various dimensions of data in order to
identify effects of interest. Moreover, TESS's detection performance is consistent even
when the treatment outcome distribution in the affected subpopulation has the
same mean as the control outcome distribution, while the competing
methods demonstrate essentially no ability to identify the affected
subpopulation in the absence of a mean shift.

After demonstrating TESS's performance through simulation, we explore the well-known Tennessee STAR experiment, searching for previously unidentified subpopulations
with significant treatment effects. As a result of this analysis, TESS uncovered an intuitive subpopulation
that seems to have experienced extremely significant improved test scores as a
result of having a teacher's aide in the classroom, a treatment that has
consistently been considered ineffective (as measured by the average treatment effect) by the literature on the
Tennessee STAR. This provides a sense of how TESS can be utilized as a tool for
generating hypotheses to be further explored and tested. We do however caution
researchers to view algorithms like TESS not as a replacement, but rather an
assistive tool, for developing scientific and behavioral theory. Results discovered by these methods should be investigated further and
evaluated to develop a deeper theoretical understanding of the phenomena they
uncover. When used to this end, these tools fill a critical void: in many
contexts it is rare to know \emph{a priori} which hypotheses are relevant and supported
by data, and the use of traditional methods (e.g., regression) puts the onus on
the researcher to know which hypothesis to test. This process
necessitates that theory comes first, and subsequent investigation is a form
of confirmatory analysis. However, such a process can become an impediment to
data-driven discovery: there is an increasing need for
scalable methods to use (big) data to generate new hypotheses, rather than just confirming pre-existing beliefs.

In the late 1970s, John W. Tukey began to outline his vision for the
future of statistics, which included a symbiotic relationship between
exploratory and confirmatory data analysis. He argues these two
forms of data analysis ``can--and should--proceed side by
side''~\citep{tukey-eda-1977} because he believed ideas ``come from previous
exploration more often than from lightning strokes''~\citep{tukey-eda_cda-1980}. To
this end Tukey advocates for using data to suggest hypotheses to test, or what
we now call data-driven hypothesis generation.
We see our work as the natural evolution of Tukey's vision of data analysis: we
develop an approach--rigorously conducted and theoretically grounded--to conduct
exploratory analysis in randomized experiments, with the hope of catalyzing
``lightning strokes'' of discovery and the advancement of science.

\bibliographystyle{abbrv}
\bibliography{tess}
\clearpage
\begin{appendices}

\section{Score Functions}
\label{sec:scoring_functions}
To begin we revisit the general form of the score function--or equivalently
the quantile treatment effect test statistic--that we refer to as the nonparametric scan statistic. Additionally, we establish equivalences, as different forms
will lend themselves to various proof strategies we implement later.

\begin{equation}
    \label{eqn:F(S)}
    \begin{alignedat}{2}
      \max_{S} F(S) &= \max_{S,\alpha} F_{\alpha}(S) &
        &=\max_{S,\alpha,\beta}F_{\alpha,\beta}(S)\\
      &= \max_{S,\alpha} \Delta \left( \alpha, N_{\alpha}(S), N(S) \right) &
        &=\max_{S,\alpha,\beta}\sum_{x \in U_X(S)}{\omega\left( \alpha, \beta,
        N_{\alpha}(x), N(x) \right)}.
    \end{alignedat}
\end{equation}

To motivate the use of our score function to evaluate the quantile treatment effect in subpopulation $S$, we will first demonstrate that simply maximizing the original conditional quantile treatment effect, $$\max_S \max_\alpha \mathbb{F}_{Y(1)|X\in S}^{-1}(\alpha) - \mathbb{F}_{Y(0)|X\in S}^{-1}(\alpha),$$ is unproductive as it reduces to the trivial solution of a singular covariate profile.
\begin{prop}
  \label{prop:trivial-max}
If $\left(S^{\ast}, \alpha^{\ast} \right)= \arg\max_{S \subseteq D,\alpha} \mathbb{F}_{Y(1)|X\in S}^{-1}(\alpha) - \mathbb{F}_{Y(0)|X\in S}^{-1}(\alpha)$ and\\ $x^{\ast} = \arg \max_{x \in D} \mathbb{F}_{Y(1)|X = x}^{-1}(\alpha^{\ast}) - \mathbb{F}_{Y(0)|X = x}^{-1}(\alpha^{\ast})$, then $S^{\ast} = x^{\ast}$.
\end{prop}
\begin{proof}

\noindent First let $S \subseteq D$ be any fixed subpopulation in our experimental data $D$, and $\alpha \in (0,1)$ be a fixed quantile. Next we re-write a (potential outcomes) distribution as follows:
\begin{equation*}
	\begin{split}
		\mathbb{F}_{Y|X \in S}(y) &= \frac{\mathbb{F}_{Y,X \in S}(y)}{F(X \in S)} \\
        &= \frac{P( Y \leq y, X \in S)}{P(X \in S)} \\
		&= \frac{\sum_{x \in S} P( \{Y \leq y, X = x\})}{P(X \in S)} \\
        &= \sum_{x \in S} P\left( Y \le y | X = x \right) P\left(X = x | x \in S\right).
	\end{split}
\end{equation*}
This allows us to also rewrite the (potential outcomes) quantile function as follows:
\begin{align*}
		\mathbb{F}_{Y|X \in S}^{-1}(\alpha) &= \sum_{x \in S} \inf_y \{y : P\left( Y \le y | X = x \right) \ge \alpha \}  P\left(X = x | x \in S\right)\\
		&= \sum_{x \in S} \mathbb{F}_{Y|X=x}^{-1}(\alpha) P\left(X = x | x \in S\right).
\end{align*}
Therefore, if we define
\begin{equation*}
    Q_x(\alpha) = \mathbb{F}_{Y(1)|X=x}^{-1}(\alpha) - \mathbb{F}_{Y(0)|X=x}^{-1}(\alpha)
\end{equation*}
then
\begin{equation*}
	\begin{split}
\mathbb{F}_{Y(1)|X\in S}^{-1}(\alpha) - \mathbb{F}_{Y(0)|X\in S}^{-1}(\alpha) &= \sum_{x \in S}  Q_x(\alpha) P\left(X = x | x \in S\right) \\
    &\le \max_{x \in S} Q_x(\alpha),
	\end{split}
\end{equation*}
where the inequality follows from the fact that $\sum_{x \in S} P\left(X = x | x \in S\right) = 1$. Moreover, because this inequality holds $\forall S,\alpha$, then it also must hold for $S^{\ast}, \alpha^{\ast}$. Therefore, we finally have
\begin{align}
    \label{eq:S<=x}
    \max_{S \subseteq D} \mathbb{F}_{Y(1)|X\in S}^{-1}(\alpha^{\ast}) - \mathbb{F}_{Y(0)|X\in S}^{-1}(\alpha^{\ast}) &\le \max_{x \in S \subseteq D} Q_x(\alpha)  \\
    \label{eq:S=x}
    &= \max_{x \in S \subseteq D} Q_x(\alpha) \\
    &= \max_{x \in S} \mathbb{F}_{Y(1)|X=x}^{-1}(\alpha) - \mathbb{F}_{Y(0)|X=x}^{-1}(\alpha) \nonumber
\end{align}
where \eqref{eq:S=x} follows from the fact that the inequality in \eqref{eq:S<=x} must be a strict equality because $\{ \max_{x \in S \subseteq D}\} \subseteq \{\max_{S \subseteq D}\}$.
\end{proof}

Now that we demonstrated that the usefulness of maximizing the original conditional quantile effect measure is compromised by its reduction to a singular covariate profile, let us consider the treatment effect evaluation measure based on our nonparametric scan statistic, by first revisiting the score function. In the main text we introduced the Berk-Jones score function which assumes that the data generating process for each $N_{\alpha}(x)$ follows a binomial distribution and therefore computes the log-likelihood ratio statistic $F(S) = \log \left( \frac { P \left( \text{Data} | H_1(S) \right)}{P \left(
    \text{Data} | H_0 \right) } \right)$, which can be written as the product of the total number of $p$-values $N(S)$ in subset $S$ and a divergence
    $Div\left(\frac{N_\alpha(S)}{N(S)},\alpha\right)$ between the observed and expected proportions of $p$-values that are significant at level $\alpha$. More specifically, we have:
\begin{flalign*}
  H_0&: N_{\alpha}(x) \sim \text{Binomial} \left( N(x), \alpha \right) \quad
  \forall x &&\\
  H_1(S)&: N_{\alpha}(x) \sim \text{Binomial} \left( N(x), \beta \right) \quad
  \forall x \in S \quad \beta \not = \alpha, &&
\end{flalign*}
with the following Berk-Jones (BJ) log-likelihood ratio statistic~\citep{berk-bj-1979}:
\begin{align*}
  F_{\alpha}^{BJ}(S) &= \log \left [ \frac { P \left( \text{Data} | H_1(S) \right)}{P \left(
    \text{Data} | H_0 \right) } \right ] \\ &=
    N_{\alpha}(S) \log  \left(\frac {\beta} {\alpha} \right) + \left( N(S)-N_{\alpha}(S) \right) \log \left( \frac{1-\beta}{1-\alpha} \right) \\ &=
    N(S)Div_{KL} \left(\frac{N_{\alpha}(S)}{N(S)}, \alpha \right),
\end{align*}
where we have used the maximum likelihood estimate $\beta =
\beta_{\text{mle}}(S) = \frac{N_{\alpha}(S)}{N(S)}$, and
$Div_{KL}(\cdot,\cdot)$ is the Kullback-Leibler divergence, $Div_{KL}(x,y) = x
\log \frac{x}{y} + (1-x) \log \frac{1-x}{1-y}$

We now also introduce the Normal-Approximation score function, which is based on the normal approximation to the binomial data generating process assumed in Berk-Jones:
\begin{flalign*}
  H_0&: N_{\alpha}(x) \sim \text{Gaussian} \left(N(x)\alpha, \alpha(1-\alpha)N(x) \right) \quad
  \forall x &&\\
  H_1(S)&: N_{\alpha}(x) \sim \text{Gaussian} \left(N(x)\beta , \alpha(1-\alpha)N(x) \right) \quad
  \forall x \in S \quad \beta \not = \alpha, &&
\end{flalign*}
with the following normal approximation (NA) log-likelihood ratio statistic:
\begin{align*}
  F_{\alpha}^{NA}(S) &= \log \left [ \frac { P \left( \text{Data} | H_1(S) \right)}{P \left(
    \text{Data} | H_0 \right) } \right ] \\
  &= \frac{N_{\alpha}(S)\left(\beta-\alpha\right)}{\alpha(1-\alpha)} +
    \frac{N(S)\left(\alpha^2-\beta^2\right)}{2\alpha(1-\alpha)} \\
  &= \frac{\left(N_{\alpha}(S)-N(S)\alpha\right)^2}{2N(S)\alpha(1-\alpha)} \\
  &= N(S)Div_{\frac{1}{2}\chi^2} \left(\frac{N_{\alpha}(S)}{N(S)}, \alpha \right).
\end{align*}
where we have again used the maximum likelihood estimate of $\beta =
\frac{N_{\alpha}(S)}{N(S)}$, and $Div_{\frac{1}{2}\chi^2}(\cdot,\cdot)$ is a
scaled $\chi^2$ divergence, $Div_{\frac{1}{2}\chi^2}(x,y) =
\frac{(x-y)^2}{2y(1-y)}$.

The first result we show is that in the limit $F^{BJ}$ is well approximated
by $F^{NA}$, which will then allow us to focus the remainder of our theoretical
results on $F^{NA}$ specifically.

\begin{restatable}{prop}{bjtona}
  \label{prop:BJ_to_NA}
  $F^{BJ}(S) \asymp F^{NA}(S)$ as $N(S)\longrightarrow \infty$.
\end{restatable}
\begin{proof}
  Recall that $K(x,y) = Div_{KL}(x,y) = x \log \frac{x}{y} + (1-x) \log \frac{1-x}{1-y}$. By expanding $K(x,y)$ through a Taylor series, we have
  \begin{align*}
    K(x,y) &= K(y,y) + \frac{\partial K(x,y)}{\partial x}\Biggr\rvert_{x=y}
      \left(x-y\right) + \frac{\partial^2 K(x,y)}{\partial^2 x}\Biggr\rvert_{x=y'}
      \frac{\left(x-y\right)^2}{2}\\
    &= 0 + 0 + \frac{\left(x-y\right)^2}{2y'(1-y')}\\
  \end{align*}
  for some $y'$ such that $|y'-x|\le |y-x|$. Therefore,
  \begin{align*}
    F^{BJ}(S) &= \max_{\alpha} N\left(S\right)K
    \left(\frac{N_{\alpha}\left(S\right)}{N\left(S\right)},\alpha\right)\\
    &=\max_{\alpha}
      N\left(S\right)\frac{\left(\frac{N_{\alpha}\left(S\right)}{N\left(S\right)}-\alpha\right)^2}{2\alpha'(1-\alpha')}
      \quad\left(\text{where
      }\biggr\rvert\alpha'-\frac{N_{\alpha}\left(S\right)}{N\left(S\right)}\biggr\rvert\le
      \biggr\rvert\alpha-\frac{N_{\alpha}\left(S\right)}{N\left(S\right)}\biggr\rvert\right)\\
    &\le \max_{\alpha}
      N\left(S\right)\left[\frac{\left(\frac{N_{\alpha}\left(S\right)}{N\left(S\right)}-\alpha\right)^2}{2\alpha(1-\alpha)}
      \bigvee
      \frac{\left(\frac{N_{\alpha}\left(S\right)}
        {N\left(S\right)}-\alpha\right)^2} {2\frac{N_{\alpha}\left(S\right)}
        {N\left(S\right)} \left(1-\frac{N_{\alpha}\left(S\right)}{N\left(S\right)}
        \right)} \right]\\
    \text{and}&{}\\
    &\ge \max_{\alpha} N\left(S\right) \left[\frac{\left(
      \frac{N_{\alpha}\left(S\right)}
      {N\left(S\right)}-\alpha\right)^2}{2\alpha(1-\alpha)} \bigwedge
      \frac{\left(\frac{N_{\alpha}\left(S\right)}{N\left(S\right)}
      -\alpha\right)^2}{2\frac{N_{\alpha}\left(S\right)} {N\left(S\right)}
      \left(1-\frac{N_{\alpha}\left(S\right)}{N\left(S\right)}\right)} \right].
  \end{align*}
  Furthermore, under $H_0$, $\frac{N_{\alpha}\left(S\right)}{N\left(S\right)}
  \asconv \alpha \implies \alpha' \asconv \alpha$, which by the continuous
  mapping theorem results in
  \begin{equation*}
    F^{BJ}(S) \asconv \max_{\alpha} N\left(S\right)\frac{\left(\frac{N_{\alpha}\left(S\right)}{N\left(S\right)}-\alpha\right)^2}{2\alpha(1-\alpha)} =F^{NA}(S).
  \end{equation*}
  However, under $H_1\left(S^T\right)$,
  $\frac{N_{\alpha}\left(S\right)}{N\left(S\right)} \asconv \beta(\alpha)$,
  therefore asymptotically for $F^{BJ}(S)$ we have,
  \begin{align*}
    \max_{\alpha} N\left(S\right)\frac{\left(\frac{N_{\alpha}\left(S\right)}{N\left(S\right)}-\alpha\right)^2}{2\alpha(1-\alpha)} \left(1 \bigwedge \frac{\alpha(1-\alpha)}{\beta(\alpha)\left(1-\beta(\alpha)\right)}\right)
    \le & F^{BJ}(S) \\
    \le & \max_{\alpha} N\left(S\right)\frac{\left(\frac{N_{\alpha}\left(S\right)}{N\left(S\right)}-\alpha\right)^2}{2\alpha(1-\alpha)} \left(1 \bigvee \frac{\alpha(1-\alpha)}{\beta(\alpha)\left(1-\beta(\alpha)\right)}\right).
  \end{align*}
  We can see that $F^{BJ}(S)$ is bounded above and below by either $F^{NA}(S)$ or a constant times $F^{NA}(S)$.
\end{proof}

We also note that there are a collection of well-known
supremum goodness-of-fit statistics used in the literature, all of which are described in \cite{jager-gof-2007}, that can be written as a
transformation of $F_{\alpha}^{NA}(S)$:\\
\\the Kolmogorov-Smirnov statistic\ignore{\citep{kolmogorov-gof-1933}}
\begin{align*}
F^{KS}(S) &=  \max_{\alpha} F_{\alpha}^{KS}(S) \\
&= \max_{\alpha} \frac{\left(N_{\alpha}(S)-N(S)\alpha\right)}{\sqrt{N(S)}} \\
&=  \max_{\alpha}  \sqrt{2\alpha(1-\alpha)F_{\alpha}^{NA}(S)},
\end{align*}
the Cramer-von Mises statistic\ignore{\citep{cramer-gof-1928, jager-gof-2007}}
\begin{align*}
F^{CV}(S) &=  \max_{\alpha} F_{\alpha}^{CV}(S) \\
&= \max_{\alpha} \frac{\left(N_{\alpha}(S)-N(S)\alpha\right)^2}{N(S)} \\
&=  \max_{\alpha}  2\alpha(1-\alpha)F_{\alpha}^{NA}(S),
\end{align*}
the Higher-Criticism statistic\ignore{\citep{donoho-gof-2004}}
\begin{align*}
F^{HC}(S) &=  \max_{\alpha} F_{\alpha}^{HC}(S) \\
&= \max_{\alpha} \frac{\left(N_{\alpha}(S)-N(S)\alpha\right)}{\sqrt{N(S)\alpha(1-\alpha)}} \\
&=  \max_{\alpha} \sqrt{2F_{\alpha}^{NA}(S)},
\end{align*}
and the Anderson-Darling statistic\ignore{\citep{anderson-gof-1952,jager-gof-2007}}
\begin{align*}
F^{AD}(S) &=  \max_{\alpha} F_{\alpha}^{AD}(S) \\
&= \max_{\alpha} \frac{\left(N_{\alpha}(S)-N(S)\alpha\right)^2}{N(S)\alpha(1-\alpha)} \\
&=  \max_{\alpha}  2F_{\alpha}^{NA}(S).
\end{align*}

\noindent As a result of this connection between $F^{NA}$ and these other statistics, we have 
the following:
\begin{prop}
  \label{prop:na_transform_max}
  If $S$ maximizes $F_{\alpha}^{NA}(S)$, then it maximizes
  $F_{\alpha}^{\text{KS}}(S), F_{\alpha}^{\text{CV}}(S),
  F_{\alpha}^{\text{HC}}(S)$ and $F_{\alpha}^{\text{AD}}(S)$.
\end{prop}
\begin{proof}
  First, we note that $T(F_{\alpha}^{NA})$, where $T(x)= (bx)^a$, for $b
  \in \{1, 2, 2\alpha(1-\alpha)\}$ and $a \in \{1,\frac{1}{2}\}$ is a
  monotonically increasing transformation. Therefore, $\arg\max_S
  F_{\alpha}^{NA}(S) = \arg\max_S T \left( F_{\alpha}^{NA}(S)
  \right)$, because $\arg\max$ is invariant to monotone transformations.
\end{proof}

We now show that assumptions (A1), (A2), and (A3) stated in \S\ref{sec:eff_scan} are satisfied by scoring functions $F_{\alpha}^{BJ}(S)$ and $F_{\alpha}^{NA}(S)$, as well as all of the other functions discussed above, given that they are monotone transformations of $F_{\alpha}^{NA}(S)$. 

\begin{restatable}{prop}{bj_assump}
  \label{prop:bj_assump}
  For a fixed value of $\alpha$, $F_{\alpha}^{BJ}(S)$ is (1) monotonically increasing with respect to $N_\alpha(S)$, (2) monotonically decreasing with respect to $N(S)$, and (3) convex with respect to $N_\alpha(S)$ and $N(S)$.
\end{restatable}
\begin{proof}
   $F_\alpha^{BJ}(S)$ can be written as $N_\alpha(S)\log \left(\frac{N_\alpha(S)}{N(S)\alpha} \right) + (N(S) - N_\alpha(S))\log\left(\frac{N(S) - N_\alpha(S)}{N(S) - N(S)\alpha}\right)$. $F_\alpha^{BJ}(S)$ is monotonically increasing w.r.t. $N_\alpha(S)$ because  
   \begin{align*}
       \frac{\partial F_\alpha^{BJ}(S)}{\partial N_\alpha(S)} &= 1 + \log\left( \frac{N_\alpha(S)}{N(S)\alpha}\right) - 1  - \log \frac{N(S) - N_\alpha(S)}{N(S) - N(S)\alpha} \\
       &= \log\left( \frac{N_\alpha(S)}{N(S)\alpha} \right) + \log \left( \frac{N(S) - N(S)\alpha}{N(S) - N_\alpha(S)} \right)\\
       &= \log\left( \frac{\frac{N_\alpha(S)}{N(S)}}{1- \frac{N_\alpha(S)}{N(S)}}\cdot \frac{1-\alpha}{\alpha}\right)\\
       &\ge 0.
   \end{align*}
\noindent The last inequality is strict when $\frac{N_\alpha(S)}{N(S)} > \alpha$, and we define $F_\alpha^{BJ}(S) = 0$ otherwise. 

Similarly, $F_\alpha^{BJ}(S)$ is monotonically decreasing w.r.t. $N(S)$ because
    \begin{align*}
       \frac{\partial F_\alpha^{BJ}(S)}{\partial N(S)} &= -\frac{N_\alpha(S)}{N(S)} +\log \left(\frac{N(S) - N_\alpha(S)}{N(S) - N(S)\alpha} \right) + \frac{N_\alpha(S)}{N(S)} \\
       &= \log \left(\frac{N(S) - N_\alpha(S)}{N(S) - N(S)\alpha} \right)\\
       &\le 0.
   \end{align*}
\noindent Again, the last inequality is strict when $\frac{N_\alpha(S)}{N(S)} > \alpha$, and we define $F_\alpha^{BJ}(S) = 0$ otherwise. 

Finally, $F_\alpha^{BJ}(S)$ is convex in $N_\alpha(S)$ and $N(S)$ because it can be written as $N(S) f\left(\frac{N_\alpha(S)}{N(S)} \right)$, where $f(x) = x \log \frac{x}{\alpha}+(1-x)\log \frac{1-x}{1-\alpha}$.  $f(x)$ is a convex function, since $\frac{d^2 f}{d x^2} = \frac{1}{x}+\frac{1}{1-x} > 0$ for $x \in (0,1)$.  Then $F_\alpha^{BJ}(S)$ is convex since it is the perspective of a convex function.
\end{proof}

\begin{restatable}{prop}{na_assump}
  \label{prop:na_assump}
  For a fixed value of $\alpha$, $F_{\alpha}^{NA}(S)$ is (1) monotonically increasing with respect to $N_\alpha(S)$, (2) monotonically decreasing with respect to $N(S)$, and (3) convex with respect to $N_\alpha(S)$ and $N(S)$.
\end{restatable}
\begin{proof} 
$F_{\alpha}^{NA}(S)$ can be written as $\frac{\left(N_\alpha(S) - N(S)\alpha \right)^2}{2 N(S) \alpha (1-\alpha)}$. $F_{\alpha}^{NA}(S)$ is monotonically increasing w.r.t. $N_\alpha(S)$ as
\begin{align*}
    \frac{\partial F_{\alpha}^{NA}(S)}{\partial N_\alpha(S)} 
    &= \frac{\frac{N_\alpha(S)}{N(S)} - \alpha}{\alpha (1-\alpha)}\\
    &\ge 0. 
\end{align*}
\noindent The last inequality is strict when $\frac{N_\alpha(S)}{N(S)} > \alpha$, and we define $F_\alpha^{NA}(S) = 0$ otherwise. 

 Similarly, $F_{\alpha}^{NA}(S)$ is monotonically decreasing w.r.t. $N(S)$ as 
\begin{align*}
    \frac{\partial F_{\alpha}^{NA}(S)}{\partial N(S)} 
    &= \frac{2N(S)^2 \alpha^3 (1-\alpha) - 2 N_\alpha(S)^2 \alpha (1-\alpha)}{4N(S)^2\alpha^2 (1-\alpha)^2}\\
    &= \frac{\alpha^2 - \left(\frac{N_\alpha(S)}{N(S)}\right)^2}{2\alpha (1-\alpha)}\\
    &\le 0,
\end{align*}
given $\alpha \in (0,1)$. Again, the last inequality is strict when $\frac{N_\alpha(S)}{N(S)} > \alpha$, and we define $F_\alpha^{NA}(S) = 0$ otherwise.

Finally, $F_\alpha^{NA}(S)$ is convex in $N_\alpha(S)$ and $N(S)$ because it can be written as $N(S) f\left(\frac{N_\alpha(S)}{N(S)} \right)$, where $f(x) = \frac{(x-\alpha)^2}{2\alpha(1-\alpha)}$.  $f(x)$ is a convex function, since $\frac{d^2 f}{d x^2} = \frac{1}{\alpha(1-\alpha)} > 0$ for $\alpha \in (0,1)$.  Then $F_\alpha^{NA}(S)$ is convex since it is the perspective of a convex function.

\end{proof}

\ignore{
\begin{restatable}{prop}{bjassums}
  \label{prop:BJ_A1-A3}
  $F^{BJ}(S) \asymp F^{NA}(S)$ as $N(S)\longrightarrow \infty$.
\end{restatable}
}

\section{Supplementary Materials: Proofs of Lemmas and Theorems}
\label{sec:proofs}
In this section, we provide detailed proofs of the Lemmas and Theorems stated in the main text. Before presenting the proofs, we (re-)introduce notation that will be used throughout the proofs.

\subsection{Notation}
$S^T$: the truly affected (rectangular) subset. \\
$S^{\ast}$: the highest scoring (rectangular) subset, $\arg \max_{S \in Rect} F(S)$, where
$Rect$ is the set of all rectangular subsets in $D$. \\
$\alpha^{\ast}$: the $\alpha$ at which $S^{\ast}$ is highest scoring, i.e., $\arg \max_\alpha F_{\alpha}(S^{\ast})$.    \\
$S^{\ast}_{u}$: the highest scoring unconstrained subset, $\arg \max_{S \subseteq D} F(S)$. \\
$\alpha^{\ast}_{u}$: the $\alpha$ at which $S^{\ast}_{u}$ is highest scoring, i.e,
$\arg \max_\alpha F_{\alpha}(S^{\ast}_{u})$. \\    
$U_{X}$: a function which returns the unique covariate profiles (non-empty tensor cells) in a set. \\
$M$: $|U_{X}(D)|$, the number of unique covariate profiles in our treatment data, or equivalently the number of cells with treatment observations in our data tensor. \\
$k$: $\frac{|U_{X}(S^T)|}{|U_{X}(D)|}$, the proportion of non-empty cells that are affected under $H_1\left(S^T\right)$.\\
$\beta(\alpha)$: $P\left(\hat{p}(y; x) \le \alpha \:\rvert\:
H_1\left(S^T\right)\right)$, for all the $p$-values of covariate profiles $x \in U_{X}(S^T)$.\\
$h(\delta)$: the critical value for the test statistic, $\max_{S\in Rect}F(S)$, at a given Type-I error rate $\delta >0$.\\
$\phi$:  Probability density function of standard normal distribution.\\
$\Phi$:  Cumulative distribution function of standard normal distribution.\\

\noindent Additionally, we will assume that each of the $M$ non-empty tensor cells (or unique treatment profiles) contain exactly $n$ $p$-values, for mathematical convenience. We are interested in the distribution of the score $\max_S F(S)$ under the null hypothesis $H_0$, assuming that all $p$-values are uniformly distributed on $[0,1]$, and under the alternative hypothesis $H_1(S^T)$, assuming that there exist some constants $\alpha$ and $\beta$ such that $\mbox{Pr}(\hat{p} < \alpha) = \beta$ in subset $S^T$, for $\beta > \alpha$.  We assume that a constant fraction of cells $k$, $0 < k \le 1$, are affected under $H_1$. \\

\subsection{Statistical Properties}
\label{sec:stats_theory}
We now demonstrate desirable statistical properties of $\max_{S \in Rect} F(S)$. Our derivations are based on $F^{NA}(S)$, but we show that these properties also extend to $F^{BJ}(S)$, our statistic of choice in the main text, and many other statistics because of their close relationship with $F^{NA}(S)$, as described in Appendix~\ref{sec:scoring_functions}. More specifically, we 
demonstrate that, using $\max_{S \in Rect} F(S)$ as a test statistic of the data, we can appropriately (fail to) reject $H_0$ with high probability. For mathematical convenience, the results derived in this section assume that $N(x) = n$
for all $x \in U_{X}(D)$, i.e., each unique covariate profile in the data has exactly $n$ data points (and therefore $n$ $p$-values). We also assume under the alternative hypothesis $H_1(S^T)$ that a constant fraction of cells $k$, $0 < k \le 1$, are affected. Finally, we consider the asymptotic regime where $n \longrightarrow \infty$. Ultimately, we will show that for any Type I error rate $\delta > 0$, we can compute a critical value $h(\delta)$ such that we have the following:
\begin{align*}
  \lim_{n \rightarrow \infty}P_{H_0}\left(\max_{S \in Rect} F(S) > h(\delta)\right) &\le \delta,\\
  \lim_{n \rightarrow \infty}P_{H_1}\left(\max_{S \in Rect} F(S) > h(\delta)\right) &= 1.
\end{align*}

We begin by first recognizing that $F(S^{\ast}) \le F(S^{\ast}_{u})$, i.e., the score of the optimal rectangular subset is upper bounded by the score of optimal unconstrained subset, because the space of rectangular subsets is contained within the space of all subsets. Consequently, the distribution $F(S^{\ast}_{u})$ under $H_0$ provides a upper bound on the distribution of $F(S^{\ast})$ under $H_0$. Therefore, we will begin by establishing a distributional upper bound on $\sqrt{F(S^{\ast}_{u})}$ under $H_0$, which by transitivity will also provide an upper bound on the distribution of $\sqrt{F(S^{\ast})}$.
\nullconverg
\begin{proof}
Recall that our experiment is made up of i.i.d.~units $\{R_{1}, \ldots, R_{N}\}$, where $R_{i} = (Y^{\text{obs}}_i, X_i, W_i)$ is a 3-tuple. Also recall from Section \ref{sec:p-values} that for every treatment unit $R_i$, we have a $p$-value $\hat{p}_i$, where $\hat{p}_i \sim U(0,1)$ under the null hypothesis $H_0$ of no treatment effect. Therefore, for each of the unique covariate profiles in our treatment data $x\in U_{X}(D)$, we can compute the number of significant $p$-values, $N_{\alpha}\left(x\right) \sim \text{Binomial}\left(N\left(x\right), \alpha\right)$, for any given value of $\alpha$. Recall that we assume $N\left(x\right) = n~\forall x\in U_{X}(D)$, that $|U_{X}(D)| = M$, and that  $\phi$ and $\Phi$ are the Gaussian pdf and cdf respectively. In order to arrive at the intended result on the distribution of $\sqrt{F(S^{\ast}_{u})}$, we will need to make a set of interrelated observations.\\
  
\underline{Observation 1:} Let us define $S^{\ast}_{\alpha, u} = \arg\max_{S_u \subseteq U_{X}(D)} F_\alpha(S_u)$, the highest scoring unconstrained subset of covariate profiles for a given $\alpha$. From Theorem \ref{thm:LTSS} we know that if the profiles are sorted $\{x_{(1)}, \ldots, x_{(M)}\}$ according to priority function $\frac{ N_{\alpha}\left(x\right)}{n}$, where
  $x_{(t)}$ has the $t^{th}$ highest priority, then
\begin{align*}
S^{\ast}_{\alpha, u} &\in \{\{x_{(1)}, \ldots, x_{(t)}\}\}_{t\in\{1,\ldots, M\}}  \\
             &= \left\{x \:|\:  N_{\alpha}\left(x\right) > t(\alpha)\right\}.
\end{align*}
Essentially, $S^{\ast}_{\alpha, u}$ will consist of all and only those profiles $x$ with $N_{\alpha}\left(x\right)$ above some threshold $t(\alpha)$. Because $N_{\alpha}\left(x\right) \sim \text{Binomial}(n, \alpha)$ under $H_0$, we can write $t(\alpha) = n\alpha + Z\sqrt{n\alpha(1-\alpha)}$ for some constant $Z$, allowing the threshold to represent $Z$-standard deviations above the expected number of significant $p$-values. Therefore, for given values of $\alpha$ and $Z$, the event that a given profile $x$ will be included in $S^{\ast}_{\alpha, u}$ can be defined as
\ignore{\stackrel{d}{\longrightarrow}}
\begin{equation}
\label{eq:N(x)}
\begin{split}
    \mathbbm{1}_{\{x \in U_{X}(S^{\ast}_{\alpha, u})\}} &~~\sim ~\mbox{Bernoulli}\left( \mbox{P}\left[N_{\alpha}\left(x\right) \ge n\alpha + Z\sqrt{n\alpha(1-\alpha)}\right] \right)\\
    &\stackrel{d}{\longrightarrow} \mbox{Bernoulli}\left( 1-\Phi\left(Z\right)\right),
\end{split}    
\end{equation}
as $n \longrightarrow \infty$; and consequently we also have
\begin{equation}
    \label{eq:N(Su)}
    |U_{X}(S^{\ast}_{\alpha, u})| \stackrel{d}{\longrightarrow} \mbox{Binomial}(M,1-\Phi\left(Z\right)).
\end{equation}
To close, from this observation, for fixed $\alpha$ and $Z$, we have the asymptotic distribution governing the event that an individual covariate profile will be included in the detected subset $S^{\ast}_{\alpha, u}$ and consequently the asymptotic distribution over the number of profiles to be included.\\

\underline{Observation 2:} From~\eqref{eq:N(x)} we also have $N_{\alpha}(x) \:|\: x \in U_{X}(S^{\ast}_{\alpha, u})  \sim \text{TruncatedBinomial}(n, \alpha,Z)$; for the covariate profiles included in $S^{\ast}_{\alpha, u}$, the observed number of significant $p$-values follow a truncated binomial distribution. Moreover, for given values of $\alpha$ and $Z$, as $n \longrightarrow \infty$,
\begin{equation}
\label{eq:N_a(Su)}
\sqrt{n}\left(\frac{N_{\alpha}\left(x\right)}{n} - \alpha\right) \bigg{\rvert} ~ x \in U_{X}(S^{\ast}_{\alpha, u}) \stackrel{d}{\longrightarrow} \mbox{TruncatedGaussian}(0,\alpha(1-\alpha),Z), \end{equation}
whose expected value is $\frac{\phi(Z)}{1-\Phi(Z)}\sqrt{\alpha(1-\alpha)}$ and variance is $\alpha(1-\alpha)V(Z)$, where $V(Z) < 1$ is the variance reduction from a truncated Gaussian, $V(Z) = 1 + \frac{Z \phi(Z)}{1-\Phi(Z)} - \left(\frac{\phi(Z)}{1-\Phi(Z)}\right)^2$. To close, from this observation, for fixed $\alpha$ and $Z$, we obtain the asymptotic distribution governing the number of significant $p$-values for covariate profiles in the detected subset $S^{\ast}_{\alpha, u}$.\\

\underline{Observation 3:} We can write
\begin{equation*}
    \frac{N_\alpha(S^{\ast}_{\alpha, u})}{N(S^{\ast}_{\alpha, u})} = \frac{\sum_{x \in U_{X}(S^{\ast}_{\alpha, u})} N_{\alpha}\left(x\right)}{n \: |U_{X}(S^{\ast}_{\alpha, u})|},
\end{equation*}
which when combined with the asymptotic distributions that govern the behaviors of profiles included in $S^{\ast}_{\alpha, u}$ (from Observations 1 and 2) we can conclude
\begin{equation}
    \label{eq:F_a(Su)}
    \begin{split}
        \frac{\sqrt{N(S^{\ast}_{\alpha, u})}\left(\frac{N_\alpha(S^{\ast}_{\alpha, u})}{N(S^{\ast}_{\alpha, u})} - \alpha\right)}{\sqrt{2\alpha(1-\alpha)}} - \sqrt{\frac{M\phi(Z)^2}{2(1-\Phi(Z))}} &= \sqrt{F_{\alpha}^{NA}(S^{\ast}_{\alpha, u})} - \sqrt{\frac{M\phi(Z)^2}{2(1-\Phi(Z))}} \\
        &\stackrel{d}{\longrightarrow} \mbox{Gaussian}\left(0, \frac{V(Z)}{2}\right),
    \end{split}
\end{equation}
by the Central Limit Theorem. To close, from this observation, for fixed $\alpha$ and $Z$, we have the asymptotic distribution that governs the score function for the detected subset $S^{\ast}_{\alpha, u}$.\\

\underline{Observation 4:} While Observation 3 provides the distribution of the score function for the detected subset optimized over $Z$, the distribution is still defined for a fixed $\alpha$. This last observation will address the supremum over $\alpha$. To begin, for a given subset $S$, we can collect all its (treatment unit) $p$-values: $P_S = \{\hat{p}_{i}~|~x_i \in S\}$ where $|P_S| = n_s$ and $\hat{p}_i \sim U(0,1)$ under $H_0$. Moreover, if we let $P_{n_s}(\alpha) = \frac{1}{n_s}\sum_{\hat{p}_i\in P_S}{\mathbbm{1}_{\{\hat{p}_i \le \alpha\}}}$, then $\mathbb{U}_{n_s}(\alpha) = \frac{\sqrt{n_s}(P_{n_s}(\alpha)-\alpha)}{\sqrt{\alpha(1-\alpha)}}$ is a normalized uniform empirical process, indexed by $\alpha \in (0,1)$. Next let $\mathbb{W}_{n_s}(\alpha_{\min},\alpha_{\max}) = \sup_{\alpha \in [\alpha_{\min},\alpha_{\max}]}|\mathbb{U}_{n_s}(\alpha)|$, be the supremum over the absolute value of the normalized uniform empirical process, restricted to $[\alpha_{\min},\alpha_{\max}]$ in the interior of $(0,1)$. Recognize that $\sqrt{F^{NA}_{\alpha}(S)}$ is a scaled version of $|\mathbb{U}_{n_s}(\alpha)|$, and therefore 
\begin{align}
 \sqrt{F^{NA}(S)} &= \sup_{\alpha \in [\alpha_{\min},\alpha_{\max}]} \sqrt{F^{NA}_{\alpha}(S)} \nonumber\\
    &\propto \mathbb{W}_{n_s}(\alpha_{\min},\alpha_{\max})  \nonumber\\
    &\xrightarrow{d} \mathbb{W}(\alpha_{\min},\alpha_{\max})  \nonumber\\
    \label{eq:B(a)_norm}
    &= \sup_{\alpha \in [\alpha_{\min},\alpha_{\max}]} \frac{|B(\alpha)|}{\sqrt{\alpha(1-\alpha)}},
\end{align}
where $B(\alpha)$ is the Brownian bridge on $[0,1]$. We consider the supremum over $[\alpha_{\min},\alpha_{\max}]$ in the interior of $[0,1]$ because if left unrestricted, $\mathbb{W}_{n_s}(0,1)$ increases with $n_s$, and $\mathbb{W}(0,1)$ becomes
arbitrarily large~\citep{miller-max_chi_square-1982}. So to close, from this observation, for fixed $Z$ but supremum over quantiles $\alpha$, we know that the score function for any subset $S$ follows a normalized uniform empirical process.\\

If we take all four observations together with the definition of $F^{NA}(S^{\ast}_{u})$, we obtain 
\begin{align*}
     \max_{S_u \subseteq U_{X}(D)}\sqrt{F^{NA}(S_u)} &= \max_{S_u \subseteq U_{X}(D), \alpha\in [\alpha_{\min},\alpha_{\max}]}\sqrt{F_{\alpha}^{NA}(S_u)}  \nonumber \\
     &= \max_{Z, \alpha\in [\alpha_{\min},\alpha_{\max}]} \frac{\sqrt{N(S^{\ast}_{u})}\left(\frac{N_\alpha(S^{\ast}_{u})}{N(S^{\ast}_{u})} - \alpha\right)}{\sqrt{2\alpha(1-\alpha)}} \\
     &\xrightarrow{d} \max_{Z,\alpha\in [\alpha_{\min},\alpha_{\max}]} ~\mbox{Gaussian}\left(\sqrt{\frac{M\phi(Z)^2}{2(1-\Phi(Z))}}, \frac{V(Z)}{2}\right) \\
     &\xrightarrow{d} \max_{Z} \left( \sqrt{\frac{M\phi(Z)^2}{2(1-\Phi(Z))}} + \mathbb{W}(\alpha_{\min},\alpha_{\max}) \sqrt{\frac{V(Z)}{2}} \right)\\
     &< 0.45 \sqrt{M} + \frac{ \mathbb{W}(\alpha_{\min},\alpha_{\max})}{\sqrt{2}},
\end{align*}
where the last inequality follows from $\max_Z \sqrt{\frac{\phi(Z)^2}{2(1-\Phi(Z))}} < 0.45$ and $\max_Z V(Z) < 1$.

Now that we have this asymptotic behavior of $F^{NA}(S^{\ast}_{u})$, we recall two results. First, under $H_0$, for our score function of choice in the main text, $F^{BJ}(S) \xrightarrow{a.s.} F^{NA}(S)~\forall S$ (Proposition~\ref{prop:BJ_to_NA}). Second, by Proposition~\ref{prop:na_transform_max}, all other score functions we reference in Appendix~\ref{sec:scoring_functions} are maximizations over continuous and monotonic transformations of $F^{NA}_{\alpha}(S)$. Therefore the limiting distribution of $\max_{S \subseteq U_{X}(D)}F(S)$ under $H_0$ for all our score functions can simply be derived from this specific result for $F^{NA}(S)$, mutatis mutandis.
\end{proof}

\falseposotive
\begin{proof}
For false positive rate $\delta > 0$, we first define $w(\delta)$, which returns $w$ such that $P\left(\mathbb{W}(\alpha_{\min},\alpha_{\max}) > w\right) = \delta$. 
\cite{miller-max_chi_square-1982} show that as $w\rightarrow\infty$,
\begin{equation}
    \label{eq:W(epsilon)-tail-prob}
    P\left(\mathbb{W}(\alpha_{\min},\alpha_{\max}) > w\right) = \left(w\log\left(\frac{\alpha_{\max}(1-\alpha_{\min})}{\alpha_{\min}(1-\alpha_{\max})}\right)+ O(w^{-1})\right)\phi(w),
\end{equation}
and also note that we can use the tables of \cite{keilson-statistical_processes-1975} to obtain the tail probability in~\eqref{eq:W(epsilon)-tail-prob} exactly. Next, we define $h(\delta) = \left(0.45\sqrt{M} + \frac{w(\delta)}{\sqrt{2}}\right)^2$.
As $n\rightarrow\infty$,
\begin{align}
\label{eq:F_NA-upperbound}
P_{H_0}\left(\max_{S \in Rect} F(S) > h(\delta) \right) &\le P_{H_0}\left(\max_{S \subseteq U_{X}(D)}\sqrt{F^{NA}(S)}> 0.45 \sqrt{M} + \frac{w(\delta)}{\sqrt{2}}\right) \\
\label{eq:F_NA_u-upperbound}
 &\le P\left(0.45 \sqrt{M} + \frac{ \mathbb{W}(\alpha_{\min},\alpha_{\max})}{\sqrt{2}} > 0.45 \sqrt{M} + \frac{ w(\delta)}{\sqrt{2}} \right) \\
 &= P\left(\mathbb{W}(\alpha_{\min},\alpha_{\max}) > w(\delta)\right) \nonumber \\
 &= \delta \nonumber
\end{align}
where the inequality in~\eqref{eq:F_NA-upperbound} follows from the fact that the space of rectangular subsets is contained within the space of all subsets, and~\eqref{eq:F_NA_u-upperbound} follows from Lemma~\ref{lem:null_converg}. 
\ignore{Given asymptotic behavior of $F^{NA}(S^{\ast}_{u})$, we can recall two results. First, under $H_0$, for our score function of choice in the main text, $F^{BJ}(S) \xrightarrow{a.s.} F^{NA}(S)~\forall S$ (Proposition~\ref{prop:BJ_to_NA}). Second, by Proposition~\ref{prop:na_transform_max} all other score functions we reference in Appendix~\ref{sec:scoring_functions} are maximizations over continuous and monotonic transformations of $F^{NA}_{\alpha}(S)$. Therefore, for all our score functions, the quantity $ h(\delta)$ that ensures $\lim_{n \rightarrow \infty}P_{H_0}\left(\max_{S \in Rect} F(S) > h(\delta)\right) \le \delta$ can simply be derived from this specific result for $F^{NA}(S)$, mutatis mutandis.} 
\end{proof}

Given asymptotic control over the score when the null hypothesis is true, we now turn our attention to the score when the null hypothesis is false. We begin by first recognizing that for any $\alpha \in [\alpha_{\min},\alpha_{\max}]$, $F_{\alpha}(S^{T}) \le F(S^{T}) \le F(S^{\ast})$. The latter inequality indicates that the score of the optimal rectangular subset is lower-bounded by the score of the truly affected (rectangular) subset, because by definition, no subset achieves a higher score than $S^{\ast}$. The former inequality indicates that the score of the true subset evaluated at a given $\alpha$ lower bounds the score of the true subset maximized over all $\alpha$. Consequently, under $H_1(S^{T})$, the distribution of $F_{\alpha}(S^T)$ for any $\alpha$ provides a lower bound on the distribution of $F(S^T)$, and consequently $F(S^{\ast})$. Therefore we will begin by establishing the distribution of $\sqrt{F_{\alpha^{\ast}}(S^T)}$ under $H_1(S^{T})$, for the specific $\alpha^{\ast} = \arg\max_{\alpha} \frac{(\beta(\alpha)-\alpha)^2}{2 \alpha(1-\alpha)}$.

\altconverg
\begin{proof}
  First, recognize that $N\left(S^T\right) = kMn$ and that $N_{\alpha^{\ast}}(S^T) \sim \mbox{Binomial}\left(N(S^T), \beta^{\ast}\right)$. Therefore, we have the following:
  \begin{align*} 
    F_{\alpha^{\ast}}^{NA}\left(S^T\right) &= \frac{\left(N_{\alpha^{\ast}}\left(S^T\right)-N\left(S^T\right)\alpha^{\ast}\right)^2}{2N\left(S^T\right)\alpha^{\ast}(1-\alpha^{\ast})} \\
    &= \frac{N\left(S^T\right)\left(\frac{N_{\alpha^{\ast}}\left(S^T\right)}{N\left(S^T\right)}-\alpha^{\ast}\right)^2}{2\alpha^{\ast}(1-\alpha^{\ast})} \\
    &= \frac{kMn\left(\hat{\beta}^{\ast} -
    \alpha^{\ast}\right)^2}{2\alpha^{\ast}(1-\alpha^{\ast})}.
  \end{align*}
Next, by the Central Limit Theorem we have
\begin{equation*}
    \sqrt{n}(\hat\beta^* - \beta^*) \stackrel{d}{\longrightarrow}\ \mbox{Gaussian}\left(0,\ \frac{\beta^*(1-\beta^*)}{kM}\right),
\end{equation*}
and therefore, by the delta method we also have
\begin{equation*}
    \sqrt{n}\left(g(\hat\beta^*) - g(\beta^*)\right) \stackrel{d}{\longrightarrow} \mbox{Gaussian}\left(0,\frac{\beta^*(1-\beta^*)}{kM}g'(\beta^*)^2\right).
\end{equation*}
If we allow $g(b) = \sqrt\frac{kM}{2\alpha^*(1-\alpha^*)} (b-\alpha^*)$, we then finally have
\begin{equation*}
    \sqrt{F_{\alpha^{\ast}}^{NA}(S^T)}-\sqrt{\frac{kMn(\beta^{\ast}-\alpha^{\ast})^2}{2\alpha^{\ast}(1-\alpha^{\ast})}} \stackrel{d}{\longrightarrow} \mbox{Gaussian}\left(0,\frac{\beta^{\ast}(1-\beta^{\ast})}{2\alpha^{\ast}(1-\alpha^{\ast})}\right).
\end{equation*}
 From Proposition~\ref{prop:BJ_to_NA} we know that, for our score function of choice in the main text, $F_{\alpha^{\ast}}^{BJ}(S)$ is bounded above and below by either $F_{\alpha^{\ast}}^{NA}(S)$ or a constant times $F_{\alpha^{\ast}}^{NA}(S)~\forall S$. Second, by Proposition~\ref{prop:na_transform_max} all other functions $F_{\alpha}$ we reference in Appendix~\ref{sec:scoring_functions} are continuous and monotonic transformations of $F^{NA}_{\alpha}(S)$. Therefore, the corresponding limiting distributions under $H_1(S^T)$ for all functions $F_{\alpha^{\ast}}(S)$ can simply be derived from this result for $F^{NA}_{\alpha^{\ast}}(S)$, mutatis mutandis.
\end{proof}

\power
\begin{proof}
  First, we note that under $H_1(S^T)$
  \begin{equation*}
    F_{\alpha^{\ast}}\left(S^T\right) \le F\left(S^{\ast}\right),
  \end{equation*}
  because the detected subset $S^{\ast} = \arg\max_{S \in Rect,~ \alpha \in [\alpha_{\min},\alpha_{\max}]}F_\alpha(S)$, while $S^T \in Rect$ and $\alpha^{\ast} \in [\alpha_{\min},\alpha_{\max}]$. Now that we have a lower bound on
  $F\left(S^{\ast}\right)$ under $H_1(S^T)$, we consider the
  critical value $h\left(\delta\right)$ for fixed Type-I error rate $\delta > 0$, and
  $F^{NA}$ or $F^{BJ}$ score functions.
  \begin{align}
    P_{H_1}\left(F(S^{\ast}) > h\left(\delta\right)\right) &\ge  P_{H_1}\left(F_{\alpha^{\ast}}\left(S^T\right) > h\left(\delta\right) \right) \nonumber \\
    \label{eq:w-to-inf-power}
    &=P_{H_1}\left(F_{\alpha^{\ast}}\left(S^T\right) >  \left(0.45 \sqrt{M} + \frac{w(\delta)}{\sqrt{2}}\right)^2 \right) \\
    &=P_{H_1}\left(F_{\alpha^{\ast}}\left(S^T\right) >  \left(0.45\sqrt{M} + O(1) \right)^2 \right) \nonumber \\
    \label{eq:F_a_ast-rej-power}
    &= P_{H_1}\left(\left(O\left(\sqrt{kMn}\right) + Z \sigma^2_{\alpha^{\ast}\beta^{\ast}} \right)^2 >  \left(0.45\sqrt{M} + O(1) \right)^2 \right),
  \end{align}

  \noindent where $Z \sim \mbox{Gaussian}\left(0,1\right)$,~\eqref{eq:w-to-inf-power} follows from Theorem~\ref{thm:false_posotive}, and~\eqref{eq:F_a_ast-rej-power} follows from Lemma
  \ref{lem:alt_converg}.  For $n\rightarrow \infty$ and constant $k$ and $M$, we know that the lhs of~\eqref{eq:F_a_ast-rej-power} goes to $\infty$ while the rhs does not, and thus $P_{H_1}\left(F(S^{\ast}) > h\left(\delta\right)\right) \rightarrow 1$.\ignore{ We also note that although $w(\delta)$ is derived from the approximation in~\eqref{eq:W(epsilon)-tail-prob}, with the assumption $w\longrightarrow\infty$, that from the above derivation in this proof we see that when $H_0$ is false, we expected $w$ to be  regime $\delta \in O(\phi(w)) \in O(e^{-w})$, which implies that $h(\delta) \in O(\ln{\frac{1}{w}})$  we are considering $O\left(\sqrt{kn}\right)\longrightarrow\infty$.} Finally, by Proposition~\ref{prop:na_transform_max}
  all other score functions we reference in Appendix~\ref{sec:scoring_functions} are maximizations over continuous and monotonic transformations of the continuous function
  $F_{\alpha}^{NA}(S)$; therefore, this result will hold mutatis mutandis for these transformations.
\end{proof}

\subsection{Subset Correctness}
\label{sec:subset_correct}
In this section, we are still interested in studying the properties of our framework
under $H_1(S^T)$.  However, we are now concerned about the
correctness of the detected subset $S^{\ast}$: our objective is for 
$S^{\ast}$ to exactly match $S^T$. If $x$ is a data element, i.e., one of the $M$ unique covariate profiles in the data; $U_{X}(D)$
is the collection of these data elements, i.e., $U_{X}(D) =
\{x_{1},\ldots,x_{M}\}$; and both $U_{X}(S^{\ast}),U_{X}\left(S^T\right) \subseteq U_{X}(D)$. The results
in this section are general, and are therefore applicable to an unconstrained
(or constrained) $S^T$; therefore $S^{\ast}$ and $\alpha^{\ast}$ will refer to the joint maximization of subsets and $\alpha$ values
over the unconstrained (or constrained) space in which $S^T$ is defined.  We begin building our theory by demonstrating that the score function of interest can be re-written as an additive function if we condition on the value of the null and alternative hypothesis parameters $\alpha$ and $\beta(\alpha)$. More specifically, the score of a subset $S$ can be decomposed into the sum of contributions (measured by a function $\omega$) from each individual covariate profile $x$ contained within the subset. For example, with respect to $F^{BJ}$, $\omega^{BJ} \left( \alpha, \beta, N_{\alpha}\left(x\right), N\left(x\right) \right) = C^{1}_{\alpha, \beta} ~  N_{\alpha}\left(x\right) +  C^{2}_{\alpha, \beta} ~ N\left(x\right)$, where each $C$ is only a function of $\alpha$ and $\beta$, and therefore constant with respect to $N_{\alpha}\left(x\right)$ and $N\left(x\right)$.
\begin{restatable}{lem}{addfuncs}
  \label{add_funcs}
  $F(S)$ can be written as $\max_{\alpha, \beta}\sum_{x \in U_{X}(S)}{\omega \left( \alpha, \beta, N_{\alpha}\left(x\right), N\left(x\right)
      \right)}$, for $\alpha, \beta \in (0,1)$ representing quantile values of the control and treatment potential outcomes distributions respectively.
\end{restatable}

\begin{proof}
  First we note that from the derivations of $F_{\alpha}^{BJ}(S)$
  and $F^{NA}_{\alpha}(S)$ in Appendix \ref{sec:scoring_functions}, that if we do not set $\beta =
  \beta_{\text{mle}}(S)$ but instead treat $\beta \in (0,1)$ as a given
  quantity, then
   \begin{align*}
   F^{BJ}(S) 
    &= \max_{\alpha,\beta}F_{\alpha,\beta}^{BJ}(S)\\
  &= \max_{\alpha,\beta}N_{\alpha}(S) \log  \left(\frac {\beta} {\alpha} \right) + \left( N(S)-N_{\alpha}(S) \right)
  \log \left( \frac{1-\beta}{1-\alpha} \right) \\
  &=\max_{\alpha,\beta} N_{\alpha}(S) \log  \left(\frac {\beta (1-\alpha)} {\alpha(1-\beta)} \right) + N(S) \log \left( \frac{1-\beta}{1-\alpha} \right) \\
  &= \max_{\alpha,\beta} \log  \left(\frac {\beta (1-\alpha)} {\alpha(1-\beta)} \right) \left( \sum_{x \in U_{X}(S)}{N_{\alpha}(x)} \right)+  \log \left( \frac{1-\beta}{1-\alpha} \right) \left( \sum_{x \in U_{X}(S)}{N(x)} \right) \\
  &=\max_{\alpha,\beta}  \sum_{x \in U_{X}(S)}{ \log \left(\frac {\beta (1-\alpha)} {\alpha(1-\beta)} \right) N_{\alpha}(x) +  \log \left( \frac{1-\beta}{1-\alpha} \right) N(x)}\\
  &= \max_{\alpha,\beta} \sum_{x \in U_{X}(S)}{ C^{BJ_1}_{\alpha, \beta} ~ N_{\alpha}(x) +  C^{BJ_2}_{\alpha, \beta} ~ N(x)}\\
  &=\max_{\alpha,\beta}\sum_{x \in U_{X}(S)} \omega^{BJ}\big( \alpha, \beta, N_{\alpha}(x), N(x) \big)
  \end{align*}
  \begin{align*}
  F^{NA}(S)
  &= \max_{\alpha,\beta}F_{\alpha,\beta}^{NA}(S)\\
  &= \max_{\alpha,\beta} \frac{N_{\alpha}(S)\left(\beta-\alpha\right)}{\alpha(1-\alpha)} + \frac{N(S)\left(\alpha^2-\beta^2\right)}{2\alpha(1-\alpha)} \\
  &= \max_{\alpha,\beta}\frac{\left(\beta-\alpha\right)}{\alpha(1-\alpha)} \left( \sum_{x \in U_{X}(S)}{N_{\alpha}(x)} \right) + \frac{\left(\alpha^2-\beta^2\right)}{2\alpha(1-\alpha)} \left( \sum_{x \in U_{X}(S)}{N(x)} \right) \\
  &=  \max_{\alpha,\beta}\sum_{x \in U_{X}(S)}{ \frac{\left(\beta-\alpha\right)}{\alpha(1-\alpha)} N_{\alpha}(x) +  \frac{\left(\alpha^2-\beta^2\right)}{2\alpha(1-\alpha)} N(x)}\\
  &=  \max_{\alpha,\beta}\sum_{x \in U_{X}(S)}{ C^{NA_1}_{\alpha, \beta} N_{\alpha}(x) +  C^{NA_2}_{\alpha, \beta} N(x)}\\
  &=\max_{\alpha,\beta}\sum_{x \in U_{X}(S)} \omega^{NA}\big( \alpha, \beta, N_{\alpha}(x), N(x) \big)
 \end{align*}

  where all the $C_{\alpha, \beta}$'s are constants with respect to given
  values of $\alpha,\beta$.
\end{proof}

We now have that the score of a subset $S$ can be decomposed into the sum of
contributions (measured by a function $\omega$) from each individual element
contained within the subset. Next, we seek to demonstrate some important
properties of the $\omega$ functions. More specifically, $\omega$ is a concave
function with respect to $\beta$, which has two roots and a unique maximum.

\begin{restatable}{lem}{naconcave}
  \label{lem:na-concave}
  $\omega^{NA}\left( \alpha, \beta, N_{\alpha}\left(x\right), N\left(x\right)
      \right)$ is
  concave with respect to $\beta$, maximized at $\beta_{\text{mle}}(x) =
  \frac{N_{\alpha}\left(x\right)}{N\left(x\right)}$, and has two roots $\left(\beta_{\min}(x),
  \beta_{\max}(x)\right)$.
\end{restatable}
\begin{proof}
  Firstly,
  \begin{align}
    \frac{\partial~\omega^{NA}\big( \alpha, \beta, N_{\alpha}(x), N(x)
      \big)}{\partial \beta} &= \frac{N_{\alpha}(x)-N(x)\beta}{\alpha(1-\alpha)}
      \nonumber \\
    &= -\frac{N(x)}{\alpha(1-\alpha)}\beta+\frac{N_{\alpha}(x)}{\alpha(1-\alpha)} \label{eq:na-line} \\
      {\text{ (set) }} ~ 0 &=  -\frac{N(x)}{\alpha(1-\alpha)}\beta+\frac{N_{\alpha}(x)}{\alpha(1-\alpha)} \nonumber \\
    0 &= -N(x)\beta+N_{\alpha}(x) \nonumber \\
    \beta &= \frac{N_{\alpha}(x)}{N(x)}, \label{eq:na-root}
  \end{align}
  \eqref{eq:na-line} shows that the first derivative is the equation of a line,
  with a negative slope, and \eqref{eq:na-root} shows that this line has one
  root at $\frac{N_{\alpha}(x)}{N(x)}$. This implies $\omega^{NA}$ is
  concave with respect to $\beta$, with at most two roots which we will refer to
  as $\beta_{\min}(x)$ and $\beta_{\max}(x)$, and is maximized at
  $\frac{N_{\alpha}(x)}{N(x)}$.
\end{proof}

\noindent We now show the same result for $\omega^{BJ}.$
\begin{lem}
  \label{lem:bj-concave}
  $\omega^{BJ}\big( \alpha, \beta, N_{\alpha}(x), N(x) \big)$ is
  concave with respect to $\beta$, maximized at $\beta_{\text{mle}}(x) =
  \frac{N_{\alpha}(x)}{N(x)}$, and has two roots $\left(\beta_{\min}(x),
  \beta_{\max}(x)\right)$.
\end{lem}
\begin{proof}
  \begin{align*}
  \frac{\partial~\omega^{BJ}\big( \alpha, \beta, N_{\alpha}(x), N(x) \big)}{\partial \beta} &= \frac{N_{\alpha}(x)-N(x)\beta}{\beta(1-\beta)} \\
    {\text{ (set)}} ~ 0 &= \frac{N_{\alpha}(x)-N(x)\beta}{\beta(1-\beta)} \\
    0 &= N_{\alpha}(x)-N(x)\beta \\
    \beta &= \frac{N_{\alpha}(x)}{N(x)}
  \end{align*}
  shows that $\omega^{BJ}$ is maximized (if it is concave) at
  $\frac{N_{\alpha}(x)}{N(x)}$ and has at most two roots, which we will refer to
  as $\beta_{\min}(x)$ and $\beta_{\max}(x)$. Additionally,
  \begin{align*}
    \frac{\partial^2 ~\omega^{BJ}\big( \alpha, \beta, N_{\alpha}(x), N(x)
      \big)}{\partial^2 \beta}\Biggr\rvert_{\beta = \frac{N_{\alpha}(x)}{N(x)}}
    &= -\frac{\beta^2N(x)+(1-2\beta)N_{\alpha}(x)}{(\beta-1)^2\beta^2}\Biggr\rvert_{\beta =
      \frac{N_{\alpha}(x)}{N(x)}} \\
    &< 0
  \end{align*}
  shows that $\omega^{BJ}$ is concave with respect to $\beta$.
\end{proof}

Intuitively, $(\beta_{\min}(x),\beta_{\max}(x))$ is the interval over which $\omega$ makes a positive contribution to the score of a subset, while this contribution is maximized at $\beta_{\text{mle}}(x)$; we note that in the case of $\omega^{NA}$ and $\omega^{BJ}$, $\beta_{\min}(x) = \alpha$. Given that we have demonstrated that $\omega$ is concave, we now demonstrate a
key insight about the relationship between $r_{\max} = \beta_{\max}(x) - \alpha$ and $r_{\text{mle}} = \beta_{\text{mle}}(x) - \alpha$.

\begin{restatable}{lem}{maxtomle}
  \label{lem:na-max-to-mle}
    With respect to $\omega^{NA}\left( \alpha, \beta, N_{\alpha}(x), N(x)
      \right)$, $\frac{r_{\max}(x)}{r_{\text{mle}}(x)} = 2$.
\end{restatable}.
\begin{proof}
  First, by Lemma \ref{lem:na-concave}, we know that, with respect to $\beta$,
  $\omega^{NA}$ is concave and has at most two roots $\left(\beta_{\min}(x),
  \beta_{\max}(x)\right)$. Therefore, we have the following:
  \begin{align*}
  \omega^{NA}\big( \alpha, \beta, N_{\alpha}(x), N(x) \big)&= \frac{N_{\alpha}(x)\left(\beta-\alpha\right)}{\alpha(1-\alpha)} + \frac{N(x)\left(\alpha^2-\beta^2\right)}{2\alpha(1-\alpha)} \\
  {\text{ (set)}}~0 &= \frac{N_{\alpha}(x)\left(\beta-\alpha\right)}{\alpha(1-\alpha)} + \frac{N(x)\left(\alpha^2-\beta^2\right)}{2\alpha(1-\alpha)} \\
  &=  2N_{\alpha}(x)\left(\beta-\alpha\right) + N(x)\left(\alpha^2 - \beta^2\right) \\
  &=  \left( -N(x) \right)\beta^2 + \left( 2N_{\alpha}(x) \right)\beta + \left( -2\alpha N_{\alpha}(x) + N(x)\alpha^2 \right) \\ \\
  \{\beta_{\min}(x), \beta_{\max}(x)\} &= \frac{ -2N_{\alpha}(x) \pm \sqrt{ \left( 2N_{\alpha}(x) \right)^2 -4 \left( -N(x) \right) \left( -2N_{\alpha}(x)\alpha + N(x)\alpha^2 \right)} }{-2N(x)} \\
  &= \frac{ -2N_{\alpha}(x) \pm \sqrt{ 4 \left(N_{\alpha}(x)^2 -2N_{\alpha}(x)N(x)\alpha + (N(x)\alpha)^2 \right)} }{-2N(x)} \\
  &= \frac{ -2N_{\alpha}(x) \pm \sqrt{ 4\left( N_{\alpha}(x) -N(x)\alpha \right)^2 } }{-2N(x)} \\
  &= \frac{ N_{\alpha}(x) \pm \left( N_{\alpha}(x) - N(x)\alpha \right) }{ N(x) } \\
  &= \{ \alpha, 2\beta_{\text{mle}}(x) - \alpha \}.
\end{align*}
This implies that $\beta_{\max}(x) - \alpha = 2\left(\beta_{\text{mle}}(x) -
\alpha\right)$ and thus $r_{\max}(x) = 2r_{\text{mle}}(x)$, with respect to
$\omega^{NA}.$
\end{proof}

\noindent We show a similar result for $\omega^{BJ}.$
\begin{lem}
  \label{lem:bj-max_to_mle}
  With respect to $\omega^{BJ}\big( \alpha, \beta, N_{\alpha}(x), N(x)
  \big)$,
  \begin{equation*}
    \frac{r_{\max}(x)}{r_{\text{mle}}(x)}
    \begin{cases}
     < 2 & \text{if } \beta_{\text{mle}}(x) > \frac{1}{2} \\
     = 2 & \text{if } \beta_{\text{mle}}(x) = \frac{1}{2} \\
     > 2 & \text{otherwise}.\\
    \end{cases}
  \end{equation*}
\end{lem}

\begin{proof}
First, by Lemma~\ref{lem:bj-concave}, we know that, with respect to $\beta$,
$\omega^{BJ}$ is concave and has at most two roots $\left(\beta_{\min}(x),
\beta_{\max}(x)\right)$. One of the solutions of $\omega^{BJ}$ must be
$\alpha$, so let us assume that $\beta_{\min}(x) = \alpha$; this will be true
when $\beta > \alpha$, which intuitively corresponds to our case of interest:
when the covariate profile contains more significant (extreme) $p$-values than
expected. Furthermore, we know that $\omega^{BJ}$ achieves a maximum at
$\beta_{\text{mle}} = \frac{N_\alpha(x)}{N(x)}$. With these properties we can
show the first case ($1 \le \frac{r_{\max}(x)}{r_{\text{mle}}(x)} < 2$) by first
recognizing that trivially $\beta_{\text{mle}} \leq \beta_{\max}$, and
$\beta_{\text{mle}} -\alpha \leq \beta_{\max} -\alpha $. To show the upper bound
of the first case, it suffices to show that $\omega^{BJ}\left(\alpha,
\beta_{\text{mle}} - \epsilon, N_\alpha(x), N(x)\right) \geq
\omega^{BJ}\left(\alpha, \beta_{\text{mle}} + \epsilon, N_\alpha(x), N(x)\right)$
for some $\epsilon >0$. The essential implication is that the concave function
$\omega^{BJ}$ increases at a slower rate (until it reaches its maximum)
than it decreases. This further implies that the distance between
$\beta_{\text{mle}}$ and $\alpha$ is larger than the distance between $\beta_{\text{mle}}$ and $\beta_{\max}$, and therefore the desired result.

Recall from Lemma~\ref{lem:bj-concave} that
\begin{align*}
  \frac{\partial~\omega^{BJ}\big( \alpha, \beta, N_{\alpha}(x), N(x)
    \big)}{\partial \beta} &= \frac{N_{\alpha}(x)-N(x)\beta}{\beta(1-\beta)} \\
  &= N(x) \left[ \frac{\beta_{\text{mle}}(x)-\beta}{\beta(1-\beta)} \right],
\end{align*}
which means the slope of $\omega^{BJ}$ is proportional to
$\frac{\beta_{\text{mle}}(x)-\beta}{\beta(1-\beta)}$. We now compare the slope
around the inflection point $\beta_{\text{mle}}(x)$, and recognize that
at 
$\beta = \beta_{\text{mle}}(x)+\epsilon$ the slope is negative with
  absolute value proportional to $\frac{\epsilon} { \left(
  \beta_{\text{mle}}(x) + \epsilon \right)  \left( 1-\beta_{\text{mle}}(x) -
  \epsilon \right) }$. At
$\beta = \beta_{\text{mle}}(x)-\epsilon$ the slope is positive with
  absolute value proportional to $\frac{\epsilon} { \left(
  \beta_{\text{mle}}(x) - \epsilon \right)  \left( 1-\beta_{\text{mle}}(x) +
  \epsilon \right) }$.
Therefore,
\begin{alignat*}{3}
  \beta_{\text{mle}}(x) > \frac{1}{2}  &\Longleftrightarrow & \left(
    \beta_{\text{mle}}(x) + \epsilon \right)  \left( 1-\beta_{\text{mle}}(x) -
    \epsilon \right) &< \left( \beta_{\text{mle}}(x) - \epsilon \right)
    \left( 1-\beta_{\text{mle}}(x) + \epsilon \right) &\\
  &\Longleftrightarrow &  \frac{\epsilon} { \left( \beta_{\text{mle}}(x) + \epsilon
    \right)  \left( 1-\beta_{\text{mle}}(x) - \epsilon \right) } &> \frac{\epsilon} { \left( \beta_{\text{mle}}(x) - \epsilon
    \right)  \left( 1-\beta_{\text{mle}}(x) + \epsilon \right) } &\\
  &\Longleftrightarrow & \frac{r_{\max}(x)}{r_{\text{mle}}(x)} &< 2. &
\end{alignat*}
The demonstration of the remaining two conditions follow precisely the same approach
above, mutatis mutandis.
\end{proof}

Now that we have built up the necessary properties of the $\omega$ functions, we
now will discuss the sufficient conditions for the detected subset to be exactly
correct, $S^{\ast} = S^T$. To begin we re-introduce some additional notation:
\begin{align*}
  r^{\text{aff}}_{\text{mle}-h} &= \max_{x \in U_{X}(S^T)} r_{\text{mle}}(x),\\
  r^{\text{aff}}_{\text{mle}-l} &= \min_{x \in U_{X}(S^T)} r_{\text{mle}}(x),\\
  r^{\text{unaff}}_{\text{mle}-h} &= \max_{x \not\in U_{X}(S^T)} r_{\text{mle}}(x),\\
  \eta &=  \left( \frac{\sum_{x \in U_{X}(S^T)}{N(x)} }{ \sum_{x \in U_{X}(D)}{N(x)} } \right),\\
  \nu-homogeneous &\colon \frac{r^{\text{aff}}_{\text{mle}-h}}{r^{\text{aff}}_{\text{mle}-l}} < \nu, \\
  \delta-strong &\colon \frac{r^{\text{aff}}_{\text{mle}-l}}{r^{\text{unaff}}_{\text{mle}-h}} > \delta,\\
  R &\colon (0,1) \mapsto (0,1).
\end{align*}
More specifically, $R$ is an invertible function such that  $R \colon
r_{\max}(x) \mapsto r_{\text{mle}}(x)$--i.e., if $R$ is applied to
$r_{\max}(x)$ it would produce the corresponding $r_{\text{mle}}(x)$. From
Lemma~\ref{lem:na-max-to-mle} we know that with respect to $\omega^{NA}$,
$R^{NA}(r) = \frac{r}{2}$, while from Lemma \ref{lem:bj-max_to_mle} we
know that with respect to $\omega^{BJ}$, $R^{BJ}(r) \le
\frac{r}{2}$ under certain conditions.

The first result we provide is a sufficient condition for guaranteeing that the
detected subset includes all the elements from the true subset ($S^{\ast}
\supseteq S^T$). More specifically, we show that such a condition is sufficient
homogeneity of the affected data elements: for a given value
$\nu$, and any pair of affected covariate profiles $(x_i, x_j \in U_{X}(S^T))$, the anomalous signal $r_\text{mle}(x)$ observed in $x_i$ is no more than $\nu$ times that which is observed in $x_j$.

\homo
\begin{proof}
  First, let $\{ x_{(1)}, \ldots, x_{(t)} \}$ be the data elements in $S^T$
  sorted by the priority function (Theorem \ref{thm:LTSS}) $G(x) =
  \frac{N_{\alpha}(x)}{N(x)} = \beta_{\text{mle}}(x)$. By the
  assumption of an observed signal that is at least $1$-strong, these data
  elements are the $t$ highest priority data elements. Additionally, let $\nu =
  \frac{r^{\text{aff}}_{\text{mle}-h}} {R \left( r^{\text{aff}}_{\text{mle}-h}
  \right)}.$ Therefore,
  \begin{alignat*}{3}
    \nu-homogeneous &\implies & \nu &> \frac{r^{\text{aff}}_{\text{mle}-h}}{r^{\text{aff}}_{\text{mle}-l}} &\ignore{\quad \left( \text{by definition} \right)}\\
      &\therefore & \frac{r^{\text{aff}}_{\text{mle}-h}}{R \left( r^{\text{aff}}_{\text{mle}-h} \right)} &> \frac{r^{\text{aff}}_{\text{mle}-h}}{r^{\text{aff}}_{\text{mle}-l}} &\ignore{\quad \left( \text{by assumption of } \nu \right)} \\
      &\implies & r^{\text{aff}}_{\text{mle}-l} &> R\left( r^{\text{aff}}_{\text{mle}-h} \right) &\\
      &\implies & R^{-1}\left( r^{\text{aff}}_{\text{mle}-l} \right) &> r^{\text{aff}}_{\text{mle}-h} &\\
      &\implies & \beta_{\max}(x_{(t)}) - \alpha &> \beta_{\text{mle}}(x_{(1)}) - \alpha &\\
      &\implies & \beta_{\max}(x_{(t)}) &> \beta_{\text{mle}}(x_{(k)}) &\quad \left( \forall k \right) \\
      &\implies & \beta_{\max}(x_{(t)}) &> \beta_{\text{mle}}(S^{\ast}) &\\
      &\therefore & \omega\big( \alpha, \beta_{\text{mle}}(S^{\ast}), N_{\alpha}(x_{(t)}), N(x_{(t)}) \big) &> 0 & \\
      &\therefore & |S^{\ast}| &\ge t &\\
      &\therefore & S^{\ast} &\supseteq S^T. &
  \end{alignat*}

  Intuitively, $\beta_{\text{mle}}(x_{(t)})$ and $\beta_{\text{mle}}(x_{(1)})$
  are respectively the smallest and largest $\beta_{\text{mle}}$ of all the $x\in
   U_{X}(S^T)$. Furthermore, $\beta_{\text{mle}}(x_{(t)}) \le
  \beta_{\text{mle}}(x_{(k)}) \le \beta_{\max}(x_{(t)})~\forall k\in[1,t]$, which means
  $\beta_{\text{mle}}(S^{\ast}) \le \beta_{\max}(x_{(t)})$ for the optimal
  subset $S^{\ast}$. Moreover, the $S^{\ast}$ that maximizes $F_{\alpha,\beta}$ will
  include any covariate profile $x$ that would make a positive contribution to the
  score $F_{\alpha,\beta}$ at the given value of $\beta$. Such a
  positive contribution occurs when the concave $\omega$ function of $x$ is
  positive. At the optimal $\alpha$ and $\beta = \beta_{\text{mle}}(S^{\ast})$
  the $\omega$ function for each of the $\{x_{(1)}, \ldots, x_{(t)}\}$ is
  positive because $\beta_{\max}$ (the larger root of the $\omega$ functions)
  for each of these elements is greater than $\beta_{\text{mle}}(S^{\ast})$.
  \end{proof}
\begin{cor}
  From Lemma \ref{lem:na-max-to-mle} we know that with respect to
  $\omega^{NA}$, $\frac{r} {R \left( r \right)} = 2$. Additionally, from
  Lemma \ref{lem:bj-max_to_mle} we know that with respect to
  $\omega^{BJ}$, $\frac{r} {R \left( r \right)} \le 2$ under certain
  conditions. Therefore, we can conclude that at $\alpha^{\ast}$, $2$-homogeneity (and $1$-strength)
  is sufficient for $S^{\ast} \supseteq S^T$ with respect to $F^{NA}$; to
  $F^{BJ}$, under some conditions; and to the other score functions described
  above, by Proposition~\ref{prop:na_transform_max}. Essentially,
  if the observed excess proportions of $p$-values significant at $\alpha^{\ast}$ vary by no more than a factor of 2 across
  all of the affected $x \in U_{X}(S^T)$,
  then the detected subset will include all of the affected data elements.
\end{cor}

The next result we provide is a sufficient condition for guaranteeing that the
detected subset will only include elements from the true subset ($S^{\ast}
\subseteq S^T$). More specifically, we show that such a condition is sufficient
strength of the affected data elements; or intuitively, for a given value
$\delta$, the anomalous signals $r_\text{mle}(x)$ observed in every affected data element are more than $\frac{\delta}{\eta}$-times that of the unaffected data elements.

\strength
\begin{proof}
  First, let $D = \{ x_{(1)}, \ldots, x_{(t)}, x_{(t+1)}, \ldots, x_{(M)}\}$ be the
  data elements sorted by the priority function (Theorem
  \ref{thm:LTSS}) $G(x) = \frac{N_{\alpha}(x)}{N(x)} =
  \beta_{\text{mle}}(x)$. By the assumption of $\delta > 1$ (an observed signal
  that is at least $1$-strong), $S^T = \{x_{(1)}, \ldots, x_{(t)}\}$.
  Additionally, if  $r^{\text{unaff}}_{{\text{mle}-h}} \le 0$ then $S^{\ast} \subseteq S^T$, trivially. Therefore, we assume $r^{\text{unaff}}_{{\text{mle}-h}} >0$. Let $\delta = \frac{R^{-1} \left(
  r^{\text{unaff}}_{\text{mle}-h} \right)}{r^{\text{unaff}}_{{\text{mle}-h}}}.$
  Therefore,
  \small
  \begin{alignat*}{3}
    \frac{\delta}{\eta}-strong &\implies & \frac{\delta}{\eta} &< \frac{r^{\text{aff}}_{\text{mle}-l}}{r^{\text{unaff}}_{\text{mle}-h}} &\ignore{\quad \left( \text{by definition} \right)}\\
      &\therefore & \frac{R^{-1} \left( r^{\text{unaff}}_{\text{mle}-h} \right)}{\eta r^{\text{unaff}}_{\text{mle}-h}} &< \frac{r^{\text{aff}}_{\text{mle}-l}}{r^{\text{unaff}}_{\text{mle}-h}} &\ignore{\quad \left( \text{by assumption of } \delta \right)} \\
      &\implies & R^{-1}\left(r^{\text{unaff}}_{\text{mle}-h} \right) &< \left( \frac{\sum_{x \in U_{X}(S^T)}{N(x)} }{ \sum_{x \in U_{X}(D)}{N(x)} } \right) r^{\text{aff}}_{\text{mle}-l} &\\
      & & &=\frac{\sum_{x \in U_{X}(S^T)}{ r^{\text{aff}}_{\text{mle}-l}N(x)} }{ \sum_{x \in U_{X}(D)}{N(x)} } &\\
      & & &\le \frac{\sum_{x \in U_{X}(S^T)}{ r_{\text{mle}}(x)N(x)} }{ \sum_{x \in U_{X}(D)}{N(x)} } ~~\left( \text{since} ~ r_{\text{mle}}(x) \ge r^{\text{aff}}_{\text{mle}-l} \ignore{~~ \forall x \in U_{X}(S^T)}  \right) &\\
      & & &\le\frac{ \sum_{x \in U_{X}(S^T)}{ r_{\text{mle}}(x)N(x)} + \sum_{x \not\in U_{X}(S^T)}{ r_{\text{mle}}(x)N(x)} }{ \sum_{x \in U_{X}(D)}{N(x)} } &\\
      & & &=\frac{ \sum_{x \in U_{X}(D)}{ r_{\text{mle}}(x)N(x)} }{ \sum_{x \in U_{X}(D)}{N(x)} } &\\
      & & &=\frac{ \sum_{x \in U_{X}(D)}{ \left( \frac{ N_{\alpha}(x) }{ N(x) } - \alpha \right) N(x) } }{ \sum_{x \in U_{X}(D)}{N(x)} } &\\
      & & &=\frac{ \sum_{x \in U_{X}(D)}{ N_{\alpha}(x) - N(x)\alpha  } }{ \sum_{x \in U_{X}(D)}{N(x)} } &\\
      & & &=\frac{ \sum_{x \in U_{X}(D)}{ N_{\alpha}(x) } - \sum_{x \in U_{X}(D)}{ N(x)\alpha } }{ \sum_{x \in U_{X}(D)}{N(x)} } &\\
      & & &=\frac{ \sum_{x \in U_{X}(D)}{ N_{\alpha}(x) } }{ \sum_{x \in U_{X}(D)}{N(x)} } - \alpha &\\
      &\therefore & \beta_{\max}(x_{(t+1)}) - \alpha &< \beta_{\text{mle}}(D) - \alpha  &\\
      &\implies & \beta_{\max}(x_{(t+1)}) &< \beta_{\text{mle}}(x_{(t)}) &\\
      &\implies & \beta_{\max}(x_{(t+1)}) &< \beta_{\text{mle}}(S^{\ast}) &\\
      &\therefore & \omega\big( \alpha, \beta_{\text{mle}}(S^{\ast}), N_{\alpha}(x_{(t+1)}), N(x_{(t+1)}) \big) &< 0 &\\
      &\implies & |S^{\ast}| &\le t &\\
      &\therefore & S^{\ast} &\subseteq S^T &\\
  \end{alignat*}
  \normalsize

  Intuitively, $\beta_{\text{mle}}(x_{(t)})$ and $\beta_{\text{mle}}(x_{(t+1)})$
  are respectively the smallest affected and largest unaffected
  $\beta_{\text{mle}}$ values. Furthermore, $\beta_{\max}(x_{(t+1)}) \le
  \beta_{\text{mle}}(x_{(t)})$, which means $\beta_{\max}(x_{(t+1)}) \le
  \beta_{\text{mle}}(S^{\ast})$ for the optimal subset $S^{\ast}$. Moreover, the
  $S^{\ast}$ that maximizes $F_{\alpha,\beta}$ will not include any data element $x$
  that has $\omega \le 0$ and thus
  makes a non-positive contribution to the score $F_{\alpha,\beta}$ at the
  given value of $\beta$. At the optimal
  $\alpha$ and $\beta = \beta_{\text{mle}}(S^{\ast})$ the $\omega$ function for
  each of the $\{x_{(t+1)}, \ldots, x_{(M)}\}$ are non-positive because
  $\beta_{\max}$ (the larger root of the $\omega$ functions) for each of these
  elements is less than $\beta_{\text{mle}}(S^{\ast})$.
\end{proof}

\begin{cor}
  From Lemma~\ref{lem:na-max-to-mle} we know that with respect to
  $\omega^{NA}$, $\frac{R^{-1} \left( r \right)} {r} = 2$. Additionally,
  from Lemma~\ref{lem:bj-max_to_mle} we know that with respect to
  $\omega^{BJ}$, $\frac{R^{-1} \left( r \right)} {r} \ge 2$ under certain
  conditions. Therefore, we can conclude that at $\alpha^{\ast}$, $\frac{2}{\eta}$-strength is sufficient for
  $S^{\ast} \subseteq S^T$ with respect to $F^{NA}$; to $F^{BJ}$, under
  some conditions; and to the other score functions described above, by Proposition~\ref{prop:na_transform_max}. Essentially,
  if the observed excess proportions of $p$-values significant at $\alpha^{\ast}$ across
  all of the $x \in  U_{X}(S^T)$ are at least $\frac{2}{\eta}$ times larger than the observed excess proportions for $x \not\in
   U_{X}(S^T)$, then the detected subset will only include affected data elements.
\end{cor}

\ignore{
\begin{thm}
  \label{thm:subset_exact}
  Under $H_1(S^T)$, where $S^T$ is a $t$-element subset of covariate profiles and for each element the true potential outcome distributions are unequal, $\exists~\nu,\delta > 1$ such that if the observed effect (as measured by $\omega$) across these covariate profiles is $\nu-homogeneous$ and $\frac{\delta}{\eta}-strong$, then $S^{\ast} = S^T$.
\end{thm}
\begin{proof}
  \begin{alignat*}{3}
    \because~& & \nu-homogeneous \implies S^{\ast} &\supseteq S^{T} &\quad(\text{by Theorem } \ref{thm:subset_homo})\\
    \because~& & \frac{\delta}{\eta}-strong \implies S^{\ast} &\subseteq S^{T} &\quad(\text{by Theorem } \ref{thm:subset_strength})\\
    \therefore~& &S^{\ast} &= S^T
  \end{alignat*}
\end{proof}
}

It follows from the above corollaries that 2-homogeneity and $\frac{2}{\eta}$-strength are sufficient for $S^\ast = S^T$ with respect to $F^{NA}$; to $F^{BJ}$, under some conditions; and to the other score functions described above, by Proposition~\ref{prop:na_transform_max}.

\assymptseteq
\begin{proof} From Theorems~\ref{thm:subset_homo} and~\ref{thm:subset_strength}, there exist constants $\nu > 1$ and $\delta > 1$ such that, if the observed effect on $S^T$ is $\nu-homogeneous$ and $\frac{\delta}{\eta}-strong$, then $S^\ast = S^T$.  (For example, for the $F^{NA}$ score function, we have shown above that $\nu = \delta = 2$.)  We show that, for any $\nu > 1$ and $\delta > 1$, as $n \xrightarrow[]{} \infty$, the probability that the observed effect is $\nu-homogeneous$ goes to 1, and the probability that the observed effect is $\frac{\delta}{\eta}-strong$ goes to 1.  Thus, as $n \xrightarrow[]{} \infty$, the observed effect is $\nu-homogeneous$ and $\frac{\delta}{\eta}-strong$ with high probability for the specific $\nu$ and $\delta$ from Theorems~\ref{thm:subset_homo} and~\ref{thm:subset_strength}, and thus $P(S^\ast = S^T) \rightarrow 1$. 
   
As a first step, we show that $F^\ast = F(S^\ast)$ is maximized for some value $\alpha^\ast$ such that $\beta(\alpha^\ast) > \alpha^\ast$ and thus $\beta(\alpha^\ast)-\alpha^\ast > 0$. It follows from Lemma~\ref{lem:alt_converg} that $F_{\alpha}(S^\ast) \rightarrow \infty$ as $n\rightarrow \infty$ if $\beta(\alpha) > \alpha$. From Lemma~\ref{lem:null_converg} we know that $F_{\alpha}(S^\ast)$ is upper bounded by a constant as $n\rightarrow \infty$ if $\beta(\alpha) \le \alpha$, and this remains true when maximizing over the entire range of $\alpha$ values (i.e., under the null hypothesis).  Thus as $n\rightarrow \infty$, the maximum score must occur for some $\alpha^\ast$ with $\beta(\alpha^\ast) > \alpha^\ast$, and we assume this value of $\alpha$ for the remainder of the proof.
   
   Next, given $r_{\text{mle}}(x) = \beta_{\text{mle}}(x) - \alpha$, we show that, as $n\xrightarrow[]{}\infty$, $r_{\text{mle}}(x) \xrightarrow[]{} \beta(\alpha) - \alpha > 0$ for all $x \in U_{X}(S^T)$, and $r_{\text{mle}}(x) \xrightarrow[]{} 0$ for all $x \not\in U_{X}(S^T)$. Note that $\beta_{\text{mle}}(x) = \frac{N_{\alpha}(x)}{N(x)}$, where $N_{\alpha}(x) \sim \text{Binomial} (N(x),p)$ with $p = \beta(\alpha)$ for all $x \in U_{X}(S^T)$ and $p = \alpha$ for all $x \not\in U_{X}(S^T)$, according to $H_1(S^T)$. Therefore, by the law of large numbers we can see that $\beta_{\text{mle}}(x) \xrightarrow[]{} \beta(\alpha)$ for $x \in U_{X}(S^T)$ and $\beta_{\text{mle}}(x) \xrightarrow[]{} \alpha$ for $x \not\in U_{X}(S^T)$. \ignore{Since the $\omega$ function is concave with $\alpha$ as its smaller root, we have $\beta_{\text{mle}}(x) \geq \alpha$, and therefore, $r_{\text{mle}}(x) \xrightarrow[]{} 0^+$ for all $x \not\in U_{X}(S^T)$.}
   
   Next, we show $r^{\text{aff}}_{\text{mle}-l} \xrightarrow[]{} \beta(\alpha) - \alpha$ and $r^{\text{aff}}_{\text{mle}-h} \xrightarrow[]{} \beta(\alpha) - \alpha$ as $n \xrightarrow[]{} \infty$. Let $i$ be the index of elements in the set $U_{X}(S^T)$ and therefore $1 \leq i \leq t$. Therefore, we have $P \left( |r^{i}_{\text{mle}}(x) - \left(\beta(\alpha) - \alpha \right) |  > \epsilon \right) \xrightarrow[]{} 0$. Thus we can see that
   \begin{align*}
       P\left(| r^{\text{aff}}_{\text{mle}-h} - \left(\beta(\alpha) - \alpha \right) | > \epsilon\right)
       &= P \left( | \max_{1\leq i\leq t} \left(r^{i}_{\text{mle}}(x) - \left(\beta(\alpha) - \alpha \right)\right)|  > \epsilon \right) \\
       &\leq P \left( \max_{1\leq i\leq t} |\left(r^{i}_{\text{mle}}(x) - \left(\beta(\alpha) - \alpha \right)\right)|  > \epsilon \right)\\
       &= P \left(\bigcup_{i=1}^{t} \left\{ |r^{i}_{\text{mle}}(x) - \left(\beta(\alpha) - \alpha \right)|  > \epsilon \right\} \right)\\
       &\leq \sum_{i=1}^{t} P \left( |r^{i}_{\text{mle}}(x) - \left(\beta(\alpha) - \alpha \right) |  > \epsilon \right)\\
       &\xrightarrow[]{} 0.
   \end{align*}
   
   The last convergence is due to a fixed value of $t$. Using similar reasoning, we can show  $r^{\text{aff}}_{\text{mle}-l} \xrightarrow[]{} \beta(\alpha) - \alpha$ and therefore, $\frac{r^{\text{aff}}_{\text{mle}-h}}{r^{\text{aff}}_{\text{mle}-l}} \xrightarrow[]{} \frac{\beta(\alpha)}{\beta(\alpha)} = 1$ by Slutzky's theorem. That is, for any $\nu>1$, as $n \xrightarrow[]{} \infty$, the probability that $\frac{r^{\text{aff}}_{\text{mle}-h}}{r^{\text{aff}}_{\text{mle}-l}} < \nu$, and thus that the observed effect on $S^T$ is $\nu-homogeneous$, goes to 1.
   
   Next, we focus on the case where $\exists~ x \not\in U_X(S^T)$ s.t. $r_{\text{mle}}(x) >0$, since $S^\ast \subseteq S^T$ holds trivially for the case when $r_{\text{mle}}(x) \leq 0~\forall x \not\in U_X(S^T)$. We now show that for a fixed $M$, as $n \xrightarrow[]{} \infty$, $r^{\text{unaff}}_{\text{mle}-h} \xrightarrow[]{} 0^+$. Let $j$ be the index of elements in the set $x\notin U_{X}(S^T)$, therefore, $1\leq j \leq M-t$. We have $r^{j}_{\text{mle}}(x) \xrightarrow[]{} 0$, that is, $P( r^{j}_{\text{mle}}(x) >\epsilon) \xrightarrow[]{} 0$ for all $1\leq j \leq M-t$. Thus we can see that
   \vspace{-0.2cm}
   \begin{align*}
       P\left( r^{\text{unaff}}_{\text{mle}-h} > \epsilon \right)
       &=P\left( \max_{1\leq j \leq M-t} r^{j}_{\text{mle}}(x) > \epsilon \right) \\
       &= P \left( \bigcup_{j=1}^{M-t}\left\{ r^{j}_{\text{mle}}(x)  > \epsilon \right\} \right)\\
       &\leq \sum_{j=1}^{M-t} P \left( r^{j}_{\text{mle}}(x)  > \epsilon \right) \\
       &\xrightarrow[]{} 0.
   \end{align*}
   
   Again, the last convergence above is due to $M-t$ being fixed. Given that we have shown $r^{\text{unaff}}_{\text{mle}-h} \xrightarrow[]{} 0^+$, and $r^{\text{aff}}_{\text{mle}-l} \xrightarrow[]{} \beta(\alpha) - \alpha > 0$, this implies that $\frac{r^{\text{aff}}_{\text{mle}-l}}{r^{\text{unaff}}_{\text{mle}-h}} \rightarrow \infty$ as $n\xrightarrow[]{}\infty$.  Thus for any finite value of $\delta > 1$, and for $\eta = \frac{t}{M} > 0$, the probability that $\frac{r^{\text{aff}}_{\text{mle}-l}}{r^{\text{unaff}}_{\text{mle}-h}} > \frac{\delta} {\eta}$, and thus that the observed effect is $\frac{\delta} {\eta}-strong$, goes to 1 as $n\xrightarrow[]{} \infty$.
\end{proof}

\assymptTESS
\begin{proof}
We represent $S^T = v^{1T} \times \ldots \times v^{dT}$, where $v^{jT} \subseteq V^j$, and $S_0 = v_0^1 \times \ldots \times v_0^d$, where $v_0^j \subseteq V^j$. Assume without loss of generality that TESS optimizes over modes $\{1,2,\ldots,d\}$ in order to obtain $S_1 =  v_1^1 \times \ldots \times v_1^d$, where $v_1^j \subseteq V^j$, then optimizes over modes $\{1,2,\ldots,d\}$ in order again to obtain $S_2 =  v_2^1 \times \ldots \times v_2^d$, where $v_2^j \subseteq V^j$.  We will show that the following hold w.h.p.: (1) $S_1 \subseteq S^T$, i.e., $v_1^j \subseteq v^{jT} \: \forall j$; (2) $S_2 = S^T$, i.e., $v_2^j = v^{jT} \:\forall j$; and (3) $\hat S^\ast = S_2$.  

First, we consider the optimization over the first mode for $S_1$, starting from $S_0$ and thus finding the subset $S = v_1^1 \times v_0^2 \times \ldots \times v_0^d$ which maximizes $F(S)$ for fixed $v_0^2 \ldots v_0^d$. Consider the ``slices'' $x_m = \{v_m\} \times v_0^2 \times \ldots \times v_0^d$, where $v_m \in V^1$.  For $v_m \not\in v^{1T}$, we know that $x_m \cap S^T = \emptyset$, and thus $r_\text{mle}(x_m) \rightarrow 0$ as $n\rightarrow\infty$.  For $v_m \in v^{1T}$, $x_m \cap S^T$ includes some non-zero proportion $\rho_m$ of cells for which $r_\text{mle}(x) \rightarrow \beta(\alpha^\ast)-\alpha^\ast$ as $n\rightarrow\infty$, and thus $r_\text{mle}(x_m) \rightarrow \rho_m(\beta(\alpha^\ast)-\alpha^\ast)$ as $n\rightarrow\infty$.  This implies that the observed effect on the $x_m$ is $\frac{\delta}{\eta}$-strong w.h.p. for any $\delta > 1$ and thus for the specific $\delta$ from Theorem~\ref{thm:subset_strength}.  Hence $v_1^1 \subseteq v^{1T}$ w.h.p.  Identical logic can be used to show $v_1^j \subseteq v^{jT}$ w.h.p. for each $j$ from 2 to $d$ in turn, and thus $S_1 \subseteq S^T$ w.h.p.

Now we consider the optimization over the first mode for $S_2$, starting from $S_1$ and thus finding the subset $S = v_2^1 \times v_1^2 \times \ldots \times v_1^d$ which maximizes $F(S)$ for fixed $v_1^2 \ldots v_1^d$.  Consider the ``slices'' $x_m = \{v_m\} \times v_1^2 \times \ldots \times v_1^d$, where $v_m \in V^1$.  For $v_m \not\in v^{1T}$, we know that $x_m \cap S^T = \emptyset$, and thus $r_\text{mle}(x_m) \rightarrow 0$ as $n\rightarrow\infty$.  For $v_m \in v^{1T}$, since $v_1^j \subseteq v^{jT}$ for $j \in \{2,\ldots,d\}$, we know $x_m \subseteq S^T$, and thus $r_\text{mle}(x_m) \rightarrow \beta(\alpha^\ast)-\alpha^\ast$ as $n\rightarrow\infty$.  This implies that the observed effect on the $x_m$ is $\frac{\delta}{\eta}$-strong and $\nu$-homogeneous w.h.p. for any $\delta > 1$ and $\nu > 1$, and thus for the specific $\delta$ and $\nu$ from Theorems~\ref{thm:subset_homo} and~\ref{thm:subset_strength}.  Hence $v_2^1 = v^{1T}$ w.h.p. Identical logic can be used to show $v_2^j = v^{jT}$ w.h.p. for each $j$ from 2 to $d$ in turn, and thus $S_2 = S^T$ w.h.p.  

Finally, identical logic can be used to show w.h.p. that, for any mode $j$, $S_2$ is the subset $S= v^j \times v_2^{-j}$ which maximizes $F(S)$ for fixed $v_2^{-j}$. Thus no optimization over any mode can further increase $F(S)$, and $\hat S^\ast = S_2 = S^T$ w.h.p.
\end{proof}
\end{appendices}
\end{document}